\documentclass[a4paper,10pt,twoside,openright]{report}

\usepackage{amsmath}
\usepackage{amssymb}
\usepackage{amsthm}
\usepackage{verbatim}
\usepackage{stmaryrd}
\usepackage{mathtools}
\usepackage{thmtools}
\usepackage{url}
\usepackage[a4paper,twoside]{geometry}
\usepackage{titlesec}
\usepackage[svgnames]{xcolor}
\long\def\beginpgfgraphicnamed#1#2\endpgfgraphicnamed{\includegraphics{#1}}
\usepackage[absolute,overlay,quiet]{textpos}
\usepackage{rotating}
\usepackage[colorlinks]{hyperref}
\usepackage{enumitem}
\usepackage{subfig}
\usepackage{caption}
\usepackage{bold-extra}

\renewcommand{\phi}{\varphi}

\newcommand{\weight}[1]{\lvert #1 \rvert}

\newcommand{\FEL}{\ensuremath{\textup{FEL}}}
\newcommand{\FFEL}{\ensuremath{\textup{FFEL}}}

\newcommand{\EqFFEL}{\ensuremath{\textup{EqFFEL}}}
\newcommand{\FE}{\ensuremath{\textsc{fe}\hspace{0.1em}}}
\newcommand{\FT}{\ensuremath{\textup{FT}}}
\newcommand{\FNF}{\ensuremath{\textup{FNF}}}
\newcommand{\nf}{\ensuremath{f}}
\newcommand{\nftitle}{\textit{f}}
\newcommand{\inv}{\ensuremath{g}}
\newcommand{\invtitle}{\textit{g}}
\newcommand{\SCL}{\ensuremath{\textup{SCL}}}
\newcommand{\FSCL}{\ensuremath{\textup{FSCL}}}
\newcommand{\FSCLT}{\ensuremath{\FSCL_\SE}}
\newcommand{\EqFSCL}{\ensuremath{\textup{EqFSCL}}}
\newcommand{\SE}{\ensuremath{\textsc{se}\hspace{0.1em}}}
\newcommand{\ST}{\ensuremath{\textup{ST}}}
\newcommand{\SNF}{\ensuremath{\textup{SNF}}}
\newcommand{\nfs}{\ensuremath{f}}
\newcommand{\nfstitle}{\textit{f}}
\newcommand{\invs}{\ensuremath{g}}
\newcommand{\invstitle}{\textit{g}}
\newcommand{\sub}[2]{\ensuremath{[#1 \mapsto #2]}}
\newcommand{\ssub}[4]{\ensuremath{[#1 \mapsto #2, #3 \mapsto #4]}}
\newcommand{\tlef}{\ensuremath{\trianglelefteq}}
\newcommand{\trig}{\ensuremath{\trianglerighteq}}
\newcommand{\T}{\ensuremath{\mathcal{T}}}
\newcommand{\Tone}{\ensuremath{\T_\Box}}
\newcommand{\Ttwo}{\ensuremath{\T_{1, 2}}}
\newcommand{\E}{\ensuremath{\textsc{ce}\hspace{0.1em}}}
\newcommand{\cd}{\ensuremath{\textup{cd}}}
\newcommand{\dd}{\ensuremath{\textup{dd}}}
\newcommand{\tsd}{\ensuremath{\textup{tsd}}}

\newcommand{\trans}{\ensuremath{h}}
\newcommand{\CP}{\ensuremath{\textup{CP}}}
\newcommand{\CPT}{\ensuremath{\textup{CT}}}
\newcommand{\lef}{\ensuremath{\triangleleft}}
\newcommand{\rig}{\ensuremath{\triangleright}}

\newcommand{\true}{\ensuremath{\textup{\textsf{T}}}}
\newcommand{\false}{\ensuremath{\textup{\textsf{F}}}}
\newcommand{\leftand}{~
     \mathbin{\setlength{\unitlength}{1ex}
     \begin{picture}(1.4,1.8)(-.3,0)
     \put(-.6,0){$\wedge$}
     \put(-.54,-0.2){\textcolor{black}{\circle*{0.6}}}
     \put(-.54,-0.2){\circle{0.6}}
     \end{picture}
     }}
\newcommand{\leftor}{~
     \mathbin{\setlength{\unitlength}{1ex}
     \begin{picture}(1.4,1.8)(-.3,0)
     \put(-.6,0){$\vee$}
     \put(-.54,1.54){\textcolor{black}{\circle*{0.6}}}
     \put(-.54,1.54){\circle{0.6}}
     \end{picture}
     }}
\newcommand{\rightand}{~
     \mathbin{\setlength{\unitlength}{1ex}
     \begin{picture}(1.4,1.8)(-.3,0)
     \put(-.8,0){$\wedge$}
     \put(.72,-0.2){\textcolor{black}{\circle*{0.6}}}
     \put(.72,-0.2){\circle{0.6}}
     \end{picture}
     }}
\newcommand{\rightor}{~
     \mathbin{\setlength{\unitlength}{1ex}
     \begin{picture}(1.4,1.8)(-.3,0)
     \put(-.8,0){$\vee$}
     \put(.72,1.7){\textcolor{black}{\circle*{0.6}}}
     \put(.72,1.7){\circle{0.6}}
     \end{picture}
     }}
\newcommand{\sleftand}{~
     \mathbin{\setlength{\unitlength}{1ex}
     \begin{picture}(1.4,1.8)(-.3,0)
     \put(-.6,0){$\wedge$}
     \put(-.54,-0.2){\textcolor{white}{\circle*{0.6}}}
     \put(-.54,-0.2){\circle{0.6}}
     \end{picture}
     }}
\newcommand{\sleftor}{~
     \mathbin{\setlength{\unitlength}{1ex}
     \begin{picture}(1.4,1.8)(-.3,0)
     \put(-.6,0){$\vee$}
     \put(-.54,1.54){\textcolor{white}{\circle*{0.6}}}
     \put(-.54,1.54){\circle{0.6}}
     \end{picture}
     }}

\hypersetup{bookmarksopen=true}
\newtheorem{definition}{Definition}[chapter]
\newtheorem{lemma}[definition]{Lemma}
\newtheorem{theorem}[definition]{Theorem}

\titleformat{\chapter}[display]
  {\bfseries\raggedleft\titlerule\vspace{3ex}}
  {\Large \MakeUppercase{\chaptertitlename} \thechapter}
  {3ex}
  {\Huge}
  [\vspace{1pc}\titlerule]
\hypersetup{ 
bookmarksopen=true,
bookmarksnumbered=true,
citecolor=black,
filecolor=black, 
linkcolor=black, 
urlcolor=black 
}
\setlist{itemsep=-2pt}
\begin{document}

%
%
\newgeometry{margin=0cm}
\pdfbookmark[0]{Title page}{titlepage}
\begin{center}
  \null\vfill
  {\Large Completeness for Two Left-Sequential Logics \par}
  \vskip2cm
  {\large \textbf{MSc Thesis} (\textsl{Afstudeerscriptie})}
  \vskip3mm
  written by
  \vskip3mm
  \textbf{D.J.C. Staudt}\\
  (born January 12, 1987 in Amsterdam, The Netherlands)
  \vskip3mm
  under the supervision of \textbf{Dr. Alban Ponse}, and submitted to the
  Board of \\ Examiners in partial fulfillment of the requirements for the
  degree of
  \vskip5mm
  {\large\textbf{MSc in Logic}}
  \vskip5mm
  at the \textit{Universiteit van Amsterdam}.
  \vskip2cm
  \begin{tabular}{ll}
    \textbf{Date of the public defense:}
      & \textbf{Members of the Thesis Committee:} \\
    \textsl{May 31st, 2012}
      & Dr. Alexandru Baltag \\
      & Dr. Inge Bethke \\
      & Dr. Alban Ponse \\
      & Prof. Dr. Frank Veltman \\
  \end{tabular}
  \vfill
  \resizebox{10cm}{!}{\includegraphics{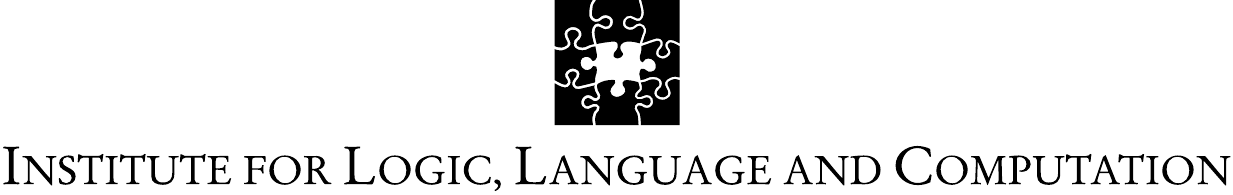}}
  \vskip5cm
\end{center}
\thispagestyle{empty}
\restoregeometry
\cleardoublepage

%
%
\begin{titlepage}
\pdfbookmark[0]{Abstract}{abstract}
\begin{abstract}
\setcounter{page}{3}
Left-sequential logics provide a means for reasoning about (closed)
propositional terms with atomic propositions that may have side effects and
that are evaluated sequentially from left to right. Such propositional terms
are commonly used in programming languages to direct the flow of a program. In
this thesis we explore two such left-sequential logics. First we discuss Fully
Evaluated Left-Sequential Logic, which employs a full evaluation strategy,
i.e., to evaluate a term every one of its atomic propositions is evaluated
causing its possible side effects to occur. We then turn to Short-Circuit
(Left-Sequential) Logic as presented in \cite{scl}, where the evaluation may be
`short-circuited', thus preventing some, if not all, of the atomic propositions
in a term being evaluated. We propose evaluation trees as a natural semantics
for both logics and provide axiomatizations for the least identifying variant
of each. From this, we define a logic with connectives that prescribe a full
evaluation strategy as well as connectives that prescribe a short-circuit
evaluation strategy.
\end{abstract}
\end{titlepage}
\cleardoublepage

%
%
\pdfbookmark[0]{Contents}{contents}
\tableofcontents
\cleardoublepage

%
%
\chapter{Introduction}
\label{chap:intro}
In computer programming it is common to use propositional terms to control the
flow of a program. These expressions occur for example in \verb'if' and
\verb'while' statements. Although at first sight it might appear as though
these expressions are terms of a Boolean algebra, it turns out that their
semantics are not governed by ordinary Propositional Logic (PL). The reason is
that many programming languages allow arbitrary instructions, e.g., method
calls, to occur as atomic propositions in such terms. Those instructions may
have side effects. Therefore the truth value of a term may depend on the state
of the execution environment, e.g., the operating system or the Java Virtual
Machine.  This state in turn can also be altered by the evaluation of (part
of) a term.  For example, in the term \verb'x && y', a side effect of the
evaluation of \verb'x' may be that any subsequent evaluation of \verb'y' yields
false. In that case \verb'x && y' will yield false, while \verb'y && x' may
yield true, i.e., conjunction is no longer commutative.


This shows that, unlike in PL, the evaluation strategy that is used impacts the
truth values of terms. Most programming languages evaluate such terms
sequentially from left to right. We refer to such an evaluation strategy as a
\emph{left-sequential evaluation strategy}. In addition, some programming
languages offer connectives that are evaluated using a \emph{short-circuit
(left-sequential) evaluation strategy}, such as \verb^&&^ and \verb^||^ in the
Java programming language, see, e.g., \cite{java}. A short-circuit evaluation
strategy is one that evaluates only as much of a propositional term as is
necessary to determine its truth value. For example, when evaluating the term
\verb'x && y', if \verb'x' evaluates to false, the entire term will be false,
regardless of the truth value of \verb'y'. In that case the evaluation is
`short-circuited' and \verb'y' is never evaluated. An evaluation strategy that
always evaluates terms in their entirety is called a \emph{full
(left-sequential) evaluation strategy}. In Java, for example, the connectives
\verb^&^ and \verb^|^ are evaluated using a full evaluation strategy. Some
languages provide both short-circuited and full versions of the binary
connectives, as Java does, thus allowing the programmer to write terms that
prescribe a mixed evaluation strategy.


In \cite{pa}, Bergstra and Ponse introduce Proposition Algebra as a means for
reasoning about sequential evaluations of propositional terms using a ternary
conditional connective, $y \lef x \rig z$, to be read as `if $x$ then $y$ else
$z$'. In \cite{scl}, they define \emph{Short-Circuit Logic} ($\SCL$) in terms
of Proposition Algebra using left-sequential versions of the standard logical
connectives. $\SCL$ formalizes equality between propositional terms that are
evaluated with a short-circuit evaluation strategy. They use $\neg$ for
negation, $\sleftor$ for (short-circuit) left-sequential disjunction and
$\sleftand$ for (short-circuit) left-sequential conjunction. The position of
the circle indicates the direction of the evaluation, i.e., from left to
right.  The negation symbol does not have a circle, because it has only one
possible evaluation strategy, i.e., evaluate the negated subterm and then
negate the result. Several variants of $\SCL$ are described in \cite{scl},
ranging from the least identifying, Free $\SCL$ ($\FSCL$), to the most
identifying, Static $\SCL$, which corresponds to PL. The only difference
between Static $\SCL$ and PL is that the connectives in Static $\SCL$ are
left-sequential and that the evaluation is short-circuited. Several semantics
have been given for $\SCL$, such as valuation congruences \cite{pa},
Hoare-McCarthy algebras \cite{hma} and truth tables \cite{tt}.


In \cite{sel} Blok first defined \emph{Fully Evaluated Left-Sequential Logic},
or \emph{Fully Evaluated Logic} ($\FEL$) for short. $\FEL$ is used for dealing
with terms that are to be evaluated using a full evaluation strategy. Blok
refers to this logic as Side-Effecting Logic, but we prefer the name $\FEL$ so
that we do not implicitly discount $\SCL$ as a logic that can be used for
reasoning about side effects. To allow for a mixed setting of $\FEL$ and
$\SCL$, we must distinguish the symbols used in $\FEL$-terms from those of
Bergstra and Ponse. We use $\leftand$ for full left-sequential conjunction and
$\leftor$ for full left-sequential disjunction. We still use $\neg$ for
negation, because it is evaluated with the same strategy as in $\SCL$. We can
now view the open circles in $\sleftand$ and $\sleftor$ as indicating
short-circuiting while the closed circles of $\leftand$ and $\leftor$ indicate
full evaluations. No variants of $\FEL$ other than Free $\FEL$ ($\FFEL$) have
yet been formally defined.


In this thesis we will also define a logic for reasoning about propositional
terms that contain both short-circuit left-sequential connectives and full
left-sequential connectives. We refer to a logic that offers both types of
connectives as a \emph{general left-sequential logic}.


The main differences between these left-sequential logics and PL is that they
employ a left-sequential evaluation strategy and that their atoms may have side
effects. We note that logics employing right-sequential evaluation strategies
can easily be expressed in terms of their left-sequential counterparts. We
study the left-sequential versions because most programming languages are
oriented left-to-right, mainly due to having been developed in the Western
world. Although side effects are well understood in programming, see e.g.,
\cite{black} or \cite{norrish}, they are often explained without a general
formal definition. In Chapter \ref{chap:concl} we will substantiate our claim
that these logics can be used to formally reason about propositional terms
whose atoms may have side effects. We note that both $\SCL$ and $\FEL$ are
sublogics of PL, in the sense that they identify fewer propositions, i.e.,
closed terms, although both have extreme variants that are equivalent with PL.


We start in Chapter \ref{chap:fel} by formally introducing $\FFEL$ and the set
of equations $\EqFFEL$. We prove that $\EqFFEL$ is an axiomatization of
$\FFEL$. In Chapter \ref{chap:scl} we introduce $\FSCLT$ as an alternative
semantics for the Proposition Algebra semantics of $\FSCL$. We also discuss the
set of equations $\EqFSCL$ and prove that it axiomatizes $\FSCLT$. In Section
\ref{sec:fsclpa} we prove that it also axiomatizes $\FSCL$. Chapters
\ref{chap:fel} and \ref{chap:scl} are written to be self-contained, hence there
is some duplication of definitions and narrative. In Chapter \ref{chap:rel} we
investigate the relations between $\FFEL$ and $\FSCL$. We show in Section
\ref{sec:ffelpa} that $\FFEL$ can also be expressed in terms of Proposition
Algebra. In Section \ref{sec:fsclpa} we prove that $\FSCLT$ is equivalent to
$\FSCL$. In Section \ref{sec:ffelfscl} we show that $\FFEL$ is a sublogic of
$\FSCL$ and use this fact to define a general left-sequential logic. We
conclude with some final remarks and provide an outlook for further study in
Chapter \ref{chap:concl}.


\chapter{Free Fully Evaluated Logic (FFEL)}
\label{chap:fel}
In this chapter we define Free Fully Evaluated Left-Sequential Logic, or Free
Fully Evaluated Logic ($\FFEL$) for short, and the set of equations $\EqFFEL$,
which we will prove axiomatizes $\FFEL$ in Section \ref{sec:felcpl}. We start
by defining $\FEL$-terms, which are built up from atomic propositions, referred
to as atoms, the truth value constants $\true$ for true and $\false$ for false
and the connectives $\neg$ for negation, $\leftand$ for full left-sequential
conjunction and $\leftor$ for full left-sequential disjunction.
\begin{definition}
Let $A$ be a countable set of atoms. \textbf{$\FEL$-terms $(\FT)$} have the
following grammar presented in Backus-Naur Form.
\begin{equation*}
P \in \FT ::= a \in A ~\mid~ \true ~\mid~ \false ~\mid~ \neg P
  ~\mid~ (P \leftand P) ~\mid~ (P \leftor P)
\end{equation*}
\end{definition}
If $A = \varnothing$ then the resulting logic is trivial.

Let us return for a moment to our motivation for left-sequential logics, i.e.,
propositional terms as used in programming languages. We will consider the
$\FEL$-term $a \leftor b$ and informally describe its evaluation, naturally
using a full evaluation strategy. We start by evaluating $a$ and let its
yield determine our next action. If $a$ yielded $\false$ we proceed by
evaluating $b$, i.e., the yield of the term as a whole will be the yield of
$b$. If $a$ yielded $\true$, we already know at this point that $a \leftor b$
will yield $\true$. We still evaluate $b$ though, but ignore its yield and
instead have the term yield $\true$. Evaluating a subterm even though its yield
is not needed to determine the yield of the term as a whole is the quintessence
of a full evaluation strategy.

Considering the more complex term $(a \leftor b) \leftand c$, we find that we
start by evaluating $a \leftor b$ and if it yielded $\true$ we proceed by
evaluating $c$. If it yielded $\false$ we still evaluate $c$, even though we
know that the term as a whole will now yield $\false$. The discussion of the
evaluations of these terms may have evoked images of trees in the mind of the
reader. We will indeed define equality of $\FEL$-terms using (evaluation)
trees. We define the set $\T$ of binary trees over $A$ with leaves in $\{\true,
\false\}$ recursively. We have that
\begin{equation*}
\true \in \T, \qquad
\false \in \T, \qquad\textrm{and}\qquad
(X \tlef a \trig Y) \in \T \textrm{ for any $X, Y \in \T$ and $a \in A$}.
\end{equation*}
In the expression $X \tlef a \trig Y$ the root is represented by $a$, the left
branch by $X$ and the right branch by $Y$. As is common, the depth of a tree
$X$ is defined recursively by $d(\true) = d(\false) = 0$ and for all $a \in A$,
$d(Y \tlef a \trig Z) = 1 + \max(d(Y), d(Z))$. Our reason for choosing this
particular notation for trees, out of the many that exist, is explained in
Chapter~\ref{chap:rel}. We shall refer to trees in $\T$ as evaluation trees, or
simply trees for short. Figure \ref{fig:exfe} shows the trees corresponding to
the evaluations of $(a \leftor b) \leftand c$ and $(a \leftand b) \leftor c$.

Returning to our example, we have seen that the tree corresponding to the
evaluation of $(a \leftor b) \leftand c$ can be composed from the tree
corresponding to the evaluation of $a \leftor b$ and that corresponding to the
evaluation of $c$. We said above that if $a \leftor b$ yielded $\true$, we
would proceed with the evaluation of $c$. This can be seen as replacing each
$\true$-leaf in the tree corresponding to the evaluation of $a \leftor b$ with
the tree that corresponds to the evaluation of $c$. Formally we define the leaf
replacement operator, `replacement' for short, on trees in $\T$ as follows. Let
$X, X', X'', Y, Z \in \T$ and $a \in A$. The replacement of $\true$ with $Y$
and $\false$ with $Z$ in $X$, denoted $X\ssub{\true}{Y}{\false}{Z}$, is defined
recursively as
\begin{align*}
\true\ssub{\true}{Y}{\false}{Z} &= Y \\
\false\ssub{\true}{Y}{\false}{Z} &= Z \\
(X' \tlef a \trig X'')\ssub{\true}{Y}{\false}{Z} &=
  X'\ssub{\true}{Y}{\false}{Z} \tlef a \trig
  X''\ssub{\true}{Y}{\false}{Z}.
\end{align*}
We note that the order in which the replacements of the leaves of $X$ is listed
inside the brackets is irrelevant. We will adopt the convention of not listing
any identities inside the brackets, i.e.,
\begin{equation*}
X\sub{\false}{Y} = X\ssub{\true}{\true}{\false}{Y}.
\end{equation*}
Furthermore we let replacements associate to the left. We also use that fact
that 
\begin{equation*}
X\sub{\true}{Y}\sub{\false}{Z} = X\ssub{\true}{Y}{\false}{Z}
\end{equation*}
if $Y$ does not contain $\false$, which can be shown by a trivial induction.
Similarly,
\begin{equation*}
X\sub{\false}{Z}\sub{\true}{Y} = X\ssub{\true}{Y}{\false}{Z}
\end{equation*}
if $Z$ does not contain $\true$. We now have the terminology and notation to
formally define the evaluation of $\FEL$-terms.

\begin{definition}
Let $A$ be a countable set of atoms and let $\T$ be the set of all finite
binary trees over $A$ with leaves in $\{\true, \false\}$. We define the unary
\textbf{Full Evaluation} function $\FE: \FT \to \T$ as:
\begin{align*}
\FE(\true) &= \true \\
\FE(\false) &= \false \\
\FE(a) &= \true \tlef a \trig \false &\textrm{for $a \in A$} \\
\FE(\neg P) &= \FE(P)\ssub{\true}{\false}{\false}{\true} \\
\FE(P \leftand Q) &=
  \FE(P)\ssub{\true}{\FE(Q)}{\false}{\FE(Q) \sub{\true}{\false}} \\
\FE(P \leftor Q) &=
  \FE(P)\ssub{\true}{\FE(Q)\sub{\false}{\true}}{\false}{\FE(Q)}.
\end{align*}
\end{definition}
Note that because we require $A$ to be a set, $\T$ is also a set. By a trivial
induction we can show that all trees in the image of $\FE$ are perfect binary
trees, i.e., all their paths are of equal length. As we can see from the
definition on atoms, the evaluation continues in the left branch if an atom
yields $\true$ and in the right branch if it yields $\false$.  Revisiting our
example once more, we indeed see how the evaluation of $a \leftor b$ is
composed of the evaluation of $a$ followed by the evaluation of $b$ in case $a$
yields $\false$ and by the evaluation of $b$, but with a constant yield of
$\true$, in case $a$ yields $\true$. We can compute $\FE(a \leftor b)$ as
follows.
\begin{align*}
\FE(a \leftor b) &= (\true \tlef a \trig \false)\ssub{\true}{(\true \tlef b
  \trig \false)\sub{\false}{\true}}{\false}{(\true \tlef b \trig \false)} \\
&= (\true \tlef a \trig \false)\ssub{\true}{(\true \tlef b \trig
  \true)}{\false}{(\true \tlef b \trig \false)} \\
&= (\true \tlef b \trig \true) \tlef a \trig (\true \tlef b \trig \false)
\end{align*}
Now the evaluation of $(a \leftor b) \leftand c$ is a composition of this tree
and $\true \tlef c \trig \false$, as can be seen in Figure \ref{fig:exfe2}.

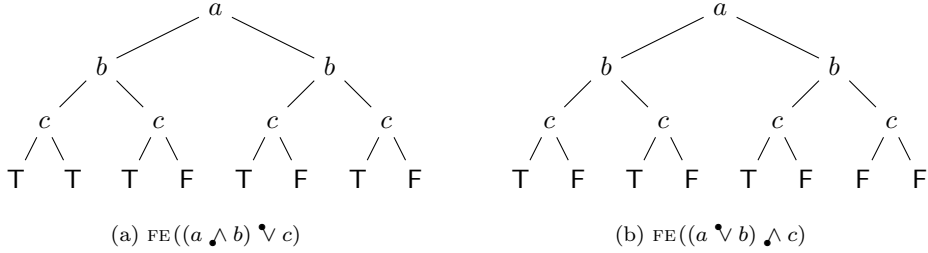
\begin{figure}[thbp]
\hrule
\centering
\subfloat[$\FE((a \protect\leftand b) \protect\leftor c)$]{\label{fig:exfe1}
\beginpgfgraphicnamed{fel1}
\begin{tikzpicture}[%
level distance=7.5mm,
level 1/.style={sibling distance=30mm},
level 2/.style={sibling distance=15mm},
level 3/.style={sibling distance=7.5mm}
]
\node (a) {$a$}
  child {node (b1) {$b$}
    child {node (c1) {$c$}
      child {node (d1) {$\true$}} 
      child {node (d2) {$\true$}}
    }
    child {node (c2) {$c$}
      child {node (d3) {$\true$}} 
      child {node (d4) {$\false$}}
    }
  }
  child {node (b2) {$b$}
    child {node (c3) {$c$}
      child {node (d5) {$\true$}} 
      child {node (d6) {$\false$}}
    }
    child {node (c4) {$c$}
      child {node (d7) {$\true$}} 
      child {node (d8) {$\false$}}
    }
  };
\end{tikzpicture}\endpgfgraphicnamed}
\qquad
\subfloat[$\FE((a \protect\leftor b) \protect\leftand c)$]{\label{fig:exfe2}
\beginpgfgraphicnamed{fel2}
\begin{tikzpicture}[%
level distance=7.5mm,
level 1/.style={sibling distance=30mm},
level 2/.style={sibling distance=15mm},
level 3/.style={sibling distance=7.5mm}
]
\node (a) {$a$}
  child {node (b1) {$b$}
    child {node (c1) {$c$}
      child {node (d1) {$\true$}} 
      child {node (d2) {$\false$}}
    }
    child {node (c2) {$c$}
      child {node (d3) {$\true$}} 
      child {node (d4) {$\false$}}
    }
  }
  child {node (b2) {$b$}
    child {node (c3) {$c$}
      child {node (d5) {$\true$}} 
      child {node (d6) {$\false$}}
    }
    child {node (c4) {$c$}
      child {node (d7) {$\false$}} 
      child {node (d8) {$\false$}}
    }
  };
\end{tikzpicture}\endpgfgraphicnamed}
\vspace{1em}
\hrule
\vspace{1em}
\caption{Trees depicting the evaluation of two $\FEL$-terms. The evaluation
starts at the root. When (the atom at) an inner node yields $\true$ the
evaluation continues in its left branch and when it yields $\false$ it
continues in its right branch. The leaves indicate the yield of the terms as a
whole.}
\label{fig:exfe}
\end{figure}

Informally we see that two $\FEL$-terms are equal when they not only yield the
same truth value given the truth values of their constituent atoms, but also
contain the same atoms \emph{in the same order}. Consider for example the terms
$a$ and $a \leftand (b \leftor \true)$ and note that the truth value of both
is determined entirely by the truth value of $a$. Also note that since $b$
occurs after $a$ in the second term, no side effect of $b$ could ever affect
$a$. When both terms would be placed in the context of another term, e.g., $a
\leftand c$ and $(a \leftand (b \leftor \true)) \leftand c$, the situation
changes. A side effect of $b$ might, for example, be that $c$ will yield
$\false$. In that case the truth value of the first term is determined by $a$
and $c$, while that of the second term is always $\false$. We are now ready to
define Fully Evaluated Left-Sequential Logic.
\begin{definition}
A \textbf{Fully Evaluated Left-Sequential Logic $(\FEL)$} is a logic that
satisfies the consequences of $\FE$-equality. \textbf{Free Fully Evaluated
Left-Sequential Logic $(\FFEL)$} is the fully evaluated left-sequential logic
that satisfies no more consequences than those of $\FE$-equality, i.e., for all
$P, Q \in \FT$,
\begin{equation*}
\FEL \vDash P = Q \ \Longleftarrow\  \FE(P) = \FE(Q)
\quad\textrm{and}\quad
\FFEL \vDash P = Q \ \Longleftrightarrow\  \FE(P) = \FE(Q).
\end{equation*}
\end{definition}
It is not considered standard to define a logic equationally, but in this case
we feel it is warranted to avoid having to mix the connectives from PL with
those from $\FEL$.


There is an immediate correspondence between trees and sets of traces, namely
the paths of such trees annotated with the truth value of each internal node.
For example, Figure \ref{fig:exfe1} would correspond to the set of traces
\begin{align*}
\{&(a\true b\true c\true, \true),
(a\true b\true c\false, \true),
(a\true b\false c\true, \true),
(a\true b\false c\false, \false), \\
&(a\false b\true c\true, \true),
(a\false b\true c\false, \false),
(a\false b\false c\true, \true),
(a\false b\false c\false, \false)\}.
\end{align*}
This means we could have defined the image of $\FE$ to be sets of such
annotated traces. We chose to define $\FEL$ with tree semantics rather than
with trace semantics because the former affords us a more succinct notation.


We now turn to the set of equations $\EqFFEL$, listed in Table
\ref{tab:eqffel}, which we will show in Section \ref{sec:felcpl} is an
axiomatization of $\FFEL$. This set of equations was first presented by Blok in
\cite{sel}.
\begin{table}[htpb]
\hrule
\begin{align}
\false &= \neg \true
  \label{ax:ft1}\tag{FEL1} \\
x \leftor y &= \neg(\neg x \leftand \neg y)
  \label{ax:ft2}\tag{FEL2} \\
\neg \neg x &= x
  \label{ax:ft3}\tag{FEL3} \\
(x \leftand y) \leftand z &= x \leftand (y \leftand z)
  \label{ax:ft7}\tag{FEL4} \\
\true \leftand x &= x
  \label{ax:ft4}\tag{FEL5} \\
x \leftand \true &= x
  \label{ax:ft5}\tag{FEL6} \\
x \leftand \false &= \false \leftand x
  \label{ax:ft6}\tag{FEL7} \\
x \leftand \false &= \neg x \leftand \false
  \label{ax:ft8}\tag{FEL8} \\
(x \leftand \false) \leftor y &= (x \leftor \true) \leftand y
  \label{ax:ft9}\tag{FEL9} \\
x \leftor (y \leftand \false) &= x \leftand (y \leftor \true)
  \label{ax:ft10}\tag{FEL10}
\end{align}
\hrule
\vspace{1em}
\caption{The set of equations \textbf{\EqFFEL}.}
\label{tab:eqffel}
\end{table}
\clearpage
If two $\FEL$-terms $s$ and $t$, possibly containing variables, are
derivable by equational logic and $\EqFFEL$, we denote this by $\EqFFEL \vdash
s = t$ and say that $s$ and $t$ are derivably equal. By virtue of
\eqref{ax:ft1} through \eqref{ax:ft3}, $\leftand$ is the dual of $\leftor$
and hence the duals of the equations in $\EqFFEL$ are also derivable. We will
use this fact implicitly throughout our proofs.


The following lemma shows some useful equations illustrating the special
properties of terms of the form $x \leftand \false$ and $x \leftor \true$. The
first is an `extension' of \eqref{ax:ft8} and the others show two different
ways how terms of the form $x \leftor \true$, and by duality terms of the form
$x \leftand \false$, can change the main connective of a term.
\begin{lemma}
\label{lem:feqs}
The following equations can all be derived by equational logic and $\EqFFEL$.
\begin{enumerate}[itemsep=5pt]
\item $x \leftand (y \leftand \false) = \neg x \leftand (y \leftand \false)$
  \label{eq:a1}
\item $(x \leftor \true) \leftand y = \neg(x \leftor \true) \leftor y$
  \label{eq:a5}
\item $x \leftor (y \leftand (z \leftor \true)) = (x \leftor y) \leftand
  (z \leftor \true)$
  \label{eq:a2}
\end{enumerate}
\end{lemma}
\begin{proof}
We derive the equations in order.
\begin{align*}
x \leftand (y \leftand \false)
&= (x \leftand \false) \leftand y
  &\textrm{by \eqref{ax:ft6} and \eqref{ax:ft7}} \\
&= (\neg x \leftand \false) \leftand y
  &\textrm{by \eqref{ax:ft8}} \\
&= \neg x \leftand (y \leftand \false)
  &\textrm{by \eqref{ax:ft6} and \eqref{ax:ft7}} \\[5pt]
(x \leftor \true) \leftand y
&= (x \leftand \false) \leftor y
  &\textrm{by \eqref{ax:ft9}} \\
&= (\neg x \leftand \false) \leftor y
  &\textrm{by \eqref{ax:ft8}} \\
&= (\neg x \leftand \neg \true) \leftor y
  &\textrm{by \eqref{ax:ft1}} \\
&= \neg(x \leftor \true) \leftor y
  &\textrm{by \eqref{ax:ft3} and \eqref{ax:ft2}} \\[5pt]
x \leftor (y \leftand (z \leftor \true))
&= x \leftor (y \leftor (z \leftand \false))
  &\textrm{by \eqref{ax:ft10}} \\
&= (x \leftor y) \leftor (z \leftand \false)
  &\textrm{by the dual of \eqref{ax:ft7}} \\
&= (x \leftor y) \leftand (z \leftor \true)
  &\textrm{by \eqref{ax:ft10}} &\qedhere
\end{align*}
\end{proof}


\begin{theorem}
\label{thm:felsnd}
For all $P, Q \in \FT$, if $\EqFFEL \vdash P = Q$ then $\FFEL \vDash P = Q$.
\end{theorem}
\begin{proof}
It is immediately clear that identity, symmetry and transitivity are preserved.
For congruence we show only that for all $P, Q, R \in \FT$, $\FFEL \vDash P =
Q$ implies $\FFEL \vDash R \leftand P = R \leftand Q$. The other cases proceed
in a similar fashion. If $\FFEL \vDash P = Q$, then $\FE(P) = \FE(Q)$ and hence
$\FE(P)\sub{\true}{\false} = \FE(Q)\sub{\true}{\false}$, so
\begin{equation*}
\FE(R)\ssub{\true}{\FE(P)}{\false}{\FE(P)\sub{\true}{\false}} =
\FE(R)\ssub{\true}{\FE(Q)}{\false}{\FE(Q)\sub{\true}{\false}}.
\end{equation*}
Therefore by definition of $\FE$, $\FFEL \vDash R \leftand P = R \leftand Q$.

The validity of the equations in $\EqFFEL$ is also easily verified. As an
example we show this for \eqref{ax:ft8}. 
\begin{align*}
\FE(P \leftand \false) &=
\FE(P)\ssub{\true}{\false}{\false}{\false\sub{\true}{\false}}
  &\textrm{by definition} \\
&= \FE(P)\sub{\true}{\false}
  &\textrm{because $\false\sub{\true}{\false} = \false$} \\
&= \FE(P)\ssub{\true}{\false}{\false}{\true}\sub{\true}{\false}
  &\textrm{by induction} \\
&= \FE(\neg P \leftand \false),
\end{align*}
where the induction that proves the third equality is on the structure of
evaluation trees.
\end{proof}


%
%
%
\section{FEL Normal Form}
\label{sec:felnf}
To aid in our completeness proof we define a normal form for $\FEL$-terms. Due
to the possible presence of side effects, $\FFEL$ does not identify terms which
contain different atoms or the same atoms in a different order. Because of
this, common normal forms for PL are not normal forms for $\FEL$-terms. For
example, rewriting a term to Conjunctive Normal Form or Disjunctive Normal Form
may require duplicating some of the atoms in the term, thus yielding a term
that is not derivably equal to the original. We first present the grammar for
our normal form, before motivating it. The normal form we present here is an
adaptation of a normal form proposed by Blok in \cite{sel}.


\begin{definition}
\label{def:fnf}
A term $P \in \FT$ is said to be in \textbf{$\FEL$ Normal Form $(\FNF)$} if it
is generated by the following grammar.
\begin{align*}
P \in \FNF &::= P^\true ~\mid~ P^\false ~\mid~ P^\true \leftand P^* \\
P^* &::= P^c ~\mid~ P^d \\
P^c &::= P^\ell ~\mid~ P^* \leftand P^d \\
P^d &::= P^\ell ~\mid~ P^* \leftor P^c \\
P^\ell &::= a \leftand P^\true ~\mid~ \neg a \leftand P^\true \\
P^\true &::= \true ~\mid~ a \leftor P^\true \\
P^\false &::= \false ~\mid~ a \leftand P^\false,
\end{align*}
where $a \in A$. We refer to $P^*$-forms as $*$-terms, to $P^\ell$-forms as
$\ell$-terms, to $P^\true$-forms as $\true$-terms and to $P^\false$-forms as
$\false$-terms. A term of the form $P^\true \leftand P^*$ is referred to as a
$\true$-$*$-term.
\end{definition}

We immediately note that if it were not for the presence of $\true$ and
$\false$ we could define a much simpler normal form. In that case it would
suffice to `push in' or `push down' the negations, thus obtaining a Negation
Normal Form, as exists for PL. Naturally if our set $A$ of atoms is empty, the
truth value constants would be a normal form.

When considering the image of $\FE$ we note that some trees only have
$\true$-leaves, some only have $\false$-leaves and some have both
$\true$-leaves and $\false$-leaves. For any $\FEL$-term $P$, $\FE(P \leftor
\true)$ is a tree with only $\true$-leaves, as can easily be seen from the
definition of $\FE$. All terms $P$ such that $\FE(P)$ only has $\true$-leaves
are rewritten to $\true$-terms.  Similarly $\FE(P \leftand \false)$ is a tree
with only $\false$-leaves. All terms $P$ such that $\FE(P)$ only has
$\false$-leaves are rewritten to $\false$-terms. The simplest trees in the
image of $\FE$ that have both $\true$-leaves and $\false$-leaves are $\FE(a)$
for $a \in A$. Any (occurrence of an) atom that determines (in whole or in part)
the yield of a term, such as $a$ in this example, is referred to as a
determinative (occurrence of an) atom. This as opposed to a non-determinative
(occurrence of an) atom, such as the $a$ in $a \leftor \true$, which does not
determine (either in whole or in part) the yield of the term.  Note that a term
$P$ such that $\FE(P)$ contains both $\true$ and $\false$ must contain at least
one determinative atom.

Terms that contain at least one determinative atom will be rewritten to
$\true$-$*$-terms. In $\true$-$*$-terms we encode each determinative atom
together with the non-determinative atoms that occur between it and the next
determinative atom in the term (reading from left to right) as an $\ell$-term.
Observe that the first atom in an $\ell$-term is the (only) determinative atom
in that $\ell$-term and that determinative atoms only occur in $\ell$-terms.
Also observe that the yield of an $\ell$-term is the yield of its determinative
atom. This is intuitively convincing, because the remainder of the atoms in any
$\ell$-term are non-determinative and hence do not contribute to its yield. The
non-determinative atoms that may occur before the first determinative atom are
encoded as a $\true$-term. A $\true$-$*$-term is the conjunction of a
$\true$-term encoding such atoms and a $*$-term, which contains only
conjunctions and disjunctions of $\ell$-terms. We could also have encoded such
atoms as an $\false$-term and then taken the disjunction with a $*$-term to
obtain a term with the same semantics. We consider $\ell$-terms to be `basic'
in $*$-terms in the sense that they are the smallest grammatical unit that
influences the yield of the $*$-term.

In the following, $P^\true, P^\ell$, etc.~are used both to denote grammatical
categories and as variables for terms in those categories. The remainder of
this section is concerned with defining and proving correct the normalization
function $\nf: \FT \to \FNF$.  We will define $\nf$ recursively using the
functions
\begin{equation*}
\nf^n: \FNF \to \FNF \quad\text{and}\quad
\nf^c: \FNF \times \FNF \to \FNF.
\end{equation*}
The first of these will be used to rewrite negated $\FNF$-terms to $\FNF$-terms
and the second to rewrite the conjunction of two $\FNF$-terms to an
$\FNF$-term. By \eqref{ax:ft2} we have no need for a dedicated function that
rewrites the disjunction of two $\FNF$-terms to an $\FNF$-term.

We start by defining $\nf^n$. Analyzing the semantics of $\true$-terms and
$\false$-terms together with the definition of $\FE$ on negations, it becomes
clear that $\nf^n$ must turn $\true$-terms into $\false$-terms and vice versa.
We also remark that $\nf^n$ must preserve the left-associativity of the
$*$-terms in $\true$-$*$-terms, modulo the associativity within $\ell$-terms.
We define $\nf^n: \FNF \to \FNF$ as follows, using the auxiliary function
$\nf^n_1: P^* \to P^*$ to `push down' or `push in' the negation symbols when
negating a $\true$-$*$-term. We note that there is no ambiguity between the
different grammatical categories present in an $\FNF$-term, i.e., any
$\FNF$-term is in exactly one of the grammatical categories identified in
Definition \ref{def:fnf}.
\begin{align}
\nf^n(\true) &= \false
  \label{eq:nfn1} \\
\nf^n(a \leftor P^\true) &= a \leftand \nf^n(P^\true)
  \label{eq:nfn2} \\
\nf^n(\false) &= \true
  \label{eq:nfn3} \\
\nf^n(a \leftand P^\false) &= a \leftor \nf^n(P^\false).
  \label{eq:nfn4} \\
\nf^n(P^\true \leftand Q^*) &= P^\true \leftand \nf^n_1(Q^*)
  \label{eq:nfn5} \\ \displaybreak[0]
\nf^n_1(a \leftand P^\true) &= \neg a \leftand P^\true
  \label{eq:nfn6} \\
\nf^n_1(\neg a \leftand P^\true) &= a \leftand P^\true
  \label{eq:nfn7} \\
\nf^n_1(P^* \leftand Q^d) &= \nf^n_1(P^*) \leftor \nf^n_1(Q^d)
  \label{eq:nfn8} \\
\nf^n_1(P^* \leftor Q^c) &= \nf^n_1(P^*) \leftand \nf^n_1(Q^c)
  \label{eq:nfn9}
\end{align}

Now we turn to defining $\nf^c$. These definitions have a great deal of
inter-dependence so we first present the definition for $\nf^c$ when the first
argument is a $\true$-term. We see that the conjunction of a $\true$-term with
another term always yields a term of the same grammatical category as the
second conjunct.
\begin{align}
\nf^c(\true, P) &= P
  \label{eq:nfc1} \\
\nf^c(a \leftor P^\true, Q^\true) &= a \leftor \nf^c(P^\true, Q^\true)
  \label{eq:nfc2} \\
\nf^c(a \leftor P^\true, Q^\false) &= a \leftand \nf^c(P^\true, Q^\false)
  \label{eq:nfc3} \\
\nf^c(a \leftor P^\true, Q^\true \leftand R^*) &= \nf^c(a \leftor P^\true,
  Q^\true) \leftand R^*
  \label{eq:nfc4}
\end{align}

For defining $\nf^c$ where the first argument is an $\false$-term we make use
of \eqref{ax:ft6} when dealing with conjunctions of $\false$-terms with
$\true$-$*$-terms. The definition of $\nf^c$ for the arguments used in the
right hand side of \eqref{eq:nfc7} starts at \eqref{eq:nfc14}. We note that
despite the high level of inter-dependence in these definitions, this does not
create a circular definition. We also note that the conjunction of an
$\false$-term with another term is always itself an $\false$-term. 
\begin{align}
\nf^c(\false, P^\true) &= \nf^n(P^\true)
  \label{eq:nfc5} \\
\nf^c(\false, P^\false) &= P^\false
  \label{eq:nfc6} \\
\nf^c(\false, P^\true \leftand Q^*) &= \nf^c(P^\true \leftand Q^*, \false)
  \label{eq:nfc7} \\
\nf^c(a \leftand P^\false, Q) &= a \leftand \nf^c(P^\false, Q)
  \label{eq:nfc8}
\end{align}

The case where the first conjunct is a $\true$-$*$-term is the most
complicated. Therefore we first consider the case where the second conjunct is
a $\true$-term. In this case we must make the $\true$-term part of the last
(rightmost) $\ell$-term in the $\true$-$*$-term, so that the result will
again be a $\true$-$*$-term. For this `pushing in' of the second conjunct we
define an auxiliary function $\nf^c_1: P^* \times P^\true \to P^*$.
\begin{align}
\nf^c(P^\true \leftand Q^*, R^\true) &= P^\true \leftand \nf^c_1(Q^*, R^\true)
  \label{eq:nfc9} \\
\nf^c_1(a \leftand P^\true, Q^\true) &= a \leftand \nf^c(P^\true, Q^\true)
  \label{eq:nfc10} \\
\nf^c_1(\neg a \leftand P^\true, Q^\true) &= \neg a \leftand \nf^c(P^\true,
  Q^\true)
  \label{eq:nfc11} \\
\nf^c_1(P^* \leftand Q^d, R^\true) &= P^* \leftand \nf^c_1(Q^d, R^\true)
  \label{eq:nfc12} \\
\nf^c_1(P^* \leftor Q^c, R^\true) &= P^* \leftor \nf^c_1(Q^c, R^\true)
  \label{eq:nfc13}
\end{align}

When the second conjunct is an $\false$-term, the result will naturally be an
$\false$-term itself. So we need to convert the $\true$-$*$-term to an
$\false$-term. Using \eqref{ax:ft7} we reduce this problem to converting a
$*$-term to an $\false$-term, for which we use the auxiliary function
$\nf^c_2: P^* \times P^\false \to P^\false$.
\begin{align}
\nf^c(P^\true \leftand Q^*, R^\false) &= \nf^c(P^\true, \nf^c_2(Q^*, R^\false)) 
  \label{eq:nfc14} \\
\nf^c_2(a \leftand P^\true, R^\false) &= a \leftand \nf^c(P^\true, R^\false)
  \label{eq:nfc15} \\
\nf^c_2(\neg a \leftand P^\true, R^\false) &= a \leftand \nf^c(P^\true,
  R^\false)
  \label{eq:nfc16} \\
\nf^c_2(P^* \leftand Q^d, R^\false) &= \nf^c_2(P^*, \nf^c_2(Q^d, R^\false))
  \label{eq:nfc17} \\
\nf^c_2(P^* \leftor Q^c, R^\false) &= \nf^c_2(P^*, \nf^c_2(Q^c, R^\false))
  \label{eq:nfc18}
\end{align}

Finally we are left with conjunctions and disjunctions of two
$\true$-$*$-terms, thus completing the definition of $\nf^c$.  We use the
auxiliary function $\nf^c_3: P^* \times P^\true \leftand P^* \to P^*$ to ensure
that the result is a $\true$-$*$-term.
\begin{align}
\nf^c(P^\true \leftand Q^*, R^\true \leftand S^*) &= P^\true \leftand 
  \nf^c_3(Q^*, R^\true \leftand S^*)
  \label{eq:nfc19} \\
\nf^c_3(P^*, Q^\true \leftand R^\ell) &= \nf^c_1(P^*, Q^\true) \leftand R^\ell
  \label{eq:nfc20} \\
\nf^c_3(P^*, Q^\true \leftand (R^* \leftand S^d)) &= \nf^c_3(P^*, Q^\true
  \leftand R^*) \leftand S^d
  \label{eq:nfc21} \\
\nf^c_3(P^*, Q^\true \leftand (R^* \leftor S^c)) &= \nf^c_1(P^*, Q^\true)
  \leftand (R^* \leftor S^c)
  \label{eq:nfc22}
\end{align}

As promised, we now define the normalization function $\nf: \FT \to \FNF$
recursively, using $\nf^n$ and $\nf^c$, as follows.
\begin{align}
\nf(a) &= \true \leftand (a \leftand \true)
  \label{eq:nf1} \\
\nf(\true) &= \true
  \label{eq:nf2} \\
\nf(\false) &= \false
  \label{eq:nf3} \\
\nf(\neg P) &= \nf^n(\nf(P))
  \label{eq:nf4} \\
\nf(P \leftand Q) &= \nf^c(\nf(P), \nf(Q))
  \label{eq:nf5} \\
\nf(P \leftor Q) &= \nf^n(\nf^c(\nf^n(\nf(P)), \nf^n(\nf(Q))))
  \label{eq:nf6}
\end{align}

\begin{restatable}{theorem}{thmnfcorrect}
\label{thm:nf}
For any $P \in \FT$, $\nf(P)$ terminates, $\nf(P) \in \FNF$ and $\EqFFEL \vdash
\nf(P) = P$.
\end{restatable}

In Appendix \ref{sec:lslnf} we first prove a number of lemmas showing that the
definitions $\nf^n$ and $\nf^c$ are correct and use those to prove the theorem.
The reader might wonder why we have used a normalization function rather than a
term rewriting system to prove the correctness of $\FNF$. The main reason is
the author's lack of experience with term rewriting systems, although the fact
that using a function relieves us of the need to prove the confluence of the
induced rewriting system, thus simplifying the proof, is also a factor.

\section{Tree Structure}
\label{sec:feltrees}
In Section \ref{sec:felcpl} we will prove that $\EqFFEL$ axiomatizes $\FFEL$ by
showing that for $P \in \FNF$ we can invert $\FE(P)$. To do this we need to
prove several structural properties of the trees in the image of $\FE$. In the
definition of $\FE$ we can see how $\FE(P \leftand Q)$ is assembled from
$\FE(P)$ and $\FE(Q)$ and similarly for $\FE(P \leftor Q)$. To decompose
these trees we will introduce some notation. The trees in the image of $\FE$
are all finite binary trees over $A$ with leaves in $\{\true, \false\}$, i.e.,
$\FE[\FT] \subseteq \T$. We will now also consider the set $\Tone$ of binary
trees over $A$ with leaves in $\{\true, \false, \Box\}$, where the `$\Box$'
symbol is pronounced `box'. Similarly we consider $\Ttwo$, the set of binary
trees over $A$ with leaves in $\{\true, \false, \Box_1, \Box_2\}$. The $\Box$,
$\Box_1$ and $\Box_2$ will be used as placeholders when composing or
decomposing trees.  Replacement of the leaves of trees in $\Tone$ and $\Ttwo$
by trees (either in $\T$, $\Tone$ or $\Ttwo$) is defined analogous to
replacement for trees in $\T$, adopting the same notational conventions.

For example we have by definition of $\FE$ that $\FE(P \leftand Q)$ can be
decomposed as
\begin{equation*}
\FE(P)\ssub{\true}{\Box_1}{\false}{\Box_2}
\ssub{\Box_1}{\FE(Q)}{\Box_2}{\FE(Q)\sub{\true}{\false}},
\end{equation*}
where $\FE(P)\ssub{\true}{\Box_1}{\false}{\Box_2} \in \Ttwo$ and $\FE(Q)$ and
$\FE(Q)\sub{\true}{\false}$ are in $\T$.  We note that this only works because
the trees in the image of $\FE$, or more general, in $\T$, do not contain any
boxes. Similarly, as we discussed previously, $\FE(P \leftand \false) =
\FE(P)\sub{\true}{\false}$, which we can write as
$\FE(P)\sub{\true}{\Box}\sub{\Box}{\false}$. We start by analyzing the
$\FE$-image of $\ell$-terms.

\begin{lemma}[Structure of $\ell$-terms]
\label{lem:litstf}
There is no $\ell$-term $P$ such that $\FE(P)$ can be decomposed as
$X\sub{\Box}{Y}$ with $X \in \Tone$ and $Y \in \T$, where $X \neq \Box$, but
does contain $\Box$, and $Y$ contains occurrences of both $\true$ and
$\false$.
\end{lemma}
\begin{proof}
Let $P$ be some $\ell$-term. When we analyze the grammar of $P$ we find that
one branch from the root of $\FE(P)$ will only contain $\true$ and not $\false$
and the other branch vice versa.  Hence if $\FE(P) = X\sub{\Box}{Y}$ and $Y$
contains occurrences of both $\true$ and $\false$, then $Y$ must contain the
root and hence $X = \Box$.
\end{proof}

By definition a $*$-term contains at least one $\ell$-term and hence for any
$*$-term $P$, $\FE(P)$ contains both $\true$ and $\false$. The following lemma
provides the $\FE$-image of the rightmost $\ell$-term in a $*$-term to witness
this fact.

\begin{lemma}[Determinativeness]
\label{lem:pert}
For all $*$-terms $P$, $\FE(P)$ can be decomposed as $X\sub{\Box}{Y}$ with $X
\in \Tone$ and $Y \in \T$ such that $X$ contains $\Box$ and $Y = \FE(Q)$ for
some $\ell$-term $Q$. Note that $X$ may be $\Box$. We will refer to $Y$ as the
witness for this lemma for $P$.
\end{lemma}
\begin{proof}
By induction on the complexity of $*$-terms $P$ modulo the complexity of
$\ell$-terms. In the base case $P$ is an $\ell$-term and $\FE(P) =
\Box\sub{\Box}{\FE(P)}$ is the desired decomposition by Lemma \ref{lem:litstf}.
For the induction we have to consider both $\FE(P \leftand Q)$ and $\FE(P
\leftor Q)$.

We treat only the case for $\FE(P \leftand Q)$, the case for $\FE(P \leftor Q)$
is analogous. Let $X\sub{\Box}{Y}$ be the decomposition for $\FE(Q)$ which we
have by induction hypothesis. Since by definition of $\FE$ on $\leftand$ we
have
\begin{equation*}
\FE(P \leftand Q) = \FE(P)\ssub{\true}{\FE(Q)}{\false}{\FE(Q)
\sub{\true}{\false}},
\end{equation*}
we also have
\begin{align*}
\FE(P \leftand Q) &= \FE(P)\ssub{\true}{X\sub{\Box}{Y}}{\false}{\FE(Q)
  \sub{\true}{\false}} \\
&= \FE(P)\ssub{\true}{X}{\false}{\FE(Q) \sub{\true}{\false}}\sub{\Box}{Y},
\end{align*}
where the second equality is due to the fact that the only boxes in
\begin{equation*}
\FE(P)\ssub{\true}{X}{\false}{\FE(Q) \sub{\true}{\false}}
\end{equation*}
are those occurring in $X$. This gives our desired decomposition.
\end{proof}

The following lemma illustrates another structural property of trees in the
image of $*$-terms under $\FE$, namely that the left branch of any
determinative atom in such a tree is different from its right branch.

\begin{lemma}[Non-decomposition]
\label{lem:nondectf}
There is no $*$-term $P$ such that $\FE(P)$ can be decomposed as
$X\sub{\Box}{Y}$ with $X \in \Tone$ and $Y \in \T$, where $X \neq \Box$ and $X$
contains $\Box$, but not $\true$ or $\false$.
\end{lemma}
\begin{proof}
By induction on $P$ modulo the complexity of $\ell$-terms. The base case covers
$\ell$-terms and follows immediately from Lemma \ref{lem:pert} ($\FE(P)$
contains occurrences of both $\true$ and $\false$) and Lemma \ref{lem:litstf}
(no non-trivial decomposition exists that contains both). For the induction
we assume that the lemma holds for all $*$-terms with lesser complexity than $P
\leftand Q$ and $P \leftor Q$.

We start with the case for $\FE(P \leftand Q)$. Suppose for contradiction that
$\FE(P \leftand Q) = X\sub{\Box}{Y}$ with $X \neq \Box$ and $X$ not containing
any occurrences of $\true$ or $\false$. Let $R$ be a witness of Lemma
\ref{lem:pert} for $P$. Now note that $\FE(P \leftand Q)$ has a subtree
\begin{equation*}
R\ssub{\true}{\FE(Q)}{\false}{\FE(Q)\sub{\true}{\false}}.
\end{equation*}
Because $Y$ must contain both the occurrences of $\false$ in the one branch
from the root of this subtree as well as the occurrences of $\FE(Q)$ in the
other (because they contain $\true$ and $\false$), Lemma \ref{lem:litstf}
implies that $Y$ must (strictly) contain $\FE(Q)$ and
$\FE(Q)\sub{\true}{\false}$. Hence there is a $Z \in \T$ such that $\FE(P) =
X\sub{\Box}{Z}$, which violates the induction hypothesis. The case for $\FE(P
\leftor Q)$ proceeds analogously.
\end{proof}

We now arrive at two crucial definitions for our completeness proof. When
considering $*$-terms we already know that $\FE(P \leftand Q)$ can be
decomposed as
\begin{equation*}
\FE(P)\ssub{\true}{\Box_1}{\false}{\Box_2}
\ssub{\Box_1}{\FE(Q)}{\Box_2}{\FE(Q)\sub{\true}{\false}}.
\end{equation*}
Our goal now is to give a definition for a type of decomposition so that this
is the only such decomposition for $\FE(P \leftand Q)$. We also ensure that
$\FE(P \leftor Q)$ does not have a decomposition of that type, so that we can
distinguish $\FE(P \leftand Q)$ from $\FE(P \leftor Q)$.  Similarly, we define
another type of decomposition so that $\FE(P \leftor Q)$ can only be decomposed
as
\begin{equation*}
\FE(P)\ssub{\true}{\Box_1}{\false}{\Box_2}
\ssub{\Box_1}{\FE(Q)\sub{\false}{\true}}{\Box_2}{\FE(Q)}
\end{equation*}
and that $\FE(P \leftand Q)$ does not have a decomposition of that type.

\begin{definition}
The pair $(Y, Z) \in \Ttwo \times \T$ is a \textbf{candidate conjunction
decomposition (ccd)} of $X \in \T$, if 
\begin{itemize}
\item $X = Y\ssub{\Box_1}{Z}{\Box_2}{Z\sub{\true}{\false}}$,
\item $Y$ contains both $\Box_1$ and $\Box_2$,
\item $Y$ contains neither $\true$ nor $\false$, and
\item $Z$ contains both $\true$ and $\false$.
\end{itemize}
Similarly, $(Y, Z)$ is a \textbf{candidate disjunction decomposition (cdd)} of
$X$, if
\begin{itemize}
\item $X = Y\ssub{\Box_1}{Z\sub{\false}{\true}}{\Box_2}{Z}$,
\item $Y$ contains both $\Box_1$ and $\Box_2$,
\item $Y$ contains neither $\true$ nor $\false$, and
\item $Z$ contains both $\true$ and $\false$.
\end{itemize}
\end{definition}

The ccd and cdd are not necessarily the decompositions we are looking for,
because, for example, $\FE((P \leftand Q) \leftand R)$ has a ccd
$(\FE(P)\ssub{\true}{\Box_1}{\false}{\Box_2}, \FE(Q \leftand R))$, whereas the
decomposition we need is $(\FE(P \leftand
Q)\ssub{\true}{\Box_1}{\false}{\Box_2}, \FE(R))$. Therefore we refine these
definitions to obtain the decompositions we seek.

\begin{definition}
The pair $(Y, Z) \in \Ttwo \times \T$ is a \textbf{conjunction decomposition
(cd)} of $X \in \T$, if it is a ccd of $X$ and there is no other ccd $(Y', Z')$
of $X$ where the depth of $Z'$ is smaller than that of $Z$.  Similarly, $(Y,
Z)$ is a \textbf{disjunction decomposition (dd)} of $X$, if it is a cdd of $X$
and there is no other cdd $(Y', Z')$ of $X$ where the depth of $Z'$ is smaller
than that of $Z$.
\end{definition}

\begin{theorem}
\label{thm:cddd}
For any $*$-term $P \leftand Q$, i.e., with $P \in P^*$ and $Q \in P^d$, $\FE(P
\leftand Q)$ has the (unique) cd
\begin{equation*}
(\FE(P)\ssub{\true}{\Box_1}{\false}{\Box_2}, \FE(Q))
\end{equation*}
and no dd. For any $*$-term $P \leftor Q$, i.e., with $P \in P^*$ and $Q \in
P^c$, $\FE(P \leftor Q)$ has no cd and its (unique) dd is
\begin{equation*}
(\FE(P)\ssub{\true}{\Box_1}{\false}{\Box_2}, \FE(Q)).
\end{equation*}
\end{theorem}
\begin{proof}
We first treat the case for $P \leftand Q$ and start with cd. Note that $\FE(P
\leftand Q)$ has a ccd $(\FE(P)\ssub{\true}{\Box_1}{\false}{\Box_2}, \FE(Q))$
by definition of $\FE$ (for the first condition) and by Lemma~\ref{lem:pert}
(for the fourth condition). It is immediate that it satisfies the second and
third conditions. It also follows that for any ccd $(Y, Z)$ either $Z$ contains
or is contained in $\FE(Q)$, for suppose otherwise, then $Y$ will contain an
occurrence of $\true$ or of $\false$, namely those we know by
Lemma~\ref{lem:pert} that $\FE(Q)$ has. Therefore it suffices to show that
there is no ccd $(Y, Z)$ where $Z$ is strictly contained in $\FE(Q)$. Suppose
for contradiction that $(Y, Z)$ is such a ccd. If $Z$ is strictly contained in
$\FE(Q)$ we can decompose $\FE(Q)$ as $\FE(Q) = U\sub{\Box}{Z}$ for some $U \in
\Tone$ that contains but is not equal to $\Box$. By Lemma \ref{lem:nondectf}
this implies that $U$ contains $\true$ or $\false$. But then so does $Y$,
because
\begin{equation*}
Y = \FE(P)\ssub{\true}{U\sub{\Box}{\Box_1}}{\false}{U\sub{\Box}{\Box_2}},
\end{equation*}
and so $(Y, Z)$ is not a ccd for $\FE(P \leftand Q)$. Therefore
$(\FE(P)\ssub{\true}{\Box_1}{\false}{\Box_2}, \FE(Q))$ is the \emph{unique} cd
for $\FE(P \leftand Q)$.

Now for the dd. It suffices to show that there is no cdd for $\FE(P \leftand
Q)$. Suppose for contradiction that $(Y, Z)$ is a cdd for $\FE(P \leftand Q)$.
We note that $Z$ cannot be contained in $\FE(Q)$, for then by Lemma
\ref{lem:nondectf} $Y$ would contain $\true$ or $\false$. So $Z$ (strictly)
contains $\FE(Q)$. But then because
\begin{equation*}
Y\ssub{\Box_1}{Z\sub{\false}{\true}}{\Box_2}{Z} = \FE(P \leftand Q),
\end{equation*}
we would have by Lemma \ref{lem:pert} that $\FE(P \leftand Q)$ does not contain
an occurrence of $\FE(Q)\sub{\true}{\false}$. But the cd of $\FE(P \leftand Q)$
tells us that it does, contradiction! Therefore there is no cdd, and hence no
dd, for $\FE(P \leftand Q)$. The case for $\FE(P \leftor Q)$ proceeds
analogously.
\end{proof}

At this point we have the tools necessary to invert $\FE$ on $*$-terms, at
least down to the level of $\ell$-terms. We note that we can easily detect if a
tree in the image of $\FE$ is in the image of $P^\ell$, because all leaves to
the left of the root are one truth value, while all the leaves to the right are
the other. To invert $\FE$ on $\true$-$*$-terms we still need to be able to
reconstruct $\FE(P^\true)$ and $\FE(Q^*)$ from $\FE(P^\true \leftand Q^*)$. To
this end we define a $\true$-$*$-decomposition.

\begin{definition}
The pair $(Y, Z) \in \Tone \times \T$ is a \textbf{$\true$-$*$-decomposition
(tsd)} of $X \in \T$, if $X = Y\sub{\Box}{Z}$, $Y$ does not contain $\true$ or
$\false$ and there is no decomposition $(U, V) \in \Tone \times \T$ of $Z$ such
that
\begin{itemize}
\item $Z = U\sub{\Box}{V}$,
\item $U$ contains $\Box$,
\item $U \neq \Box$, and
\item $U$ contains neither $\true$ nor $\false$.
\end{itemize}
\end{definition}

\begin{theorem}
\label{thm:tsd}
For any $\true$-term $P$ and $*$-term $Q$ the (unique) tsd of $\FE(P \leftand
Q)$ is 
\begin{equation*}
(\FE(P)\sub{\true}{\Box}, \FE(Q)).
\end{equation*}
\end{theorem}
\begin{proof}
First we observe that $(\FE(P)\sub{\true}{\Box}, \FE(Q))$ is a tsd because by
definition of $\FE$ on $\leftand$ we have $\FE(P)\sub{\true}{\FE(Q)} = \FE(P
\leftand Q)$ and $\FE(Q)$ is non-decomposable by Lemma \ref{lem:nondectf}.

Suppose for contradiction that there is another tsd $(Y, Z)$ of $\FE(P
\leftand Q)$. Now $Z$ must contain or be contained in $\FE(Q)$ for otherwise
$Y$ would contain $\true$ or $\false$, i.e., the ones we know $\FE(Q)$ has by
Lemma \ref{lem:pert}.

If $Z$ is strictly contained in $\FE(Q)$, then $\FE(Q) = U\sub{\Box}{Z}$ for
some $U \in \Tone$ with $U \neq \Box$ and $U$ not containing $\true$ or
$\false$ (because then $Y$ would too).  But this violates Lemma
\ref{lem:nondectf}, which states that no such decomposition exists. If $Z$
strictly contains $\FE(Q)$, then $Z$ contains at least one atom from $P$.  But
the left branch of any atom in $\FE(P)$ is equal to its right branch and hence
$Z$ is decomposable. Therefore $(\FE(P)\sub{\true}{\Box}, \FE(Q))$ is the
\emph{unique} tsd of $\FE(P \leftand Q)$.
\end{proof}

\section{Completeness}
\label{sec:felcpl}
With the two theorems from the previous section, we can prove completeness for
$\FFEL$. We define three auxiliary functions to aid in our definition of the
inverse of $\FE$ on $\FNF$. Let $\cd : \T \to \Ttwo \times \T$ be the function
that returns the conjunction decomposition of its argument, $\dd$ of the same
type its disjunction decomposition and $\tsd: \T \to \Tone \times \T$ its
$\true$-$*$-decomposition. Naturally, these functions are undefined when their
argument does not have a decomposition of the specified type. Each of these
functions returns a pair and we will use $\cd_1$ ($\dd_1$, $\tsd_1$) to denote
the first element of this pair and $\cd_2$ ($\dd_2$, $\tsd_2$) to denote the
second element.

We define $\inv: \T \to \FT$ using the functions $\inv^\true: \T \to \FT$ for
inverting trees in the image of $\true$-terms and $\inv^\false$, $\inv^\ell$
and $\inv^*$ of the same type for inverting trees in the image of
$\false$-terms, $\ell$-terms and $*$-terms, respectively. These functions are
defined as follows.
\begin{align}
\inv^\true(X) &=
  \begin{cases}
    \true
      &\textrm{if $X = \true$} \\
    a \leftor \inv^\true(Y)
      &\textrm{if $X = Y \tlef a \trig Z$}
  \end{cases} \\
\intertext{We note that we might as well have used the right branch from the
root in the recursive case. We chose the left branch here to more closely
mirror the definition of the corresponding function for $\FSCLT$, defined in
Chapter \ref{chap:scl}.}
\inv^\false(X) &=
  \begin{cases}
    \false
      &\textrm{if $X = \false$} \\
    a \leftand \inv^\false(Z)
      &\textrm{if $X = Y \tlef a \trig Z$}
  \end{cases} \\
\intertext{Similarly, we could have taken the left branch in this case.}
\inv^\ell(X) &=
  \begin{cases}
    a \leftand \inv^\true(Y) 
      &\textrm{if $X = Y \tlef a \trig Z$ for some $a \in A$} \\
      &\textrm{and $Y$ only has $\true$-leaves} \\
    \neg a \leftand \inv^\true(Z)
      &\textrm{if $X = Y \tlef a \trig Z$ for some $a \in A$} \\
      &\textrm{and $Z$ only has $\true$-leaves}
  \end{cases} \\
\inv^*(X) &=
  \begin{cases}
    \inv^*(\cd_1(X)\ssub{\Box_1}{\true}{\Box_2}{\false}) \leftand
      \inv^*(\cd_2(X))
      &\textrm{if $X$ has a cd} \\
    \inv^*(\dd_1(X)\ssub{\Box_1}{\true}{\Box_2}{\false}) \leftor
      \inv^*(\dd_2(X))
      &\textrm{if $X$ has a dd} \\
    \inv^\ell(X)
      &\textrm{otherwise}
  \end{cases} \\
\intertext{We can immediately see how Theorem \ref{thm:cddd} will be used in
the correctness proof of $\inv^*$.}
\inv(X) &=
  \begin{cases}
    \inv^\true(X)
      &\textrm{if $X$ has only $\true$-leaves} \\
    \inv^\false(X)
      &\textrm{if $X$ has only $\false$-leaves} \\
    \inv^\true(\tsd_1(X)\sub{\Box}{\true}) \leftand \inv^*(\tsd_2(X))
      &\textrm{otherwise}
  \end{cases}
\end{align}
Similarly, we can see how Theorem \ref{thm:tsd} is used in the correctness
proof of $\inv$. It should come as no surprise that $\inv$ is indeed correct
and inverts $\FE$ on $\FNF$.

\begin{restatable}{theorem}{thminvcorrect}
\label{thm:felinv}
For all $P \in \FNF$, $\inv(\FE(P)) \equiv P$.
\end{restatable}

The proof for this theorem can be found in Appendix \ref{sec:felinv}. For the
sake of completeness, we separately state the completeness result below.

\begin{theorem}
\label{thm:felcpl}
For all $P, Q \in \FT$, if $\FFEL \vDash P = Q$ then $\EqFFEL \vdash
P = Q$.
\end{theorem}
\begin{proof}
It suffices to show that for $P,Q \in \FNF$, $\FE(P) = \FE(Q)$ implies $P
\equiv Q$. To see this suppose that $P'$ and $Q'$ are two $\FEL$-terms and
$\FE(P') = \FE(Q')$. We know that $P'$ is derivably equal to an $\FNF$-term
$P$, i.e., $\EqFFEL \vdash P' = P$, and that $Q'$ is derivably equal to an
$\FNF$-term $Q$, i.e., $\EqFFEL \vdash Q' = Q$. Theorem \ref{thm:felsnd} then
gives us $\FE(P') = \FE(P)$ and $\FE(Q') = \FE(Q)$. Hence by the result $P
\equiv Q$ and in particular $\EqFFEL \vdash P = Q$. Transitivity then gives
us $\EqFFEL \vdash P' = Q'$ as desired.

The result follows immediately from Theorem \ref{thm:felinv}.
\end{proof}

\chapter{Free Short-Circuit Logic (FSCL)}
\label{chap:scl}
In this chapter we define Free Short-Circuit Logic on evaluation trees and
present the set of equations $\EqFSCL$, which we will prove axiomatizes this
logic in Section \ref{sec:felcpl}. Formally, $\SCL$-terms are built up from
atomic propositions that may have side effects, called atoms, the truth value
constants $\true$ for true and $\false$ for false and the connectives $\neg$
for negation, $\sleftand$ for (short-circuit) left-sequential conjunction and
$\sleftor$ for (short-circuit) left-sequential disjunction.
\begin{definition}
Let $A$ be a countable set of atoms. \textbf{$\SCL$-terms $(\ST)$} have the
following grammar presented in Backus-Naur Form.
\begin{equation*}
P \in \ST ::= a \in A ~\mid~ \true ~\mid~ \false ~\mid~ \neg P
  ~\mid~ (P \sleftand P) ~\mid~ (P \sleftor P)
\end{equation*}
\end{definition}
As is the case with $\FEL$, if $A = \varnothing$ then resulting logic is
trivial.

First we return for a moment to our motivation for left-sequential logics,
i.e., propositional terms as used in programming languages. We will consider
the $\SCL$-term $a \sleftor b$ and informally describe its evaluation,
naturally using a short-circuit evaluation strategy. We start by evaluating $a$
and let its yield determine our next action. If $a$ yielded $\false$ we proceed
by evaluating $b$, i.e., the yield of the term as a whole will be the yield of
$b$. If $a$ yielded $\true$, we already know at this point that $a \sleftor b$
will yield $\true$.  We skip the evaluation of $b$ and let the term yield
$\true$, i.e., $b$ is short-circuited.

Considering the more complex term $(a \sleftor b) \sleftand c$, we find that we
start by evaluating $a \sleftor b$ and if it yields $\true$ we proceed by
evaluating $c$. If it yields $\false$ we skip the evaluation of $c$, because we
know the term will yield $\false$. This example shows that evaluating
$\SCL$-terms is an interactive procedure, where the yield of the previous atom
is needed to determine which atom to evaluate next. We believe these semantics
are best captured in trees. Hence we will define equality of $\SCL$-terms
using (evaluation) trees. We define the set $\T$ of finite binary trees over
$A$ with leaves in $\{\true, \false\}$ recursively. We have that
\begin{equation*}
\true \in \T, \qquad
\false \in \T, \qquad\textrm{and}\qquad
(X \tlef a \trig Y) \in \T \textrm{ for any $X, Y \in \T$ and $a \in A$}.
\end{equation*}
In the expression $X \tlef a \trig Y$ the root is represented by $a$, the left
branch by $X$ and the right branch by $Y$. We define the depth of a tree $X$
recursively by $d(\true) = d(\false) = 0$ and $d(Y \tlef a \trig Z) = 1 +
\max(d(Y), d(Z))$ for $a \in A$. The reason for our choice of notation for
trees will become apparent in Chapter \ref{chap:rel}. We refer to trees in $\T$
as evaluation trees, or trees for short.  Figure \ref{fig:exse} shows the trees
corresponding to the evaluations of $(a \sleftor b) \sleftand c$ and $(a
\sleftand b) \sleftor c$.

Returning to our example, we have seen that the tree corresponding to the
evaluation of $(a \sleftor b) \sleftand c$ can be composed from the tree
corresponding to the evaluation of $a \sleftor b$ and that corresponding to the
evaluation of $c$. We said above that if $a \sleftor b$ yielded $\true$, we
would proceed with the evaluation of $c$. This can be seen as replacing each
$\true$-leaf in the tree corresponding to the evaluation of $a \sleftor b$ with
the tree that corresponds to the evaluation of $c$. Formally we define the leaf
replacement operator, `replacement' for short, on trees in $\T$ as follows. Let
$X, X', X'' Y, Z \in \T$ and $a \in A$. The replacement of $\true$ with $Y$ and
$\false$ with $Z$ in $X$, denoted $X\ssub{\true}{Y}{\false}{Z}$, is defined
recursively as
\begin{align*}
\true\ssub{\true}{Y}{\false}{Z} &= Y \\
\false\ssub{\true}{Y}{\false}{Z} &= Z \\
(X' \tlef a \trig X'')\ssub{\true}{Y}{\false}{Z} &=
  X'\ssub{\true}{Y}{\false}{Z} \tlef a \trig
  X''\ssub{\true}{Y}{\false}{Z}.
\end{align*}
We note that the order in which the replacements of the leaves of $X$ is listed
inside the brackets is irrelevant. We will adopt the convention of not listing
any identities inside the brackets, i.e.,
\begin{equation*}
X\sub{\false}{Y} = X\ssub{\true}{\true}{\false}{Y}.
\end{equation*}
Furthermore we let replacements associate to the left. We also use that fact
that
\begin{equation*}
X\sub{\true}{Y}\sub{\false}{Z} = X\ssub{\true}{Y}{\false}{Z}
\end{equation*}
if $Y$ does not contain $\false$, which can be shown by a trivial induction.
Similarly,
\begin{equation*}
X\sub{\false}{Z}\sub{\true}{Y} = X\ssub{\true}{Y}{\false}{Z}
\end{equation*}
if $Z$ does not contain $\true$.  We now have the terminology and notation to
formally define the mapping from $\SCL$-terms to evaluation trees.

\begin{definition}
Let $A$ be a countable set of atoms and let $\T$ be the set all finite binary
trees over $A$ with leaves in $\{\true, \false\}$. We define the unary
\textbf{Short-Circuit Evaluation} function $\SE: \ST \to \T$ as:
\begin{align*}
\SE(\true) &= \true \\
\SE(\false) &= \false \\
\SE(a) &= \true \tlef a \trig \false &\textrm{for $a \in A$} \\
\SE(\neg P) &= \SE(P)\ssub{\true}{\false}{\false}{\true} \\
\SE(P \sleftand Q) &= \SE(P)\sub{\true}{\SE(Q)} \\
\SE(P \sleftor Q) &= \SE(P)\sub{\false}{\SE(Q)}.
\end{align*}
\end{definition}
As we can see from the definition on atoms, the evaluation continues in the
left branch if an atom yields $\true$ and in the right branch if it yields
$\false$. Revisiting our example once more, we indeed see how the evaluation of
$a \sleftor b$ is composed of the evaluation of $a$ followed by the evaluation
of $b$ in case $a$ yields $\false$. We can compute
\begin{align*}
\SE(a \sleftor b)
&= (\true \tlef a \trig \false)\sub{\false}{(\true \tlef b \trig \false)} \\
&= \true \tlef a \trig (\true \tlef b \trig \false).
\end{align*}
Now the evaluation of $(a \sleftor b) \sleftand c$ is a composition of this
tree and $\true \tlef c \trig \false$, as can be seen in Figure
\ref{fig:exfe2}.

\begin{figure}[thbp]
\hrule
\centering
\subfloat[$\SE((a \protect\sleftand b) \protect\sleftor c)$]{\label{fig:exse1}
\beginpgfgraphicnamed{scl1}
\begin{tikzpicture}[%
level distance=7.5mm,
level 1/.style={sibling distance=30mm},
level 2/.style={sibling distance=15mm},
level 3/.style={sibling distance=7.5mm}
]
\node (a) {$a$}
  child {node (b1) {$b$}
    child {node (c1) {$\true$}}
    child {node (c2) {$c$}
      child {node (d3) {$\true$}} 
      child {node (d4) {$\false$}}
    }
  }
  child {node (b2) {$c$}
    child {node (d5) {$\true$}} 
    child {node (d6) {$\false$}}
  };
\end{tikzpicture}\endpgfgraphicnamed}
\qquad
\subfloat[$\SE((a \protect\sleftor b) \protect\sleftand c)$]{\label{fig:exse2}
\beginpgfgraphicnamed{scl2}
\begin{tikzpicture}[%
level distance=7.5mm,
level 1/.style={sibling distance=30mm},
level 2/.style={sibling distance=15mm},
level 3/.style={sibling distance=7.5mm}
]
\node (a) {$a$}
  child {node (b1) {$c$}
    child {node (d1) {$\true$}} 
    child {node (d2) {$\false$}}
  }
  child {node (b2) {$b$}
    child {node (c3) {$c$}
      child {node (d5) {$\true$}} 
      child {node (d6) {$\false$}}
    }
    child {node (c4) {$\false$}}
  };
\end{tikzpicture}\endpgfgraphicnamed}
\vspace{1em}
\hrule
\vspace{1em}
\caption{Trees depicting the evaluation of two $\SCL$-terms. The evaluation
starts at the root. When (the atom at) an inner node yields $\true$ the
evaluation continues in its left branch and when it yields $\false$ it
continues in its right branch. The leaves indicate the yield of the terms as a
whole.}
\label{fig:exse}
\end{figure}
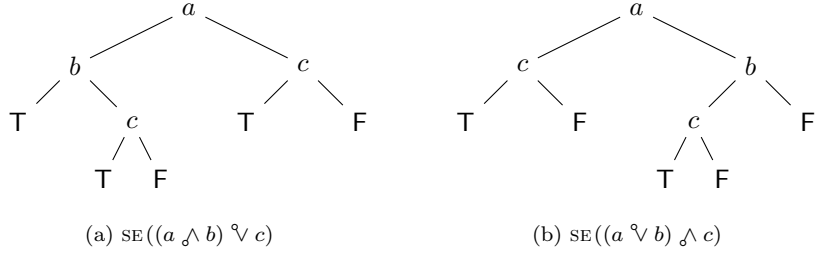

These trees show us that a function of the yield of the atoms in an $\SCL$-term
is insufficient to determine the semantics of the term as a whole. They show us
that we must also consider the (conditional) order in which the atoms occur in
the term. In particular we see that in $P \sleftor Q$, $Q$ will be
short-circuited if $P$ yields $\true$, while in $P \sleftand Q$, it will be
short-circuited if $P$ yields $\false$. We are now ready to define
Short-Circuit Logic on evaluation trees.
\begin{definition}
\textbf{Free Short-Circuit Logic $(\FSCLT)$} is the logic that satisfies
exactly the consequences of $\SE$-equality, i.e., for all $P, Q \in \ST$,
\begin{equation*}
\FSCLT \vDash P = Q \ \Longleftrightarrow\  \SE(P) = \SE(Q).
\end{equation*}
\end{definition}
Using the completeness result we shall prove in Section \ref{sec:sclcpl}, we
will show that $\FSCLT$ is in fact equivalent to $\FSCL$ as defined by Bergstra
and Ponse in \cite{scl}. This should come as no surprise given the tree-like
structure that Proposition Algebra terms exhibit, see, e.g., \cite{pa} or
\cite{pascl}.

We choose a representation of $\T$ as trees rather than as sets of traces,
i.e., the paths of those trees annotated with truth values for the atoms,
because the tree notation allows us to be more succinct. These tree semantics
were first given, although presented as trace semantics, by Ponse in
\cite{tt}.

We now turn to the set of equations $\EqFSCL$, listed in Table
\ref{tab:eqfscl}, which we will show in Section \ref{sec:sclcpl} is an
axiomatization of $\FSCLT$. This set of equations is based on one presented by
Bergstra and Ponse in \cite{scl}. If two $\SCL$-terms $s$ and $t$, where we
extend the definition to allow for terms containing variables, are derivable by
equational logic and $\EqFSCL$, we denote this by $\EqFSCL \vdash s = t$ and
say that $s$ and $t$ are derivably equal. As a consequence of \eqref{ax:scl1}
through \eqref{ax:scl3}, $\sleftand$ is the dual of $\sleftor$ and hence the
duals of the equations in $\EqFSCL$ are also derivable. We will use this fact
implicitly throughout our proofs. Observe that unlike with $\EqFFEL$, we have
an equation in $\EqFSCL$ for (a special case of) distributivity, i.e.,
\eqref{ax:scl10s}.

\begin{table}[htbp]
\hrule
\begin{align}
\false &= \neg\true 
  \label{ax:scl1}\tag{SCL1}\\
x \sleftor y &= \neg(\neg x \sleftand \neg y) 
  \label{ax:scl2}\tag{SCL2}\\
\neg \neg x &= x 
  \label{ax:scl3}\tag{SCL3}\\
(x \sleftand y) \sleftand z &= x \sleftand (y \sleftand z)
  \label{ax:scl7}\tag{SCL4}\\
\true \sleftand x &= x 
  \label{ax:scl4}\tag{SCL5}\\
x \sleftand \true &= x 
  \label{ax:scl5}\tag{SCL6}\\
\false \sleftand x &= \false 
  \label{ax:scl6}\tag{SCL7}\\
x \sleftand \false &= \neg x \sleftand \false
  \label{ax:scl8s}\tag{SCL8} \\
(x \sleftand \false) \sleftor y &= (x \sleftor \true) \sleftand y
  \label{ax:scl9s}\tag{SCL9} \\
(x \sleftand y) \sleftor (z \sleftand \false) &= (x \sleftor (z \sleftand
  \false)) \sleftand (y \sleftor (z \sleftand \false))
  \label{ax:scl10s}\tag{SCL10}
\end{align}
\hrule
\vspace{1em}
\caption{The set of equations \textbf{\EqFSCL}.}
\label{tab:eqfscl}
\end{table}


The following lemma shows some equations that will prove useful in Section
\ref{sec:snf}. These equations show how terms of the form $x \sleftand \false$
and $x \sleftor \true$ can be used to change the order in which atoms occur in
an $\SCL$-term. This is very different from the situation with $\FEL$, where
terms that contain the same atoms, but in a different order, are never
derivably equal. In terms of a comparison between $\EqFSCL$ and $\EqFFEL$ this
can be seen as a consequence of \eqref{ax:scl10s}.
\clearpage
\begin{lemma}
\label{lem:seqs}
The following equations can all be derived by equational logic and $\EqFSCL$.
\begin{enumerate}[itemsep=5pt]
\item $(x \sleftor y) \sleftand (z \sleftand \false) = (\neg x \sleftor (z
  \sleftand \false)) \sleftand (y \sleftand (z \sleftand \false))$
  \label{eq:b1}
\item $(x \sleftor (y \sleftand \false)) \sleftand (z \sleftand \false) = (\neg
  x \sleftor (z \sleftand \false)) \sleftand (y \sleftand \false)$
  \label{eq:b2}
\item $(x \sleftand (y \sleftor \true)) \sleftor (z \sleftand \false) = (x
  \sleftor (z \sleftand \false)) \sleftand (y \sleftor \true)$
  \label{eq:b3}
\end{enumerate}
\end{lemma}
\begin{proof}
We derive the equations in order.
\begin{align*}
(x \sleftor y) &\sleftand (z \sleftand \false) \\
&= (x \sleftor y) \sleftand ((z \sleftand \false) \sleftand \false)
  &\textrm{by \eqref{ax:scl6} and \eqref{ax:scl7}} \\
&= (x \sleftor y) \sleftand (\neg(z \sleftand \false) \sleftand \false)
  &\textrm{by \eqref{ax:scl8s}} \\
&= ((x \sleftor y) \sleftand \neg(z \sleftand \false)) \sleftand \false
  &\textrm{by \eqref{ax:scl7}} \\
&= ((\neg x \sleftand \neg y) \sleftor (z \sleftand \false)) \sleftand \false
  &\textrm{by \eqref{ax:scl8s}, \eqref{ax:scl2} and \eqref{ax:scl3}} \\
&= ((\neg x \sleftor (z \sleftand \false)) \sleftand (\neg y \sleftor (z
  \sleftand \false))) \sleftand \false
  &\textrm{by \eqref{ax:scl10s}} \\
&= (\neg x \sleftor (z \sleftand \false)) \sleftand ((\neg y \sleftor (z
  \sleftand \false)) \sleftand \false)
  &\textrm{by \eqref{ax:scl7}} \\
&= (\neg x \sleftor (z \sleftand \false)) \sleftand ((y \sleftand \neg(z
  \sleftand \false)) \sleftand \false)
  &\textrm{by \eqref{ax:scl8s}, \eqref{ax:scl2} and \eqref{ax:scl3}} \\
&= ((\neg x \sleftor (z \sleftand \false)) \sleftand y) \sleftand (\neg(z
  \sleftand \false) \sleftand \false)
  &\textrm{by \eqref{ax:scl7}} \\
&= ((\neg x \sleftor (z \sleftand \false)) \sleftand y) \sleftand ((z
  \sleftand \false) \sleftand \false)
  &\textrm{by \eqref{ax:scl8s}} \\
&= ((\neg x \sleftor (z \sleftand \false)) \sleftand y) \sleftand (z
  \sleftand \false)
  &\textrm{by \eqref{ax:scl7} and \eqref{ax:scl6}} \\
&= (\neg x \sleftor (z \sleftand \false)) \sleftand (y \sleftand (z
  \sleftand \false))
  &\textrm{by \eqref{ax:scl7}} \\[5pt]
(x \sleftor (y &\sleftand \false)) \sleftand (z \sleftand
  \false)\\
&= (\neg x \sleftor (z \sleftand \false)) \sleftand ((y \sleftand \false)
  \sleftand (z \sleftand \false))
  &\textrm{by part \eqref{eq:b1} of this lemma} \\
&= (\neg x \sleftor (z \sleftand \false)) \sleftand (y \sleftand \false)
  &\textrm{by \eqref{ax:scl6} and \eqref{ax:scl7}} \\[5pt]
(x \sleftand (y &\sleftor \true)) \sleftor (z \sleftand
  \false)\\
&= (x \sleftor (z \sleftand \false)) \sleftand ((y \sleftor \true) \sleftor (z
  \sleftand \false))
  &\textrm{by \eqref{ax:scl10s}} \\
&= (x \sleftor (z \sleftand \false)) \sleftand (y \sleftor \true)
  &\textrm{by the duals of \eqref{ax:scl6} and \eqref{ax:scl7}}
  &\qedhere
\end{align*}
\end{proof}

\begin{theorem}
\label{thm:sclsnd}
For all $P, Q \in \ST$, if $\EqFSCL \vdash P = Q$ then $\FSCLT \vDash P = Q$.
\end{theorem}
\begin{proof}
To see that identity, symmetry, transitivity and congruence hold in $\FSCLT$,
we refer the reader to the proof of Theorem \ref{thm:felsnd} and note that the
proofs for $\FSCLT$ are highly similar.

Verifying the validity of the equations in $\EqFSCL$ is cumbersome, but not
difficult. As an example we show it for \eqref{ax:scl3}. We have 
\begin{equation*}
\SE(\neg \neg P) = \SE(P)\ssub{\true}{\false}{\false}{\true}
\ssub{\true}{\false}{\false}{\true}  = \SE(P)
\end{equation*}
by a trivial structural induction on evaluation trees.  
\end{proof}

\section{SCL Normal Form}
\label{sec:snf}
To aid in our completeness proof we define a normal form for $\SCL$-terms.
Because the atoms in $\SCL$-terms may have side effects common normal forms for
PL such as Conjunctive Normal Form or Disjunctive Normal Form are not normal
forms for $\SCL$. For example, the term $a \sleftand (b \sleftor c)$ would be
written as $(a \sleftand b) \sleftor (a \sleftand c)$ in Disjunctive Normal
Form, but a trivial examination shows that the evaluation trees of these terms
are not the same. Our normal form is inspired by the $\FEL$ Normal Form
presented in Chapter \ref{chap:fel}. We present the grammar for our normal form
before we motivate it.


\begin{definition}
\label{def:snf}
A term $P \in \ST$ is said to be in \textbf{$\SCL$ Normal Form $(\SNF)$} if it
is generated by the following grammar.
\begin{align*}
P \in \SNF &::= P^\true ~\mid~ P^\false ~\mid~ P^\true \sleftand P^* \\
P^* &::= P^c ~\mid~ P^d \\
P^c &::= P^\ell ~\mid~ P^* \sleftand P^d \\
P^d &::= P^\ell ~\mid~ P^* \sleftor P^c \\
P^\ell &::= (a \sleftand P^\true ) \sleftor P^\false
  ~\mid~ (\neg a \sleftand P^\true ) \sleftor P^\false \\
P^\true &::= \true ~\mid~ (a \sleftand P^\true) \sleftor P^\true \\
P^\false &::= \false ~\mid~ (a \sleftor P^\false) \sleftand P^\false,
\end{align*}
where $a \in A$. We refer to $P^*$-forms as $*$-terms, to $P^\ell$-forms as
$\ell$-terms, to $P^\true$-forms as $\true$-terms and to $P^\false$-forms as
$\false$-terms. A term of the form $P^\true \sleftand P^*$ is referred to as a
$\true$-$*$-term.
\end{definition}

Without the presence of $\true$ and $\false$ in our language, a traditional
Negation Normal Form would have sufficed. Furthermore, if $A = \varnothing$, an
even more trivial normal form could be used, i.e., just $\true$ or $\false$.

When considering trees in the image of $\SE$ we note that some trees only have
$\true$-leaves, some only $\false$-leaves and some both $\true$-leaves and
$\false$-leaves. For any $\SCL$-term $P$, $\SE(P \sleftor \true)$ is a tree
with only $\true$-leaves, as can easily be seen from the definition of $\SE$.
All terms $P$ such that $\SE(P)$ is a tree with only $\true$-leaves are
rewritten to $\true$-terms. Similarly, for any term $P$, $\SE(P \sleftand
\false)$ is a tree with only $\false$-leaves. All $P$ such that $\SE(P)$ has
only $\false$-leaves are rewritten to $\false$-terms. The simplest trees in the
image of $\SE$ that have both types of leaves are $\SE(a)$ for $a \in A$. Any
(occurrence of an) atom that determines (in whole or in part) the yield of the
term, such as $a$ in this example, is referred to as a determinative (occurrence
of an) atom. This as opposed to a non-determinative (occurrence of an) atom,
such as the $a$ in $a \sleftor \true$, which does not determine (either in
whole or in part) the yield of the term.  Note that a term $P$ such that
$\SE(P)$ contains both $\true$ and $\false$ must contain at least one
determinative atom.

Terms that contain at least one determinative atom will be rewritten to
$\true$-$*$-terms. In $\true$-$*$-terms we encode each determinative atom
together with the non-determinative atoms that occur between it and the next
determinative atom in the term (reading from left to right) as an $\ell$-term.
Observe that the first atom in an $\ell$-term is the (only) determinative atom
in that $\ell$-term and that determinative atoms only occur in $\ell$-terms.
Also observe that the yield of an $\ell$-term is the yield of its determinative
atom. This is intuitively convincing, because the remainder of the atoms in any
$\ell$-term are non-determinative and hence do not contribute to its yield. The
non-determinative atoms that may occur before the first determinative atom are
encoded as a $\true$-term. A $\true$-$*$-term is the conjunction of a
$\true$-term encoding such atoms and a $*$-term, which contains only
conjunctions and disjunctions of $\ell$-terms. We could also have encoded such
atoms as an $\false$-term and then taken the disjunction with a $*$-term to
obtain a term with the same semantics. We consider $\ell$-terms to be `basic'
in $*$-terms in the sense that they are the smallest grammatical unit that
influences the yield of the $*$-term.

The $\ell$-terms in $\SNF$ are more complex than those in $\FEL$ Normal Form,
because short-circuiting allows for the possibility of evaluating different
non-determinative atoms depending on the yield of the determinative atom. This
is also the reason why the $\true$-terms and the $\false$-terms are more
complex. Although the atoms occurring in them are not determinative, their
yield can influence which atoms in the $\true$-term ($\false$-term) are
evaluated next.

We use $P^\true$, $P^*$, etc.~both to denote grammatical categories and as
variables for terms in those categories. The remainder of this section is
concerned with defining and proving correct the normalization function $\nfs:
\ST \to \SNF$. We will define $\nfs$ recursively using the functions
\begin{equation*}
\nfs^n: \SNF \to \SNF \quad\text{and}\quad
\nfs^c: \SNF \times \SNF \to \SNF.
\end{equation*}
The first of these will be used to rewrite negated $\SNF$-terms to $\SNF$-terms
and the second to rewrite the conjunction of two $\SNF$-terms to an
$\SNF$-term. By \eqref{ax:scl2} we have no need for a dedicated function that
rewrites the disjunction of two $\SNF$-terms to an $\SNF$-term.

We start by defining $\nfs^n$. Analyzing the semantics of $\true$-terms and
$\false$-terms together with the definition of $\SE$ on negations, it becomes
clear that $\nfs^n$ must turn $\true$-terms into $\false$-terms and vice versa.
We also remark that $\nfs^n$ must preserve the left-associativity of the
$*$-terms in $\true$-$*$-terms, modulo the associativity within $\ell$-terms.
We define $\nfs^n: \SNF \to \SNF$ as follows, using the auxiliary function
$\nfs^n_1: P^* \to P^*$ to `push down' or `push in' the negation symbols when
negating a $\true$-$*$-term. We note that there is no ambiguity between the
different grammatical categories present in an $\SNF$-term, i.e., any
$\SNF$-term is in exactly one of the grammatical categories identified in
Definition
\ref{def:snf}.
\begin{align}
\nfs^n(\true) &= \false
  \label{eq:nfsn1} \\
\nfs^n((a \sleftand P^\true) \sleftor Q^\true) &= (a \sleftor
  \nfs^n(Q^\true)) \sleftand \nfs^n(P^\true)
  \label{eq:nfsn2} \\
\nfs^n(\false) &= \true
  \label{eq:nfsn3} \\
\nfs^n((a \sleftor P^\false) \sleftand Q^\false) &= (a \sleftand
  \nfs^n(Q^\false)) \sleftor \nfs^n(P^\false)
  \label{eq:nfsn4} \\
\nfs^n(P^\true \sleftand Q^*) &= P^\true \sleftand \nfs^n_1(Q^*)
  \label{eq:nfsn5} \\
\nfs^n_1((a \sleftand P^\true) \sleftor Q^\false) &= (\neg a \sleftand
  \nfs^n(Q^\false)) \sleftor \nfs^n(P^\true)
  \label{eq:nfsn6} \\
\nfs^n_1((\neg a \sleftand P^\true) \sleftor Q^\false) &= (a \sleftand
  \nfs^n(Q^\false)) \sleftor \nfs^n(P^\true)
  \label{eq:nfsn7} \\
\nfs^n_1(P^* \sleftand Q^d) &= \nfs^n_1(P^*) \sleftor \nfs^n_1(Q^d)
  \label{eq:nfsn8} \\
\nfs^n_1(P^* \sleftor Q^c) &= \nfs^n_1(P^*) \sleftand \nfs^n_1(Q^c).
  \label{eq:nfsn9}
\end{align}

Now we turn to defining $\nfs^c$. These definitions have a great deal of
inter-dependence so we first present the definition for $\nfs^c$ when the first
argument is a $\true$-term. We see that the conjunction of a $\true$-term with
another terms always yields a term of the same grammatical category as the
second conjunct. 
\begin{align}
\nfs^c(\true, P) &= P
  \label{eq:nfsc1} \\
\nfs^c((a \sleftand P^\true) \sleftor Q^\true, R^\true) &= (a \sleftand
  \nfs^c(P^\true, R^\true)) \sleftor \nfs^c(Q^\true, R^\true)
  \label{eq:nfsc2} \\
\nfs^c((a \sleftand P^\true) \sleftor Q^\true, R^\false) &= (a \sleftor
  \nfs^c(Q^\true, R^\false)) \sleftand \nfs^c(P^\true, R^\false)
  \label{eq:nfsc3} \\
\nfs^c((a \sleftand P^\true) \sleftor Q^\true, R^\true \sleftand S^*) &=
  \nfs^c((a \sleftand P^\true) \sleftor Q^\true, R^\true) \sleftand S^*
  \label{eq:nfsc4}
\end{align}

For defining $\nfs^c$ where the first argument is an $\false$-term, we make use
of \eqref{ax:scl6}. This immediately shows that the conjunction of an
$\false$-term with another term is itself an $\false$-term.
\begin{align}
\nfs^c(P^\false, Q) &= P^\false
  \label{eq:nfsc5}
\end{align}

The case where the first conjunct is a $\true$-$*$-term and the second conjunct
is a $\true$-term is defined next. We will use an auxiliary function,
$\nfs^c_1: P^* \times P^\true \to P^*$, to turn conjunctions of a $*$-term with
a $\true$-term into $*$-terms. Together with \eqref{ax:scl7} this allows us to
define $\nfs^c$ for this case.
\begin{align}
\nfs^c(P^\true \sleftand Q^*, R^\true) &= P^\true \sleftand
  \nfs^c_1(Q^*, R^\true) 
  \label{eq:nfsc6} \\
\nfs^c_1((a \sleftand P^\true) \sleftor Q^\false, R^\true) &= (a \sleftand
  \nfs^c(P^\true, R^\true)) \sleftor Q^\false
  \label{eq:nfsc7} \\
\nfs^c_1((\neg a \sleftand P^\true) \sleftor Q^\false, R^\true) &= (\neg a
  \sleftand \nfs^c(P^\true, R^\true)) \sleftor Q^\false
  \label{eq:nfsc8} \\
\nfs^c_1(P^* \sleftand Q^d, R^\true) &= P^* \sleftand \nfs^c_1(Q^d, R^\true)
  \label{eq:nfsc9} \\
\nfs^c_1(P^* \sleftor Q^c, R^\true) &= \nfs^c_1(P^*, R^\true) \sleftor
  \nfs^c_1(Q^c, R^\true)
  \label{eq:nfsc10}
\end{align}

When the second conjunct is an $\false$-term, the result will naturally be an
$\false$-term itself. So we need to convert the $\true$-$*$-term to an
$\false$-term. Using \eqref{ax:scl7} we reduce this problem to converting a
$*$-term to an $\false$-term, for which we use the auxiliary function
$\nfs^c_2: P^* \times P^\false \to P^\false$.
\begin{align}
\nfs^c(P^\true \sleftand Q^*, R^\false) &= \nfs^c(P^\true, \nfs^c_2(Q^*,
  R^\false))
  \label{eq:nfsc11} \\
\nfs^c_2((a \sleftand P^\true) \sleftor Q^\false, R^\false) &= (a \sleftor
  Q^\false) \sleftand \nfs^c(P^\true, R^\false)
  \label{eq:nfsc12} \\
\nfs^c_2((\neg a \sleftand P^\true) \sleftor Q^\false, R^\false) &= (a
  \sleftor \nfs^c(P^\true, R^\false)) \sleftand Q^\false
  \label{eq:nfsc13} \\
\nfs^c_2(P^* \sleftand Q^d, R^\false) &= \nfs^c_2(P^*, \nfs^c_2(Q^d,
  R^\false))
  \label{eq:nfsc14} \\
\nfs^c_2(P^* \sleftor Q^c, R^\false) &= \nfs^c_2(\nfs^n(\nfs^c_1(P^*,
  \nfs^n(R^\false))), \nfs^c_2(Q^c, R^\false))
  \label{eq:nfsc15}
\end{align}

Finally we are left with conjunctions of two $\true$-$*$-terms, thus completing
the definition of $\nfs^c$.  We use the auxiliary function $\nfs^c_3: P^*
\times P^\true \sleftand P^* \to P^*$ to ensure that the result is a
$\true$-$*$-term.
\begin{align}
\nfs^c(P^\true \sleftand Q^*, R^\true \sleftand S^*) &= P^\true \sleftand 
  \nfs^c_3(Q^*, R^\true \sleftand S^*)
  \label{eq:nfsc16} \\
\nfs^c_3(P^*, Q^\true \sleftand R^\ell) &= \nfs^c_1(P^*, Q^\true) \sleftand
  R^\ell
  \label{eq:nfsc17} \\
\nfs^c_3(P^*, Q^\true \sleftand (R^* \sleftand S^d)) &= \nfs^c_3(P^*, Q^\true
  \sleftand R^*) \sleftand S^d
  \label{eq:nfsc18} \\
\nfs^c_3(P^*, Q^\true \sleftand (R^* \sleftor S^c)) &= \nfs^c_1(P^*, Q^\true)
  \sleftand (R^* \sleftor S^c)
  \label{eq:nfsc19} 
\end{align}

As promised, we now define the normalization function $\nfs: \ST \to \SNF$
recursively, using $\nfs^n$ and $\nfs^c$, as follows.
\begin{align}
\nfs(a) &= \true \sleftand ((a \sleftand \true) \sleftor \false)
  \label{eq:nfs1} \\
\nfs(\true) &= \true
  \label{eq:nfs2} \\
\nfs(\false) &= \false
  \label{eq:nfs3} \\
\nfs(\neg P) &= \nfs^n(\nfs(P))
  \label{eq:nfs4} \\
\nfs(P \sleftand Q) &= \nfs^c(\nfs(P), \nfs(Q))
  \label{eq:nfs5} \\
\nfs(P \sleftor Q) &= \nfs^n(\nfs^c(\nfs^n(\nfs(P)), \nfs^n(\nfs(Q))))
  \label{eq:nfs6}
\end{align}

\begin{restatable}{theorem}{thmnfscorrect}
\label{thm:nfs}
For any $P \in \ST$, $\nfs(P)$ terminates, $\nfs(P) \in \SNF$ and $\EqFSCL
\vdash \nfs(P) = P$.
\end{restatable}

In Appendix \ref{sec:sclnf} we first prove a number of lemmas showing that the
definitions $\nfs^n$ and $\nfs^c$ are correct and use those to prove the
theorem. We have chosen to use a function rather than a rewriting system to
prove the correctness of the normal form, because the author lacks experience
with term rewriting systems and because using a function relieves us of the
task of proving confluence for the underlying rewriting system.

In Section \ref{sec:ffelfscl} we show that $\FFEL$ is a sublogic of $\FSCL$ and
that any $\FEL$-term can be rewritten to an $\SCL$-term with the same
semantics. There we will pay special attention to the application of that
translation to terms in $\FEL$ Normal Form.

\section{Tree Structure}
In Section \ref{sec:sclcpl} we will prove that $\EqFSCL$ axiomatizes $\FSCLT$
by showing that if $P \in \SNF$ we can invert $\SE(P)$. To do this we need to
prove several structural properties of the trees in the image of $\SE$. In the
definition of $\SE$ we can see how $\SE(P \sleftand Q)$ is assembled from
$\SE(P)$ and $\SE(Q)$ and similarly for $\SE(P \sleftor Q)$. To decompose these
trees we will introduce some notation. The trees in the image of $\SE$ are all
finite binary trees over $A$ with leaves in $\{\true, \false\}$, i.e.,
$\SE[\ST] \subseteq \T$. We will now also consider the set $\Tone$ of binary
trees over $A$ with leaves in $\{\true, \false, \Box\}$, where `$\Box$' is
pronounced `box'. The box will be used as a placeholder when composing or
decomposing trees. Replacement of the leaves of trees in $\Tone$ by trees in
$\T$ or $\Tone$ is defined analogous to replacement for trees in $\T$, adopting
the same notational conventions.

For example we have by definition of $\SE$ that $\SE(P \sleftand Q)$ can be
decomposed as
\begin{equation*}
\SE(P)\sub{\true}{\Box}\sub{\Box}{\SE(Q)},
\end{equation*}
where $\SE(P)\sub{\true}{\Box} \in \Tone$ and $\SE(Q) \in \T$. We note that
this only works because the trees in the image of $\SE$, or in $\T$ in general,
do not contain any boxes. We start by analyzing the $\SE$-image of
$\ell$-terms.

\begin{lemma}[Structure of $\ell$-terms]
\label{lem:slitstf}
There is no $\ell$-term $P$ such that $\SE(P)$ can be decomposed as
$X\sub{\Box}{Y}$ with $X \in \Tone$ and $Y \in \T$, where $X \neq \Box$, but
does contain $\Box$, and $Y$ contains occurrences of both $\true$ and
$\false$.
\end{lemma}
\begin{proof}
Let $P$ be some $\ell$-term. When we analyze the grammar of $P$ we find that
one branch from the root of $\SE(P)$ will only contain $\true$ and not $\false$
and the other branch vice versa. Hence if $\SE(P) = X\sub{\Box}{Y}$ and $Y$
contains occurrences of both $\true$ and $\false$, then $Y$ must contain the
root and hence $X = \Box$.
\end{proof}

By definition a $*$-term contains at least one $\ell$-term and hence for any
$*$-term $P$, $\SE(P)$ contains both $\true$ and $\false$. The following lemma
provides the $\SE$-image of the rightmost $\ell$-term in a $*$-term to witness
this fact.

\begin{lemma}[Determinativeness]
\label{lem:sperttf}
For all $*$-terms $P$, $\SE(P)$ can be decomposed as $X\sub{\Box}{Y}$ with $X
\in \Tone$ and $Y \in \T$ such that $X$ contains $\Box$ and $Y = \SE(Q)$ for
some $\ell$-term $Q$. Note that $X$ may be $\Box$. We will refer to $Y$ as the
witness for this lemma for $P$.
\end{lemma}
\begin{proof}
By induction on the complexity of $*$-terms $P$ modulo the complexity of
$\ell$-terms. In the base case $P$ is an $\ell$-term and $\SE(P) =
\Box\sub{\Box}{\SE(P)}$ is the desired decomposition by Lemma
\ref{lem:slitstf}. For the induction we have to consider both $\SE(P \sleftand
Q)$ and $\SE(P \sleftor Q)$.

We start with $\SE(P \sleftand Q)$ and let $X\sub{\Box}{Y}$ be the
decomposition for $\SE(Q)$ which we have by induction hypothesis. Since by
definition of $\SE$ on $\sleftand$ we have
\begin{equation*}
\SE(P \sleftand Q) = \SE(P)\sub{\true}{\SE(Q)}
\end{equation*}
we also have
\begin{equation*}
\SE(P \sleftand Q) = \SE(P)\sub{\true}{X\sub{\Box}{Y}} =
\SE(P)\sub{\true}{X}\sub{\Box}{Y}.
\end{equation*}
The last equality is due to the fact that $\SE(P)$ does not contain any boxes.
This gives our desired decomposition. The case for $\SE(P \sleftor Q)$ is
analogous.
\end{proof}

The following lemma illustrates another structural property of trees in the
image of $*$-terms under $\SE$, namely that the left branch of any
determinative atom in such a tree is different from its right branch.

\begin{lemma}[Non-decomposition]
\label{lem:snondectf}
There is no $*$-term $P$ such that $\SE(P)$ can be decomposed as
$X\sub{\Box}{Y}$ with $X \in \Tone$ and $Y \in \T$, where $X \neq \Box$ and $X$
contains $\Box$, but not $\true$ or $\false$.
\end{lemma}
\begin{proof}
By induction on $P$ modulo the complexity of $\ell$-terms. The base case covers
$\ell$-terms and follows immediately from Lemma~\ref{lem:sperttf} ($\SE(P)$
contains occurrences of both $\true$ and $\false$) and Lemma~\ref{lem:slitstf}
(no non-trivial decomposition exists that contains both). For the induction we
assume that the lemma holds for all $*$-terms with lesser complexity than $P
\sleftand Q$ and $P \sleftor Q$.

We start with the case for $\SE(P \sleftand Q)$. Suppose for contradiction that
$\SE(P \sleftand Q) = X\sub{\Box}{Y}$ with $X \neq \Box$ and $X$ not containing
any occurrences of $\true$ or $\false$. Let $R$ be a witness of Lemma
\ref{lem:sperttf} for $P$. Now note that $\SE(P \sleftand Q)$ has a subtree
$R\sub{\true}{\SE(Q)}$. Because $Y$ must contain both the occurrences of
$\false$ in the one branch of $R\sub{\true}{\SE(Q)}$ as well as the occurrences
of $\SE(Q)$ in the other (because they contain $\true$ and $\false$), Lemma
\ref{lem:slitstf} implies that $Y$ must (strictly) contain $\SE(Q)$. Hence
there is a $Z \in \T$ such that $\FE(P) = X\sub{\Box}{Z}$, which violates the
induction hypothesis. The case for $\SE(P \sleftor Q)$ is symmetric.
\end{proof}

We now arrive at two crucial definitions for our completeness proof. When
considering $*$-terms, we already know that $\SE(P \sleftand Q)$ can be
decomposed as
\begin{equation*}
\SE(P)\sub{\true}{\Box}\sub{\Box}{\SE(Q)}.
\end{equation*}
Our goal now is to give a definition for a type of decomposition so that this
is the only such decomposition for $\SE(P \sleftand Q)$. We also ensure that
$\SE(P \sleftor Q)$ does not have a decomposition of that type, so that we can
distinguish $\SE(P \sleftand Q)$ from $\SE(P \sleftor Q)$. Similarly, we need
to define another type of decomposition so that $\SE(P \sleftor Q)$ can only be
decomposed as 
\begin{equation*}
\SE(P)\sub{\false}{\Box}\sub{\Box}{\SE(Q)}
\end{equation*}
and that $\SE(P \sleftand Q)$ does not have a decomposition of that type.

\begin{definition}
The pair $(Y, Z) \in \Tone \times \T$ is a \textbf{candidate conjunction
decomposition (ccd)} of $X \in \T$, if
\begin{itemize}
\item $X = Y\sub{\Box}{Z}$,
\item $Y$ contains $\Box$,
\item $Y$ contains $\false$, but not $\true$, and
\item $Z$ contains both $\true$ and $\false$.
\end{itemize}
Similarly, $(Y, Z)$ is a \textbf{candidate disjunction decomposition (cdd)} of
$X$, if
\begin{itemize}
\item $X = Y\sub{\Box}{Z}$,
\item $Y$ contains $\Box$,
\item $Y$ contains $\true$, but not $\false$, and
\item $Z$ contains both $\true$ and $\false$.
\end{itemize}
\end{definition}

We note that the ccd and cdd are not necessarily the decompositions we are
looking for, because, for example, $\SE((P \sleftand Q) \sleftand R)$ has a ccd
$(\SE(P)\sub{\true}{\Box}, \SE(Q \sleftand R))$, whereas the decomposition we
need is $(\SE(P \sleftand Q)\sub{\true}{\Box}, \SE(R))$. Therefore we refine
these definitions to obtain the decompositions we seek.

\begin{definition}
The pair $(Y, Z) \in \Tone \times \T$ is a \textbf{conjunction decomposition
(cd)} of $X \in \T$, if it is a ccd of $X$ and there is no other ccd $(Y', Z')$
of $X$ where the depth of $Z'$ is smaller than that of $Z$. Similarly, $(Y, Z)$
is a \textbf{disjunction decomposition (dd)} of $X$, if it is a cdd of $X$ and
there is no other cdd $(Y', Z')$ of $X$ where the depth of $Z'$ is smaller than
that of $Z$.
\end{definition}

\begin{theorem}
\label{thm:scddd}
For any $*$-term $P \sleftand Q$, i.e., with $P \in P^*$ and $Q \in P^d$,
$\SE(P \sleftand Q)$ has the (unique) cd
\begin{equation*}
(\SE(P)\sub{\true}{\Box}, \SE(Q))
\end{equation*}
and no dd. For any $*$-term $P \sleftor Q$, i.e., with $P \in P^*$ and $Q \in
P^c$, $\SE(P \sleftor Q)$ has no cd and its (unique) dd is
\begin{equation*}
(\SE(P)\sub{\false}{\Box}, \SE(Q)).
\end{equation*}
\end{theorem}
\begin{proof}
By simultaneous induction on $P \sleftand Q$ and $P \sleftor Q$ modulo the
complexity of $\ell$-terms. In the basis we have to consider, for $\ell$-terms
$P$ and $Q$, the terms $P \sleftand Q$ and $P \sleftor Q$. Both of these are
covered by the cases in the induction where the second conjunct (or disjunct)
$Q$ is an $\ell$-term. This is valid reasoning, since we don't call upon the
induction hypothesis in those cases. For the induction we assume that the
theorem holds for all $*$-terms with lesser complexity than $P \sleftand Q$ and
$P \sleftor Q$. We first treat the case for $P \sleftand Q$.

First for the cd. Note that $(\SE(P)\sub{\true}{\Box}, \SE(Q))$ is a ccd of
$\SE(P \sleftand Q)$ by definition of $\SE$ on $\sleftand$ (for the first
condition) and Lemma \ref{lem:sperttf} (for the third and fourth condition). We
also know that for any ccd $(Y, Z)$ either $Z$ contains or is contained in
$\SE(Q)$. For suppose otherwise, then $Y$ will contain an occurrence of
$\true$, namely the one we know by Lemma \ref{lem:sperttf} that $\SE(Q)$ has.
By the above it suffices to show that there is no ccd $(Y, Z)$ where $Z$ is
strictly contained in $\SE(Q)$. Suppose for contradiction that such a ccd $(Y,
Z)$ does exist.

By definition of $*$-terms $Q$ is either an $\ell$-term or a disjunction. If
$Q$ is an $\ell$-term and $Z$ is strictly contained in $\SE(Q)$ then $Z$
does not contain both $\true$ and $\false$ by Lemma \ref{lem:slitstf}.
Therefore $(\SE(P)\sub{\true}{\Box}, \SE(Q))$ is the \emph{unique} cd for
$\SE(P \sleftand Q)$.

If $Q$ is a disjunction, then if $Z$ is strictly contained in $\SE(Q)$ we can
decompose $\SE(Q)$ as $\SE(Q) = U\sub{\Box}{Z}$ for some $U \in \Tone$ that
contains but is not equal to $\Box$. By Lemma~\ref{lem:snondectf} this implies
that $U$ contains either $\true$ or $\false$. If it contains $\true$, then so
does $Y$, because 
\begin{equation*}
Y = \SE(P)\sub{\true}{U},
\end{equation*}
and $(Y, Z)$ is not a ccd for $\SE(P \sleftand Q)$. If it only contains
$\false$ then $(U, Z)$ is a ccd for $\SE(Q)$ which violates the induction
hypothesis. Therefore $(\SE(P)\sub{\true}{\Box}, \SE(Q))$ is the \emph{unique}
cd for $\SE(P \sleftand Q)$.

Now for the dd. It suffices to show that there is no cdd for $\SE(P \sleftand
Q)$. Again $Q$ is either an $\ell$-term or a disjunction. Suppose for
contradiction that $(Y, Z)$ is a cdd for $\SE(P \sleftand Q)$. If $Q$ is an
$\ell$-term, then $Z$ must contain all occurrences of $\false$ in $\SE(P
\sleftand Q)$. So in particular it must contain all occurrences of $\false$ in
$\SE(Q)$. It must also contain at least one occurrence of $\true$. Hence by
Lemma \ref{lem:slitstf} $Z$ must contain $\SE(Q)$. But then $Z$ contains all
the occurrences of $\true$ in $\SE(P \sleftand Q)$ and hence $X$ does not
contain any occurrences of $\true$.  Therefore there is no cdd for $\SE(P
\sleftand Q)$.

If $Q$ is a disjunction then $Z$ must contain all occurrences of $\false$ in
$\SE(P \sleftand Q)$. Let $R$ be a witness of Lemma \ref{lem:sperttf} for $P$.
Now note that $R\sub{\true}{\SE(Q)}$ is a subtree of $\SE(P \sleftand Q)$. Also
note that Lemma \ref{lem:slitstf} implies that there is no way to decompose
$R\sub{\true}{\SE(Q)}$ such that $R\sub{\true}{\SE(Q)} = U\sub{\Box}{V}$ for
some $U \in \Tone$ that contains but is not equal to $\Box$ and some $V \in \T$
containing occurrences of both $\SE(Q)$ and $\false$. So because $Z$ must
contain all occurrences of $\false$ in $\SE(P \sleftand Q)$, it must strictly
contain $\SE(Q)$. But all the occurrences of $\true$ in $\SE(P \sleftand Q)$
are in occurrences of $\SE(Q)$. Hence $X$ does not contain any occurrences of
$\true$. Therefore there is no cdd for $\SE(P \sleftand Q)$.  The case for
$\SE(P \sleftor Q)$ is symmetric.
\end{proof}

At this point we have the tools necessary to invert $\SE$ on $*$-terms, at
least down to the level of $\ell$-terms. We can easily detect if a tree in the
image of $\SE$ is in the image of $P^\ell$, because all leaves to the left of
the root are one truth value, while all the leaves to the right are the other.
To invert $\SE$ on $\true$-$*$-terms we still need to be able to reconstruct
$\SE(P^\true)$ and $\SE(Q^*)$ from $\SE(P^\true \sleftand Q^*)$. To this end we
define a $\true$-$*$-decomposition, as with cds and dds we first define a
candidate $\true$-$*$-decomposition.

\begin{definition}
The pair $(Y, Z) \in \Tone \times \T$ is a \textbf{candidate
$\true$-$*$-decomposition (ctsd)} of $X \in \T$, if $X = Y\sub{\Box}{Z}$, $Y$
does not contain $\true$ or $\false$ and there is no decomposition $(U, V) \in
\Tone \times \T$ of $Z$ such that
\begin{itemize}
\item $Z = U\sub{\Box}{V}$,
\item $U$ contains $\Box$,
\item $U \neq \Box$, and
\item $U$ contains neither $\true$ nor $\false$.
\end{itemize}
\end{definition}

Unlike with $\FEL$, this is not the decomposition we seek in this case.
Consider for example that there is a $\true$-term with the following semantics:
\begin{center}
\beginpgfgraphicnamed{scl3}
\begin{tikzpicture}[%
level distance=7.5mm,
level 1/.style={sibling distance=30mm},
level 2/.style={sibling distance=15mm},
level 3/.style={sibling distance=7.5mm}
]
\node (a) {$a$}
  child {node (b1) {$b$}
    child {node (c1) {$c$}
      child {node (d1) {$\true$}} 
      child {node (d2) {$\true$}}
    }
    child {node (c2) {$d$}
      child {node (d3) {$\true$}} 
      child {node (d4) {$\true$}}
    }
  }
  child {node (b2) {$b$}
    child {node (c3) {$c$}
      child {node (d5) {$\true$}} 
      child {node (d6) {$\true$}}
    }
    child {node (c4) {$d$}
      child {node (d7) {$\true$}} 
      child {node (d8) {$\true$}}
    }
  };
\end{tikzpicture}
\endpgfgraphicnamed
\end{center}
Let $P^\true$ be the $\true$-term with these semantics and observe that
$\SE(P^\true \sleftand Q^*)$ has a ctsd
\begin{equation*}
(\Box \tlef a \trig \Box, (\SE(Q^*) \tlef c \trig \SE(Q^*)) \tlef b \trig
(\SE(Q^*) \tlef d \trig \SE(Q^*))).
\end{equation*}
But the decomposition we seek is $(\SE(P^\true)\sub{\true}{\Box}, \SE(Q^*)$.
Hence we will refine this definition to aid in the theorem below.

\begin{definition}
The pair $(Y, Z) \in \Tone \times \T$ is a \textbf{$\true$-$*$-decomposition
(tsd)} of $X \in \T$, if it is a ctsd of $X$ and there is no other ctsd $(Y',
Z')$ of $X$ where the depth of $Z'$ is smaller than that of $Z$.
\end{definition}

\begin{theorem}
\label{thm:stsd}
For any $\true$-term $P$ and $*$-term $Q$ the (unique) tsd of $\SE(P \sleftand
Q)$ is
\begin{equation*}
(\SE(P)\sub{\true}{\Box}, \SE(Q)).
\end{equation*}
\end{theorem}
\begin{proof}
First we observe that $(\SE(P)\sub{\true}{\Box}, \SE(Q))$ is a ctsd because by
definition of $\SE$ on $\sleftand$ we have $\SE(P)\sub{\true}{\SE(Q)} = \SE(P
\sleftand Q)$ and $\SE(Q)$ is non-decomposable by Lemma \ref{lem:snondectf}.

Suppose for contradiction that there is ctsd $(Y, Z)$ such that the depth of
$Z$ is smaller than that of $\SE(Q)$. Now $Z$ must contain or be contained in
$\SE(Q)$ for otherwise $Y$ would contain $\true$ or $\false$, i.e., the ones we
know $\SE(Q)$ has by Lemma \ref{lem:sperttf}. Clearly the former cannot be the
case, for then $Z$ would have a greater depth than $\SE(Q)$. So the latter is
the case and $\SE(Q) = U\sub{\Box}{Z}$ for some $U \in \Tone$ that is not equal
to $\Box$ and does not contain $\true$ or $\false$ (because then $Y$ would
too). But this violates Lemma \ref{lem:snondectf}, which states that no such
decomposition exists.
\end{proof}

\section{Completeness}
\label{sec:sclcpl}
With the two theorems from the previous section, we can prove completeness for
$\FSCLT$. We define three auxiliary functions to aid in our definition of the
inverse of $\SE$ on $\SNF$. Let $\cd : \T \to \Tone \times \T$ be the function
that returns the conjunction decomposition of its argument, $\dd$ of the same
type its disjunction decomposition and $\tsd$, also of the same type, its
$\true$-$*$-decomposition. Naturally, these functions are undefined when their
argument does not have a decomposition of the specified type. Each of these
functions returns a pair and we will use $\cd_1$ ($\dd_1$, $\tsd_1$) to denote
the first element of this pair and $\cd_2$ ($\dd_2$, $\tsd_2$) to denote the
second element.

We define $\invs: \T \to \ST$ using the functions $\invs^\true: \T \to \ST$ for
inverting trees in the image of $\true$-terms and $\invs^\false$, $\invs^\ell$
and $\invs^*$ of the same type for inverting trees in the image of
$\false$-terms, $\ell$-terms and $*$-terms, respectively. These functions are
defined as follows.
\begin{align}
\invs^\true(X) &=
  \begin{cases}
    \true
      &\textrm{if $X = \true$} \\
    (a \sleftand \invs^\true(Y)) \sleftor \invs^\true(Z) 
      &\textrm{if $X = Y \tlef a \trig Z$}
  \end{cases} \\ \displaybreak[0]
\invs^\false(X) &=
  \begin{cases}
    \false
      &\textrm{if $X = \false$} \\
    (a \sleftor \invs^\false(Z)) \sleftand \invs^\false(Y)
      &\textrm{if $X = Y \tlef a \trig Z$}
  \end{cases} \\ \displaybreak[0]
\invs^\ell(X) &=
  \begin{cases}
    (a \sleftand \invs^\true(Y)) \sleftor \invs^\false(Z) 
      &\textrm{if $X = Y \tlef a \trig Z$ for some $a \in A$} \\
      &\textrm{and $Y$ only has $\true$-leaves} \\
    (\neg a \sleftand \invs^\true(Z)) \sleftor \invs^\false(Y)
      &\textrm{if $X = Y \tlef a \trig Z$ for some $a \in A$} \\
      &\textrm{and $Z$ only has $\true$-leaves}
  \end{cases} \\ \displaybreak[0]
\invs^*(X) &=
  \begin{cases}
    \invs^*(\cd_1(X)\sub{\Box}{\true}) \sleftand \invs^*(\cd_2(X))
      &\textrm{if $X$ has a cd} \\
    \invs^*(\dd_1(X)\sub{\Box}{\false}) \sleftor \invs^*(\dd_2(X))
      &\textrm{if $X$ has a dd} \\
    \invs^\ell(X)
      &\textrm{otherwise}
  \end{cases} \\ \displaybreak[0]
\invs(X) &=
  \begin{cases}
    \invs^\true(X)
      &\textrm{if $X$ has only $\true$-leaves} \\
    \invs^\false(X)
      &\textrm{if $X$ has only $\false$-leaves} \\
    \invs^\true(\tsd_1(X)\sub{\Box}{\true}) \sleftand \invs_*(\tsd_2(X))
      &\textrm{otherwise}
  \end{cases}
\end{align}

\begin{restatable}{theorem}{thminvscorrect}
\label{thm:sclinv}
For all $P \in \SNF$, $\invs(\SE(P)) \equiv P$.
\end{restatable}

The proof for this theorem can be found in Appendix \ref{sec:sclinv}. For the
sake of completeness, we separately state the completeness result below.

\begin{theorem}
\label{thm:sclcpl}
For all $P, Q \in \ST$, if $\FSCLT \vDash P = Q$ then $\EqFSCL \vdash P = Q$.
\end{theorem}
\begin{proof}
It suffices to show that for $P,Q \in \SNF$, $\SE(P) = \SE(Q)$ implies $P
\equiv Q$. To see this suppose that $P'$ and $Q'$ are two $\SCL$-terms and
$\SE(P') = \SE(Q')$. We know that $P'$ is derivably equal to an $\SNF$-term
$P$, i.e., $\EqFSCL \vdash P' = P$, and that $Q'$ is derivably equal to an
$\SNF$-term $Q$, i.e., $\EqFSCL \vdash Q' = Q$. Theorem \ref{thm:sclsnd} then
gives us $\SE(P') = \SE(P)$ and $\SE(Q') = \SE(Q)$. Hence by the result $P
\equiv Q$ and in particular $\EqFSCL \vdash P = Q$. Transitivity then gives us
$\EqFSCL \vdash P' = Q'$ as desired.

The result follows immediately from Theorem \ref{thm:sclinv}.
\end{proof}

In Section \ref{sec:fsclpa} we use this result to prove that $\FSCLT$ is
equivalent to $\FSCL$ as it is defined by Bergstra and Ponse in \cite{scl}.

\chapter{Relating FFEL to FSCL and Proposition Algebra}
\label{chap:rel}

This chapter is concerned with explaining the connection between $\FFEL$ and
$\FSCL$, a connection we formalize in the setting of Proposition Algebra.
Because the main results of this thesis were not presented in this setting, we
will forgo a detailed introduction to Proposition Algebra. Instead we refer
the reader to \cite{pa} for a proper introduction to Proposition Algebra and to
\cite{scl} and \cite{pascl} for an introduction to $\FSCL$ as it is defined in
terms of Proposition Algebra. We will however, very briefly, fix some core
concepts and notation from this setting.

In \cite{pa} Bergstra and Ponse introduce Proposition Algebra for reasoning
about the sequential evaluation of propositional terms using the ternary
connective $y \lef x \rig z$, to be read as `if $x$ then $y$ else $z$' and
called `Hoare's conditional operator', defined in \cite{Hoare85}. The terms
under consideration can be described in Backus-Naur Form, letting $A$ be a
countable set of atoms, by
\begin{equation*}
P \in \CPT ::= a \in A ~\mid~ \true ~\mid~ \false ~\mid~ P \lef P \rig P.
\end{equation*}
The equality of two of these terms is defined by the set of axioms $\CP$, for
`Conditional Propositions':
\begin{align}
x \lef \true \rig y &= x
  \label{ax:cp1}\tag{CP1} \\
x \lef \false \rig y &= y
  \label{ax:cp2}\tag{CP2} \\
\true \lef x \rig \false &= x
  \label{ax:cp3}\tag{CP3} \\
x \lef (y \lef z \rig u) \rig v &= (x \lef y \rig v) \lef z \rig
  (x \lef u \rig v).
  \label{ax:cp4}\tag{CP4}
\end{align}
When the equality of two terms $s$ and $t$ in $\CPT$, possibly containing
variables, can be derived from equational logic and
\eqref{ax:cp1}--\eqref{ax:cp4}, we denote this by $\CP \vdash s = t$. Bergstra
and Ponse extend Proposition Algebra with negation and the short-circuit
connectives $\sleftand$ and $\sleftor$ to obtain the set $\CPT_s$ of closed
terms, where we see that $\ST \subset \CPT_s$. They extend $\CP$ with the
following defining equations for the newly introduced connectives.
\begin{align}
\neg x &= \false \lef x \rig \true 
  \label{eq:sneg} \\
x \sleftand y &= y \lef x \rig \false 
  \label{eq:sand} \\
x \sleftor y &= \true \lef x \rig y
  \label{eq:sor}
\end{align}
Let $\CP_s$ denote $\CP$ together with these defining equations. If the
equality of two terms $s$ and $t$ in $\CPT_s$, possibly containing variables,
can be derived from equational logic and $\CP_s$, we denote this by $\CP_s
\vdash s = t$. Free Short-Circuit Logic ($\FSCL$) is then defined as follows.
For all $P, Q \in \ST$,
\begin{equation*}
\FSCL \vDash P = Q \ \Longleftrightarrow\  \CP_s \vdash P = Q.
\end{equation*}

We will show in Section \ref{sec:ffelpa} that $\FFEL$ can also be expressed in
terms of Proposition Algebra with an extended signature, by adding defining
equations for the \emph{full} left-sequential connectives to $\CP$. In Section
\ref{sec:fsclpa} we will prove that $\FSCLT$, as we defined it in Chapter
\ref{chap:scl}, is equivalent to $\FSCL$.  Finally in Section
\ref{sec:ffelfscl} we will prove that $\FFEL$ is a sublogic of $\FSCL$ by
showing that its connectives are definable in $\FSCL$, which will allow us to
define a general left-sequential logic.

\section{Relating FFEL to Proposition Algebra}
\label{sec:ffelpa}

To relate $\FFEL$ to Proposition Algebra, we first add negation, $\leftand$ and
$\leftor$ to the signature of Proposition Algebra, to obtain the set $\CPT_f$
of closed terms, where we note that $\FT \subset \CPT_f$. We then extend the
set of equations $\CP$ with the following defining equations for the full
left-sequential connectives.
\begin{align}
\neg x &= \false \lef x \rig \true
  \label{eq:fneg} \\
x \leftand y &= y \lef x \rig (\false \lef y \rig \false)
  \label{eq:fand} \\
x \leftor y &= (\true \lef y \rig \true) \lef x \rig y
  \label{eq:for}
\end{align}
Let $\CP_f$ denote $\CP$ together with these three equations. When two terms
$s, t \in \CPT_f$, possibly containing variables, are derivable by equational
logic and $\CP_f$ we denote this by $\CP_f \vdash s = t$.

Our goal is to prove that $\FFEL$ can also be characterized by $\CP_f$, i.e.,
we will show that for all $P, Q \in \FT$,
\begin{equation*}
\FFEL \vDash P = Q \ \Longleftrightarrow\ \CP_f \vdash P = Q.
\end{equation*}
To this end we first define the function $\E: \CPT \to \T$ that will interpret
$\CPT$-terms in $\T$. We will then extend this definition to $\E_f: \CPT_f
\to \T$.

With the informal semantics of `if $x$ then $y$ else $z$' in mind for terms of
the form $y \lef x \rig z$, defining $\E$ becomes straightforward.
\begin{definition}
Let $A$ be a countable set of atoms and let $\T$ be the set all finite binary
trees over $A$ with leaves in $\{\true, \false\}$. We define the unary
\textbf{Conditional Evaluation} function $\E: \CPT \to \T$ as:
\begin{align*}
\E(\true) &= \true \\
\E(\false) &= \false \\
\E(a) &= \true \tlef a \trig \false &\textrm{for $a \in A$} \\
\E(Q \lef P \rig R) &= \E(P)\ssub{\true}{\E(Q)}{\false}{\E(R)}
\end{align*}
\end{definition}
We observe that $\E[\CPT] = \T$. We extend this definition to $\E_f: \CPT_f \to
\T$ by adding the following clauses to the definition.
\begin{align*}
\E_f(\neg P) &= \E_f(P)\ssub{\true}{\false}{\false}{\true} \\
\E_f(P \leftand Q) &=
  \E_f(P)\ssub{\true}{\E_f(Q)}{\false}{\E_f(Q)\sub{\true}{\false}} \\
\E_f(P \leftor Q) &=
\E_f(P)\ssub{\true}{\E_f(Q)\sub{\false}{\true}}{\false}{\E_f(Q)}
\end{align*}
It is now trivial to see that $\E_f$ restricted to $\FEL$-terms is equal to
$\FE$. 

\begin{theorem}
For all $P, Q \in \FT$,
\begin{equation*}
\FFEL \vDash P = Q \ \Longleftrightarrow\  \CP_f \vdash P = Q.
\end{equation*}
\end{theorem}
\begin{proof}
First we show that $\FFEL \vDash P = Q$ implies that $\CP_f \vdash P = Q$.
If $\FFEL \vDash P = Q$, then $\FE(P) = \FE(Q)$ and by Theorem \ref{thm:felcpl}
$\EqFFEL \vdash P = Q$. It suffices to prove that $\EqFFEL$ is sound with
respect to $\CP_f$. It is immediate that identity, symmetry, transitivity and
congruence are sound and checking that \eqref{ax:ft1}--\eqref{ax:ft10} are
valid is cumbersome, but not difficult. As an example we show this for
\eqref{ax:ft8}:
\begin{align*}
P \leftand \false
&= \false \lef P \rig (\false \lef \false \rig \false) 
  &\textrm{by \eqref{eq:fand}}\\
&= (\false \lef \false \rig (\false \lef \false \rig \false)) \lef P \rig
  \false
  &\textrm{by \eqref{ax:cp2}} \\
&= (\false \lef \false \rig (\false \lef \false \rig \false)) \lef P \rig
  (\false \lef \true \rig (\false \lef \false \rig \false))
  &\textrm{by \eqref{ax:cp1}} \\
&= \false \lef (\false \lef P \rig \true) \rig (\false \lef \false \rig
  \false)
  &\textrm{by \eqref{ax:cp4}} \\
&= \neg P \leftand \false
  &\textrm{by \eqref{eq:fand} and \eqref{eq:fneg}}
\end{align*}
Given this soundness, $\FFEL \vDash P = Q$ implies $\CP_f \vdash P = Q$.

Next we show that $\CP_f \vdash P = Q$ implies $\FFEL \vDash P = Q$. Because
$\E_f$ restricted to $\FT$ equals $\FE$, it suffices to show that $\CP_f$ is
sound with respect to $\E_f$-equality, i.e., that for all $R, S \in \CPT_f$,
$\CP_f \vdash R = S$ implies $\E_f(R) = \E_f(S)$. This proof is also trivial.
As an example we show this for \eqref{ax:cp2} as follows:
\begin{align*}
\E_f(R \lef \false \rig S) &=
  \E_f(\false)\ssub{\true}{\E_f(R)}{\false}{\E_f(S)} \\
&= \false\ssub{\true}{\E_f(R)}{\false}{\E_f(S)} \\
&= \E_f(S),
\end{align*}
and for \eqref{eq:fand} as
\begin{align*}
\E_f(R \leftand S)
&= \E_f(R)\ssub{\true}{\E_f(S)}{\false}{\E_f(S)\sub{\true}{\false}} \\
&= \E_f(R)\ssub{\true}{\E_f(S)}{\false}{\E_f(S)
  \ssub{\true}{\false}{\false}{\false}} \\
&= \E_f(R)\ssub{\true}{\E_f(S)}{\false}{\E_f(\false \lef S \rig \false)} \\
&= \E_f(S \lef R \rig (\false \lef S \rig \false)). &&\qedhere
\end{align*}
\end{proof}

Given the definition of $\FSCL$, the reason for defining $\FFEL$ in terms of
evaluation trees rather than by using Proposition Algebra deserves some
clarification. We feel that this makes our completeness proof more
straightforward than it would have been had we defined $\FFEL$ in terms of
Proposition Algebra. Although $\CPT$-terms are easily interpreted as trees, we
would have had to use a basic form for $\CPT$-terms, such as \cite[Definition
3.1]{pa}, to perform our analysis as done in Section \ref{sec:feltrees}. In
fact our function $\E$ converts $\CPT$-terms to such basic forms if we read
`$\tlef$' as `$\lef$' and `$\trig$' as `$\rig$', thus explaining our choice of
notation for trees. We have taken the notation $y \tlef x \trig z$ from the
setting of Thread Algebra, see, e.g., \cite{threadalg}, where it is used to
denote the post-conditional composition of threads.

\section{EqFSCL axiomatizes FSCL}
\label{sec:fsclpa}
It has been an open question since $\FSCL$ was first defined in \cite{scl}
whether or not $\EqFSCL$, or an equivalent set of equations such as the one
presented in \cite{scl} itself, axiomatizes $\FSCL$. Given Theorem
\ref{thm:sclcpl} it now suffices to show that $\FSCLT$ is equivalent to
$\FSCL$, which we shall prove analogous to how we proved the theorem in the
previous section.

We define the function $\E_s: \CPT_s \to \T$ that interprets
$\CPT_s$-terms as evaluation trees by extending the definition of $\E$ with
the following clauses.
\begin{align*}
\E_s(\neg P) &= \E_s(P)\ssub{\true}{\false}{\false}{\true} \\
\E_s(P \sleftand Q) &= \E_s(P)\sub{\true}{\E_s(Q)} \\
\E_s(P \sleftor Q) &= \E_s(P)\sub{\false}{\E_s(Q)}
\end{align*}
It is now trivial to see that $\E_s$ restricted to $\SCL$-terms is equal to
$\SE$. 

\begin{theorem}
\label{thm:fsclpa}
For all $P, Q \in \ST$,
\begin{equation*}
\FSCLT \vDash P = Q \ \Longleftrightarrow\ \FSCL \vDash P = Q.
\end{equation*}
\end{theorem}
\begin{proof}
If $\FSCLT \vDash P = Q$, then by Theorem \ref{thm:sclcpl} we have $\EqFSCL
\vdash P = Q$. So it suffices to show that $\EqFSCL$ is sound with respect to
$\CP_s$, i.e., that $\EqFSCL \vdash P = Q$ implies $\CP_s \vdash P = Q$. This
proof can be found in \cite{scl}.

For the other direction we must show that if $\CP_s \vdash P = Q$, then $\FSCLT
\vDash P = Q$. Because $\E_s$ restricted to $\ST$ equals $\SE$, it suffices to
show that $\CP_s$ is sound with respect to $\E_s$-equality, i.e., that for all
$R, S \in \CPT_s$, $\CP_s \vdash R = S$ implies $\E_s(R) = \E_s(S)$. This is
again a trivial proof. For example, we show it for \eqref{ax:cp3} as follows:
\begin{align*}
\E_s(\true \lef R \rig \false)
&= \E_s(R)\ssub{\true}{\E_s(\true)}{\false}{\E_s(\false)} \\
&= \E_s(R)\ssub{\true}{\true}{\false}{\false} \\
&= \E_s(R),
\end{align*}
and for \eqref{eq:sor} as
\begin{align*}
\E_s(R \sleftor S)
&= \E_s(R)\sub{\false}{\E_s(S)} \\
&= \E_s(R)\ssub{\true}{\true}{\false}{\E_s(S)} \\
&= \E_s(\true \lef R \rig S). &&\qedhere
\end{align*}
\end{proof}

The reader may wonder why in this thesis we presented the completeness of
$\EqFSCL$ with respect to $\FSCLT$ rather than $\FSCL$. The reason for this is
that the author discovered the result after proving the completeness of
$\EqFFEL$ with respect to $\FFEL$, and this presentation emphasizes the
similarities and differences of that proof with the proof for $\FSCL$.

\section{FFEL is a sublogic of FSCL}
\label{sec:ffelfscl}
When we consider a simple $\FEL$-term such as $a \leftand b$ and picture $\FE(a
\leftand b)$, we see that we can reconstruct the normal form of the original
term as $\true \leftand ((a \leftand \true) \leftand (b \leftand \true))$.
However, we can also reconstruct this tree as $\true \sleftand ((a \sleftand
\true) \sleftor ((b \sleftor \false) \sleftand \false)) \sleftand ((b \sleftand
\true) \sleftor \false))$. We will indeed show that for any $\FEL$-term $P$
there is an $\SCL$-term $Q$ such that $\FE(P) = \SE(Q)$. To this end we define
a translation function $\trans$, which translates $\FEL$-terms to $\SCL$-terms
with the same evaluation tree semantics as follows.
\begin{align}
\trans(\true) &= \true \\
\trans(\false) &= \false \\
\trans(a) &= a \\
\trans(\neg P) &= \neg \trans(P) \\
\trans(P \leftand Q) &= (\trans(P) \sleftor (\trans(Q) \sleftand \false))
  \sleftand \trans(Q) \\
\trans(P \leftor Q) &= (\trans(P) \sleftand (\trans(Q) \sleftor \true))
  \sleftor \trans(Q)
\end{align}

We immediately turn to the proof that $\trans$ has the desired property.
\begin{theorem}
$\FFEL$ is a sublogic of $\FSCL$, i.e., for all $P, Q \in FT$,
\begin{equation*}
\FFEL \vDash P = Q \ \Longrightarrow\  \FSCL \vDash \trans(P) = \trans(Q).
\end{equation*}
\end{theorem}
\begin{proof}
It suffices by Theorem \ref{thm:fsclpa} to show that for any $\FEL$-term $P$,
$\FE(P) = \SE(\trans(P))$, which we shall do by induction on $P$. The base
cases are trivial. For the induction we have
\begin{align*}
\FE(P \leftand Q)
&= \FE(P)\ssub{\true}{\FE(Q)}{\false}{\FE(Q)\sub{\true}{\false}} \\
&= \SE(\trans(P))\ssub{\true}{\SE(\trans(Q))}{\false}{\SE(\trans(Q))
  \sub{\true}{\false}}
  &\textrm{by induction hypothesis} \\
&= \SE(\trans(P))\sub{\false}{\SE(\trans(Q))\sub{\true}{\false}}
  \sub{\true}{\SE(\trans(Q))} \\
&= \SE((\trans(P) \sleftor (\trans(Q) \sleftand \false)) \sleftand \trans(Q)) \\
&= \SE(\trans(P \leftand Q)),
\end{align*}
where the third equality follows from the fact that $\SE(\trans(Q))
\sub{\true}{\false}$ does not contain $\true$. The case for $P \leftor Q$ is
similar.
\end{proof}

Let us consider the translation of $\FEL$-terms in $P^\true$. In the base case
we have that $\trans(\true) = \true$ and in the inductive case we have
$\trans(a \leftor P) = (a \sleftand (\trans(P) \sleftor \true)) \sleftor
\trans(P)$. By Lemma \ref{lem:ptpfs} and the induction hypothesis this is
equal to $(a \sleftand \trans(P)) \sleftor \trans(P)$. In other words,
$\FEL$-terms in $P^\true$ are equivalent to $\SCL$-terms in $P^\true$.
Similarly, $\FEL$-terms in $P^\false$ are translated to $\SCL$-terms in
$P^\false$. In both cases they are equivalent with $\true$-terms
($\false$-terms) $P$ for which the paths of $\SE(P)$ all contain the same atoms
\emph{in the same order}.

In Chapter \ref{chap:intro} we promised to define a general left-sequential
logic, i.e., a logic for reasoning about propositional terms that contain both
short-circuit left-sequential connectives and full left-sequential connectives.
We can now easily define such a logic by adding the following two equations
to $\EqFSCL$:
\begin{equation*}
x \leftand y = (x \sleftor (y \sleftand \false)) \sleftand y
\quad\textrm{and}\quad
x \leftor y = (x \sleftand (y \sleftor \true)) \sleftor y.
\end{equation*}
By the results from Chapter \ref{chap:scl} and this chapter, it is immediate
that this set of equations axiomatizes a (free) general left-sequential logic.
Naturally, we could also express this logic in terms of Proposition Algebra by
adding both types of connectives to its signature and considering $\CP$
together with \eqref{eq:sneg}--\eqref{eq:sor}, \eqref{eq:fand} and
\eqref{eq:for}. Without making any assumptions about the side effects that
atoms may have, this logic can be used to reason about propositional terms in
programming languages which offer both types of connectives, such as Java.

In the next chapter we will discuss our motivations for examining $\FEL$
separately from $\SCL$.

\chapter{Conclusion and Outlook}
\label{chap:concl}

The evaluation strategy prescribed by a propositional-based logic is key to
determining the semantics of propositional terms as they are used in
programming languages to direct the flow of a program. For any such evaluation
strategy to be of use to programmers, it must be deterministic. For suppose
otherwise, then two evaluations of the same term in the same execution
environment could yield different results and such a system can hardly be
called a logic. Given that these evaluation strategies must be deterministic,
we are left with considering sequential evaluation strategies and parallel
evaluation strategies.


We have not found any programming language that uses a parallel evaluation
strategy when dealing with propositional terms for program flow control. The
likely reason being that to truly evaluate several subterms in parallel, the
state of the entire execution environment prior to the start of the evaluation
must be copied so that each subterm can be evaluated in the same environment.
This would likely cause the evaluation to be slowed down to such an extent as
to render it useless in practice. The `merging' of multiple copies of the
environment after the subterms have been evaluated would also be a highly
non-trivial exercise. Perhaps in the setting of Quantum Computing we can
imagine an evaluation in superposition resulting in a superposition of
environments, but it is not at all clear how we should interpret the
superposition of the yields of the individual atoms.


Therefore we focus on sequential evaluation strategies, which are widely used
in programming languages. We focus entirely on left-sequential evaluation
strategies, since most programming languages read from left to right. We have
examined both short-circuit evaluation strategies, in the form of $\FSCLT$,
and full evaluation strategies, in the form of $\FFEL$. Both induce
right-sequential evaluation strategies. For example, for FEL we could introduce
the symbols `$\rightand$' for full right-sequential conjunction and
`$\rightor$' for full right-sequential disjunction. The equations in
$\EqFFEL$ could then easily be rewritten to accommodate the new direction,
e.g., \eqref{ax:ft9} would become $y \rightor (\false \rightand x) = y
\rightand (\true \rightor x)$. For examining a setting with both
left-sequential and right-sequential connectives, we would naturally define the
right-sequential connectives in terms of their left-sequential counterparts.


In ordinary propositional logic no particular evaluation strategy is
prescribed. \textsc{SAT}-solvers make eager use of this freedom and often
employ complex evaluation strategies that go far beyond simple left-sequential
evaluation. See, e.g., \cite{Gu96} for a survey of some of the different
algorithms used for satisfiability solving. We emphasize that
\textsc{SAT}-solvers deal with propositional terms whose atoms do not have side
effects. Both $\FFEL$ and $\FSCLT$ are designed to deal with atoms that
\emph{do} have side effects, to which can be contributed much of the complexity
of these logics.


As promised in Chapter \ref{chap:intro} we return to our claim about the
applicability of left-sequential logics for reasoning about side effects. The
yield of any occurrence of an atom $a \in A$ in some term $P$ can be influenced
by the side effects of the atoms that precede the occurrence of $a$ in $P$ as
well as the state of the execution environment in which $P$ is evaluated. If an
atom $a \in A$ does not have a side effect, then it always behaves either as
the constant $\true$ or as the constant $\false$ depending on the atoms that
were evaluated before it and the state of the execution environment. For $a \in
A$ and $P \in \ST$ let $[\true / a]P$ denote the term which results from
replacing each occurrence of $a$ in $P$ by $\true$.  Similarly, let $[\false /
a]P$ be the term that results from replacing each occurrence of $a$ in $P$ with
$\false$. Let $y_e$ be the function that returns the boolean yield of an
$\SCL$-term when it is evaluated in execution environment $e$. An atom $a \in
A$ has a side effect if there is some execution environment $e$ and there are
$P, Q \in \ST$ with $y_e(P) = y_e(Q)$ such that either
\begin{equation*}
y_e([\true / a]P) \neq y_e([\true / a]Q
\quad\text{or}\quad
y_e([\false / a]P) \neq y_e([\false / a]Q).
\end{equation*}


As an example consider atoms $a$ and $b$ and suppose that a side effect of $a$
is that any evaluation of $b$ that follows it will yield $\true$. Also suppose
that if $b$ were not preceded by $a$ it would yield $\false$. To make this
concrete we could imagine $a$ being a method that sets some global variable in
the execution environment and always yields $\true$. We could then see $b$ as
being a method that checks whether that variable has been set, in which case it
yields $\true$, or not, in which case it yields $\false$. Letting $e$ be some
execution environment where the global variable is not set, or alternatively
the empty execution environment, we see that $y_e(a \leftand b) = \true =
y_e(\neg b)$ and that $y_e(\true \leftand b) = \false \neq y_e(\neg b)$. Hence
$a$ has a side effect by our definition.


This opens the door to more involved reasoning about side effects. The example
above hints at the possibility of defining what it means for an atom to be
`impacted' by the side effect of a single other atom. We could, for example,
restrict our attention to $P, Q \in \ST$ containing only $a$ and some other
atom $b$. If there is an environment $e$ and there are $P, Q \in \ST$ that
contain only $a$ and $b$ such that $y_e(P) = y_e(Q)$, but $y_e([\true / a]P)
\neq y_e([\true / a]Q$ or $y_e([\false / a]P) \neq y_e([\false / a]Q$, then we
know that $b$ is impacted by a side effect of $a$. Another interesting
definition would be that of a `positive side effect'. In that case we could say
that an atom $a$ has a positive side effect if there is some environment
$e$ and some $P, Q \in \ST$ such that $y_e(P) = y_e(Q)$, but $y_e([\true / a]P)
\neq y_e([\true / a]Q)$.


With this definition of a side effect we see that $\FEL$, unlike PL, preserves
side effects in the sense that $\FFEL \vDash P = Q$ implies $\FFEL \vDash
[\true / a]P = [\true / a]Q$ and $\FFEL \vDash [\false / a]P = [\false / a]Q$
for all $P, Q \in \FT$ and all $a \in A$. The same goes for $\SCL$.  Thus, if
we adopt our proposed definition of side effects, both $\FEL$ and $\SCL$ can be
used to reason about propositional expressions with atoms that may have side
effects.


For defining and reasoning about side effects as we have done above we need
constants for our truth values. The constants $\true$ and $\false$ are not
definable in $\SCL$ and $\FEL$, except in terms of one another. This is unlike
PL, where the law of excluded middle allows one to define $\true$ in terms of
negation and disjunction, naturally assuming we have at least one atom at our
disposal. The constant $\false$ can then be defined in terms of $\true$ or by
the law of non-contradiction. With $\SCL$ and $\FEL$ we cannot exclude the
possibility that every atom has a side effect that causes any subsequent
evaluation of that same atom to yield the opposite truth value. Hence neither
the law of excluded middle nor the law of non-contradiction are valid in $\SCL$
or $\FEL$.


We now turn to the question of the usefulness of $\FEL$. As mentioned in
Chapter \ref{chap:intro} we find that some programming languages offer full
left-sequential connectives, which motivates the initial investigation of
$\FEL$. We claim that $\FEL$ has a greater value than merely to act as means of
writing certain $\SCL$ terms using fewer symbols. The usefulness of a full
evaluation strategy lies in the increased predictability of the state of the
environment after a (sub)term has been evaluated. In particular, we know that
the side effects of all the atoms in the term have occurred. Naturally we leave
errors and error values out of this discussion.


$\SCL$ is characterized by its efficiency, in the sense that atoms are not
evaluated if their yield is not needed to determine the yield of a term as a
whole. From that perspective $\FEL$ might seem rather inefficient, but this is
not necessarily so. To determine the state of the environment after the
evaluation of a $\FEL$ term in a given environment, we need only compute how
each atom in the term transforms the environment. It is \emph{not} necessary to
compute the yield of any of the atoms. With $\SCL$ terms in general we must
know what the first atom yields in order to determine which atom is next to
transform the environment. Thus to compute the state of the environment after
the evaluation of an $\SCL$ term we must compute the yield of each atom that
transforms the environment \emph{and} we must compute the transformation of the
environment for each atom that affects it. Consider the $\SCL$-term
\begin{equation*}
((a \sleftor (b \sleftand \false)) \sleftand \false) \sleftor c
\end{equation*}
and note that to compute its yield we must first compute the yield of $a$ to
determine whether or not $b$ is short-circuited. Consider atoms $a$ and $b$
that have no side effects, and hence do not affect the environment. If
computing the yield of $a$ is computationally very demanding, the $\FEL$-term
\begin{equation*}
((a \leftor (b \leftand \false)) \leftand \false) \leftor c
\end{equation*}
can be evaluated more quickly, because it is not necessary to compute the yield
of $a$.


In \cite{scl} several variants of $\SCL$ are defined in addition to $\FSCL$. In
this thesis we have only defined one variant of $\FEL$, i.e., $\FFEL$. An
important variant of $\SCL$ is Memorizing Short-Circuit Logic, MSCL, which is
defined in the same way as $\FSCL$, but adding the following axiom to $\CP$:
\begin{equation}
x \lef y \rig (z \lef u \rig (v \lef y \rig w)) =
  x \lef y \rig (z \lef u \rig w)
  \label{ax:cpmem}\tag{CPmem}
\end{equation}
As we can see from \eqref{ax:cpmem}, once an atom has been evaluated all
subsequent evaluations of the same atom will yield the same truth value. An
example of such a `memorizing' atom in programming might be a call to a
memoizing function\footnote{A memoizing function is a function which maintains
a cache of function values for arguments it has previously been called with.
See, e.g., \url{http://en.wikipedia.org/wiki/Memoization} for a detailed
description of memoization.} with a fixed argument. Naturally we could define
Memorizing Fully Evaluated Left-Sequential Logic, MFEL, in a similar fashion.
Given our evaluation tree semantics however, we can also define a
`post-processing' on our trees instead. We simply take the $\FE$ image of a
term and recursively walk down the tree. Whenever we encounter $X \tlef a \trig
Y$ in a left subtree of an $a$, we replace it by $X \tlef a \trig X$.
Similarly, we replace it by $Y \tlef a \trig Y$ if we are in the right
subtree of an $a$.


Another variant of $\SCL$ is Static $\SCL$, SSCL, which is defined in
\cite{scl} in the same way as $\FSCL$, but adding \eqref{ax:cpmem} and the
following equation to $\CP$.
\begin{equation*}
\false \lef x \rig \false = \false
\end{equation*}
This equation implies that $x \sleftand \false = \false$, and more generally,
$x \sleftand y = y \sleftand x$. As shown in \cite{scl}, this variant is the
same as PL, except that a particular evaluation strategy is prescribed.
Naturally we could define Static $\FEL$, SFEL, similarly. We believe that the
most elegant method of defining variants, other than the free variants, for
$\FEL$ and $\SCL$ is by means of Proposition Algebra. We believe that
evaluation trees offer a didactically interesting alternative definition for
the free variants, because they offer a straightforward semantics just for the
left-sequential connectives.


When considering these and other variants of $\FEL$ and $\SCL$ it is useful to
consider what these logics express in terms of the properties of atoms. Any
atom in MFEL (or MSCL) is memorizing in the sense that its yield becomes
constant after it is first evaluated. The atoms in SFEL (or SSCL) have no
side effects according to our definition. For practical applications of the
theory of left-sequential logics it may be useful to partition the set $A$ of
atoms into sets of atoms possessing certain of these properties. A compound
logic, geared towards the optimization of propositional statements in
programming, could then be defined. In such a logic, for example, we would have
$x \sleftand \false = \false$ for any atom in the `static' partition of $A$.
The potential for optimization, i.e., evaluating as few atoms as possible to
compute the yield of a term, becomes even greater when we consider the variants
of Contractive $\SCL$ and Repetition-Proof $\SCL$, see \cite{scl}.


In \cite{regenboog} Regenboog showed that $\CP$ is $\omega$-complete if and
only if, for $A$ a countable set of atoms, $\weight{A} > 1$. For an
axiomatization of Static $\CP$ which is interderivable with the one presented
above he showed $\omega$-completeness for any countable set of atoms. He also
showed that the axioms in $\CP$ and several axiomatizations extending it are
independent. We have shown neither $\omega$-completeness nor independence for
$\EqFFEL$, although we would consider such theorems valuable future work. It is
also an open question whether $\EqFSCL$ is $\omega$-complete or independent.
The set $\EqFSCL$ as presented in \cite{scl} and \cite{pascl} is a different
set from the one we have introduced in Chapter \ref{chap:scl}, although they
are interderivable. The set we presented came about in consultation with the
authors of the original definition. Because the set $\EqFSCL$ is somewhat `in
flux' in this sense, and to a lesser degree because its independence and
$\omega$-completeness are still open questions, we have refrained from
referring to the equations in $\EqFSCL$ as axioms.  Similarly, our definition
of $\EqFFEL$ differs slightly from that in \cite{sel}, hence we do not refer to
those equations as axioms either. Formal definitions of variants of $\FEL$
other than $\FFEL$ and a comparative analysis of these and the corresponding
variants of $\SCL$ we also consider a great avenue for further study.


\section*{Acknowledgments}
Firstly I would like to thank my thesis supervisor, Alban Ponse, for the many
hours he spent reviewing drafts and providing valuable feedback. His experience
has proved a great help in cleaning up the notation used in this work as well
as its presentation in general. If it were not for his confidence in an earlier
version of the completeness proof for $\FEL$, I might never have discovered the
completeness proof for $\SCL$ and improved the presentation of the $\FEL$ proof
significantly in the process.

I also wish to thank Alwin Blok for the work he has done on Fully Evaluated
Left-Sequential Logic, in particular his discovery of the axiomatization for
Free $\FEL$ and an early version of a $\FEL$ Normal Form, together with its
proof of correctness. Our discussions about the motivations for that normal
form and the semantic structures of evaluation trees have been of great help.

Finally, I thank my family for the moral support they gave me while writing
this thesis.


%
%
\appendix
\chapter{Proofs for FFEL}
\label{chap:felproofs}

%
%
%
\section{Correctness of \nftitle}
\label{sec:lslnf}
To prove that $\nf: \FT \to \FNF$ is indeed a normalization function we need to
prove that for all $\FEL$-terms $P$, $\nf(P)$ terminates, $\nf(P) \in \FNF$ and
$\EqFFEL \vdash \nf(P) = P$. To arrive at this result, we prove several
intermediate results about the functions $\nf^n$ and $\nf^c$, roughly in the
order in which their definitions were presented in Section \ref{sec:felnf}. For
the sake of brevity we will not explicitly prove that these functions
terminate. To see that each function terminates consider that a termination
proof would closely mimic the proof structure of the lemmas dealing with the
grammatical categories of the images of these functions.

\begin{lemma}
\label{lem:ptpf}
For any $P^\false$ and $P^\true$, $\EqFFEL \vdash P^\false = P^\false \leftand
\false$ and $\EqFFEL \vdash P^\true = P^\true \leftor \true$.
\end{lemma}
\begin{proof}
We prove both claims simultaneously by induction. In the base case we have
$\false = \true \leftand \false$ by \eqref{ax:ft4}, which is equal to $\false
\leftand \false$ by \eqref{ax:ft8} and \eqref{ax:ft1}. The base case for the
second claim follows from that for the first claim by duality.

For the induction we have $a \leftand P^\false = a \leftand (P^\false \leftand
\false)$ by the induction hypothesis and the result follows from
\eqref{ax:ft7}. For the second claim we again appeal to duality.
\end{proof}

\begin{lemma}
\label{lem:feqs2}
The following equations can be derived by equational logic and $\EqFFEL$.
\begin{enumerate}[itemsep=5pt]
\item $x \leftand (y \leftand (z \leftand \false)) = (x \leftor y) \leftand (z
  \leftand \false)$
  \label{eq:a3}
\item $\neg x \leftand (y \leftor \true) = \neg (x \leftand (y \leftor \true))$
  \label{eq:a4}
\end{enumerate}
\end{lemma}
\begin{proof}
\begin{align*}
x \leftand (y \leftand (z \leftand \false))
&= x \leftand ((\neg y \leftand z) \leftand \false)
  &\textrm{by \eqref{ax:ft7} and \ref{lem:feqs} \eqref{eq:a1}} \\
&= (\neg x \leftand \neg y) \leftand (z \leftand \false)
  &\textrm{by \eqref{ax:ft7} and \ref{lem:feqs} \eqref{eq:a1}} \\
&= \neg(\neg x \leftand \neg y) \leftand (z \leftand \false)
  &\textrm{by Lemma \ref{lem:feqs} \eqref{eq:a1}} \\
&= (x \leftor y) \leftand (z \leftand \false)
  &\textrm{by \eqref{ax:ft2}} \\[5pt]
\neg x \leftand (y \leftor \true)
&= \neg x \leftor (y \leftand \false)
  &\textrm{by \eqref{ax:ft10}} \\
&= \neg (x \leftand \neg(y \leftand \false))
  &\textrm{by \eqref{ax:ft2} and \eqref{ax:ft3}} \\
&= \neg (x \leftand \neg(\neg y \leftand \neg \true))
  &\textrm{by \eqref{ax:ft8} and \eqref{ax:ft1}} \\
&= \neg (x \leftand (y \leftor \true))
  &\textrm{by \eqref{ax:ft2}} &\qedhere
\end{align*}
\end{proof}

\begin{lemma}
\label{lem:nfn}
For all $P \in \FNF$, if $P$ is a
$\true$-term then $\nf^n(P)$ is an $\false$-term, if it is an $\false$-term
then $\nf^n(P)$ is a $\true$-term, if it is a $\true$-$*$-term then so is
$\nf^n(P)$, and
\begin{equation*}
\EqFFEL \vdash \nf^n(P) = \neg P.
\end{equation*}
\end{lemma}
\begin{proof}
We start with proving the claims for $\true$-terms, by induction on $P^\true$.
In the base case $\nf^n(\true) = \false$. It is immediate that $\nf^n(\true)$
is an $\false$-term. The claim that $\EqFFEL \vdash \nf^n(\true) = \neg \true$
is immediate by \eqref{ax:ft1}. For the inductive case we have that $\nf^n(a
\leftor P^\true) = a \leftand \nf^n(P^\true)$, where we may assume that
$\nf^n(P^\true)$ is an $\false$-term and that $\EqFFEL \vdash \nf^n(P^\true) =
\neg P^\true$. The grammatical claim now follows immediately from the induction
hypothesis. Furthermore, noting that by the induction hypothesis we may assume
that $\nf^n(P^\true)$ is an $\false$-term, we have:
\begin{align*}
\nf^n(a \leftor P^\true)
&= a \leftand \nf^n(P^\true)
  &\textrm{by definition} \\
&= a \leftand (\nf^n(P^\true) \leftand \false)
  &\textrm{by Lemma \ref{lem:ptpf}} \\
&= \neg a \leftand (\nf^n(P^\true) \leftand \false)
  &\textrm{by Lemma \ref{lem:feqs} \eqref{eq:a1}} \\
&= \neg a \leftand \nf^n(P^\true)
  &\textrm{by Lemma \ref{lem:ptpf}} \\
&= \neg a \leftand \neg P^\true
  &\textrm{by induction hypothesis} \\
&= \neg (a \leftor P^\true).
  &\textrm{by \eqref{ax:ft3} and \eqref{ax:ft2}}
\end{align*}

For $\false$-terms we prove our claims by induction on $P^\false$. In the base
case $\nf^n(\false) = \true$. It is immediate that $\nf^n(\false)$ is a
$\true$-term. The claim that $\EqFFEL \vdash \nf^n(\false) = \neg \false$ is
immediate by the dual of \eqref{ax:ft1}. For the inductive case we have that
$\nf^n(a \leftand P^\false) = a \leftor \nf^n(P^\false)$, where we may assume
that $\nf^n(P^\false)$ is a $\true$-term and $\EqFFEL \vdash \nf^n(P^\false) =
\neg P^\false$. It follows immediately from the induction hypothesis that
$\nf^n(a \leftand P^\false)$ is a $\true$-term.  Furthermore, noting that by
the induction hypothesis we may assume that $\nf^n(P^\false)$ is a
$\true$-term, we prove the remaining claim as follows:
\begin{align*}
\nf^n(a \leftand P^\false)
&= a \leftor \nf^n(P^\false)
  &\textrm{by definition} \\
&= a \leftor (\nf^n(P^\false) \leftor \true)
  &\textrm{by Lemma \ref{lem:ptpf}} \\
&= \neg a \leftor (\nf^n(P^\false) \leftor \true)
  &\textrm{by the dual of Lemma \ref{lem:feqs} \eqref{eq:a1}} \\
&= \neg a \leftor \nf^n(P^\false)
  &\textrm{by Lemma \ref{lem:ptpf}} \\
&= \neg a \leftor \neg P^\false
  &\textrm{by induction hypothesis} \\
&= \neg (a \leftand P^\false).
  &\textrm{by \eqref{ax:ft3} and \eqref{ax:ft2}}
\end{align*}

To prove the lemma for $\true$-$*$-terms we first verify that the auxiliary
function $\nf^n_1$ returns a $*$-term and that for any $*$-term $P$, $\EqFFEL
\vdash \nf^n_1(P) = \neg P$. We show this by induction on the number of
$\ell$-terms in $P$. For the base cases, i.e., for $\ell$-terms, it is
immediate that $\nf^n_1(P)$ is a $*$-term. If $P$ is an $\ell$-term with a
positive determinative atom we have:
\begin{align*}
\nf^n_1(a \leftand P^\true)
&= \neg a \leftand P^\true
  &\textrm{by definition} \\
&= \neg a \leftand (P^\true \leftor \true)
  &\textrm{by Lemma \ref{lem:ptpf}} \\
&= \neg (a \leftand (P^\true \leftor \true))
  &\textrm{by Lemma \ref{lem:feqs2} \eqref{eq:a4}} \\
&= \neg (a \leftand P^\true).
  &\textrm{by Lemma \ref{lem:ptpf}}
\end{align*}
If $P$ is an $\ell$-term with a negative determinative atom the proof proceeds
the same, substituting $\neg a$ for $a$ and applying \eqref{ax:ft3} where
needed. For the inductive step we assume that the result holds for $*$-terms
with fewer $\ell$-terms than $P^* \leftand Q^d$ and $P^* \leftor Q^c$. We
note that each application of $\nf^n_1$ changes the main connective (not
occurring inside an $\ell$-term) and hence the result is a $*$-term. Derivable
equality is, given the induction hypothesis, an instance of (the dual of)
\eqref{ax:ft2}.

With this result we can now see that $\nf^n(P^\true \leftand Q^*)$ is indeed a
$\true$-$*$-term. Furthermore we find that:
\begin{align*}
\nf^n(P^\true \leftand Q^*) 
&= P^\true \leftand \nf^n_1(Q^*)
  &\textrm{by definition} \\
&= P^\true \leftand \neg Q^*
  &\textrm{as shown above} \\
&= (P^\true \leftor \true) \leftand \neg Q^*
  &\textrm{by Lemma \ref{lem:ptpf}} \\
&= \neg(P^\true \leftor \true) \leftor \neg Q^*
  &\textrm{by Lemma \ref{lem:feqs} \eqref{eq:a5}} \\
&= \neg P^\true \leftor \neg Q^*
  &\textrm{by Lemma \ref{lem:ptpf}} \\
&= \neg (P^\true \leftand Q^*).
  &\textrm{by \eqref{ax:ft2} and \eqref{ax:ft3}}
\end{align*}
Hence for all $P \in \FNF$, $\EqFFEL \vdash \nf^n(P) = \neg P$.
\end{proof}

\begin{lemma}
\label{lem:nfc1}
For any $\true$-term $P$ and $Q \in \FNF$, $\nf^c(P, Q)$ has the same
grammatical category as $Q$ and 
\begin{equation*}
\EqFFEL \vdash \nf^c(P, Q) = P \leftand Q.
\end{equation*}
\end{lemma}
\begin{proof}
By induction on the complexity of the first argument. In the base case we see
that $\nf^c(\true, P) = P$ and hence has the same grammatical category as $P$.
Derivable equality follows from \eqref{ax:ft4}.

For the induction step we make a case distinction on the grammatical category
of the second argument. If the second argument is a $\true$-term we have that
$\nf^c(a \leftor P^\true, Q^\true) = a \leftor \nf^c(P^\true, Q^\true)$, where
we assume that $\nf^c(P^\true, Q^\true)$ is a $\true$-term and $\EqFFEL
\vdash \nf^c(P^\true, Q^\true) = P^\true \leftand Q^\true$. The grammatical
claim follows immediately from the induction hypothesis. The claim about
derivable equality is proved as follows:
\begin{align*}
\nf^c(a \leftor P^\true, Q^\true)
&= a \leftor \nf^c(P^\true, Q^\true)
  &\textrm{by definition} \\
&= a \leftor (P^\true \leftand Q^\true)
  &\textrm{by induction hypothesis} \\
&= a \leftor (P^\true \leftand (Q^\true \leftor \true))
  &\textrm{by Lemma \ref{lem:ptpf}} \\
&= (a \leftor P^\true) \leftand (Q^\true \leftor \true)
  &\textrm{by Lemma \ref{lem:feqs} \eqref{eq:a2}} \\
&= (a \leftor P^\true) \leftand Q^\true.
  &\textrm{by Lemma \ref{lem:ptpf}}
\end{align*}

If the second argument is an $\false$-term we assume that $\nf^c(P^\true,
Q^\false)$ is an $\false$-term and that $\EqFFEL \vdash \nf^c(P^\true,
Q^\false) = P^\true \leftand Q^\false$. The grammatical claim follows
immediately from the induction hypothesis. Derivable equality is proved as
follows:
\begin{align*}
\nf^c(a \leftor P^\true, Q^\false)
&= a \leftand \nf^c(P^\true, Q^\false)
  &\textrm{by definition} \\
&= a \leftand (P^\true \leftand Q^\false)
  &\textrm{by induction hypothesis} \\
&= a \leftand (P^\true \leftand (Q^\false \leftand \false))
  &\textrm{by Lemma \ref{lem:ptpf}} \\
&= (a \leftor P^\true) \leftand (Q^\false \leftand \false)
  &\textrm{by Lemma \ref{lem:feqs2} \eqref{eq:a3}} \\
&= (a \leftor P^\true) \leftand Q^\false.
  &\textrm{by Lemma \ref{lem:ptpf}} \\
\end{align*}

Finally, if the second argument is a $\true$-$*$-term then $\nf^c(a \leftor
P^\true, Q^\true \leftand R^*) = \nf^c(a \leftor P^\true, Q^\true) \leftand
R^*$. The fact that this is a $\true$-$*$-term follows from the fact that
$\nf^c(a \leftor P^\true, Q^\true)$ is a $\true$-term as was shown above.
Derivable equality follows from the case where the second argument is a
$\true$-term and \eqref{ax:ft7}.
\end{proof}

\begin{lemma}
\label{lem:nfc4}
For any $\true$-$*$-term $P$ and $\false$-term $Q$, $\nf^c(P, Q)$ is an
$\false$-term and
\begin{equation*}
\EqFFEL \vdash \nf^c(P, Q) = P \leftand Q.
\end{equation*}
\end{lemma}
\begin{proof}
By \eqref{ax:ft7} and Lemma \ref{lem:nfc1} it suffices to show that
$\nf^c_2(P^*, Q^\false)$ is an $\false$-term and that $\EqFFEL \vdash
\nf^c_2(P^*, Q^\false) = P^* \leftand Q^\false$. We prove this by induction on
the number of $\ell$-terms in $P^*$. In the base cases, i.e., $\ell$-terms, the
grammatical claims follow from Lemma \ref{lem:nfc1}. The claim about derivable
equality in the case of $\ell$-terms with positive determinative atoms follows
from Lemma \ref{lem:nfc1} and \eqref{ax:ft7}. For $\ell$-terms with negative
determinative atoms it follows from Lemma \ref{lem:nfc1}, Lemma \ref{lem:ptpf},
\eqref{ax:ft6}, \eqref{ax:ft7} and \eqref{ax:ft8}.

For the induction step we assume the claims hold for any $*$-terms with fewer
$\ell$-terms than $P^* \leftand Q^d$ and $P^* \leftor Q^c$. In the case of
conjunctions we have $\nf^c_2(P^* \leftand Q^d, R^\false) = \nf^c_2(P^*,
\nf^c_2(Q^d, R^\false))$ and the grammatical claim follows from the induction
hypothesis (applied twice). Derivable equality follows from the induction
hypothesis and \eqref{ax:ft7}.

For disjunctions we have $\nf^c_2(P^* \leftor Q^c, R^\false) = \nf^c_2(P^*,
\nf^c_2(Q^c, R^\false))$ and the grammatical claim follows from the induction
hypothesis (applied twice). The claim about derivable equality is proved as
follows:
\begin{align*}
\nf^c_2(P^* \leftor Q^c, R^\false)
&= \nf^c_2(P^*, \nf^c_2(Q^c, R^\false))
  &\textrm{by definition} \\
&= P^* \leftand (Q^c \leftand R^\false)
  &\textrm{by induction hypothesis} \\
&= P^* \leftand (Q^c \leftand (R^\false \leftand \false))
  &\textrm{by Lemma \ref{lem:ptpf}} \\
&= (P^* \leftor Q^c) \leftand (R^\false \leftand \false)
  &\textrm{by Lemma \ref{lem:feqs2} \eqref{eq:a3}} \\
&= (P^* \leftor Q^c) \leftand R^\false.
  &\textrm{by Lemma \ref{lem:ptpf}} &\qedhere
\end{align*}
\end{proof}

\begin{lemma}
\label{lem:nfc2}
For any $\false$-term $P$ and $Q \in \FNF$, $\nf^c(P, Q)$ is an
$\false$-term and
\begin{equation*}
\EqFFEL \vdash \nf^c(P, Q) = P \leftand Q.
\end{equation*}
\end{lemma}
\begin{proof}
We make a case distinction on the grammatical category of the second argument.
If the second argument is a $\true$-term we proceed by induction on the first
argument. In the base case we have $\nf^c(\false, P^\true) = \nf^n(P^\true)$
and the result is by Lemma \ref{lem:nfn}, Lemma \ref{lem:ptpf}, \eqref{ax:ft6}
and \eqref{ax:ft8}. In the inductive case we have $\nf^c(a \leftand P^\false,
Q^\true) = a \leftand \nf^c(P^\false, Q^\true)$, where we assume that
$\nf^c(P^\false, Q^\true)$ is an $\false$-term and $\EqFFEL \vdash
\nf^c(P^\false, Q^\true) = P^\false \leftand Q^\true$. The result now follows
from the induction hypothesis and \eqref{ax:ft7}.

If the second argument is an $\false$-term the proof is almost the same, except
that we need not invoke Lemma \ref{lem:nfn} or \eqref{ax:ft8} in the base case.

Finally, if the second argument is a $\true$-$*$-term we again proceed by
induction on the first argument. In the base case we have $\nf^c(\false,
P^\true \leftand Q^*) = \nf^c(P^\true \leftand Q^*, \false)$. The grammatical
claim now follows from Lemma \ref{lem:nfc4} and derivable equality follows
from Lemma \ref{lem:nfc4} and \eqref{ax:ft6}. For for the inductive case the
results follow from the induction hypothesis and \eqref{ax:ft7}.
\end{proof}

\begin{lemma}
\label{lem:nfc3}
For any $\true$-$*$-term $P$ and $\true$-term $Q$, $\nf^c(P, Q)$ has the same
grammatical category as $P$ and
\begin{equation*}
\EqFFEL \vdash \nf^c(P, Q) = P \leftand Q.
\end{equation*}
\end{lemma}
\begin{proof}
By \eqref{ax:ft7} it suffices to prove the claims for $\nf^c_1$, i.e., that
$\nf^c_1(P^*, Q^\true)$ has the same grammatical category as $P^*$ and that
$\EqFFEL \vdash \nf^c_1(P^*, Q^\true) = P^* \leftand Q^\true$. We prove this by
induction on the number of $\ell$-terms in $P^*$. In the base case we deal with
$\ell$-terms and the results follow from Lemma \ref{lem:nfc1} and
\eqref{ax:ft7}.

For the inductive cases we assume that the results hold for any $*$-term with
fewer $\ell$-terms than $P^* \leftand Q^d$ and $P^* \leftor Q^c$.  In the case
of conjunctions the results follow from the induction hypothesis and
\eqref{ax:ft7}. In the case of disjunctions the grammatical claim follows from
the induction hypothesis. For derivable equality we have:
\begin{align*}
\nf^c_1(P^* \leftor Q^c, R^\true)
&= P^* \leftor \nf^c_1(Q^c, R^\true)
  &\textrm{by definition} \\
&= P^* \leftor (Q^c \leftand R^\true)
  &\textrm{by induction hypothesis} \\
&= P^* \leftor (Q^c \leftand (R^\true \leftor \true))
  &\textrm{by Lemma \ref{lem:ptpf}} \\
&= (P^* \leftor Q^c) \leftand (R^\true \leftor \true)
  &\textrm{by Lemma \ref{lem:feqs} \eqref{eq:a2}} \\
&= (P^* \leftor Q^c) \leftand R^\true.
  &\textrm{by Lemma \ref{lem:ptpf}} &\qedhere
\end{align*}
\end{proof}

\begin{lemma}
\label{lem:nfc5}
For any $P, Q \in \FNF$, $\nf^c(P, Q)$ is in $\FNF$ and
\begin{equation*}
\EqFFEL \vdash \nf^c(P, Q) = P \leftand Q.
\end{equation*}
\end{lemma}
\begin{proof}
By the four preceding lemmas it suffices to show that $\nf^c(P^\true \leftand
Q^*, R^\true \leftand S^*)$ is in $\FNF$ and that $\EqFFEL \vdash \nf^c(P^\true
\leftand Q^*, R^\true \leftand S^*) = (P^\true \leftand Q^*) \leftand (R^\true
\leftand S^*)$. By \eqref{ax:ft7}, in turn, it suffices to prove that
$\nf^c_3(P^*, Q^\true \leftand R^*)$ is a $*$-term and that $\EqFFEL \vdash
\nf^c_3(P^*, Q^\true \leftand R^*) = P^* \leftand (Q^\true \leftand R^*)$. We
prove this by induction on the number of $\ell$-terms in $R^*$. In the base
case we have that $\nf^c_3(P^*, Q^\true \leftand R^\ell) = \nf^c_1(P^*,
Q^\true) \leftand R^\ell$. The results follow from Lemma \ref{lem:nfc3} and
\eqref{ax:ft7}.

For the inductive cases we assume that the results hold for all $*$-terms with
fewer $\ell$-terms than $R^* \leftand S^d$ and $R^* \leftor S^c$. For
conjunctions the result follows from the induction hypothesis and
\eqref{ax:ft7} and for disjunctions it follows from Lemma \ref{lem:nfc3} and
\eqref{ax:ft7}.
\end{proof}

\thmnfcorrect*
\begin{proof}
By induction on the complexity of $P$. If $P$ is an atom, the result is by
\eqref{ax:ft4} and \eqref{ax:ft5}. If $P$ is $\true$ or $\false$ the result is
by identity. For the induction we assume that the result holds for all
$\FEL$-terms of lesser complexity than $P \leftand Q$ and $P \leftor Q$. The
result now follows from the induction hypothesis, Lemma \ref{lem:nfn}, Lemma
\ref{lem:nfc5} and \eqref{ax:ft2}.
\end{proof}

\section{Correctness of \invtitle}
\label{sec:felinv}

\thminvcorrect*
\begin{proof}
We first prove that for all $\true$-terms $P$, $\inv^\true(\FE(P)) \equiv P$,
by induction on $P$. In the base case $P \equiv \true$ and we have
$\inv^\true(\FE(P)) \equiv \inv^\true(\true) \equiv \true \equiv P$. For the
inductive case we have $P \equiv a \leftor Q^\true$ and
\begin{align*}
\inv^\true(\FE(P)) &\equiv \inv^\true(\FE(Q^\true) \tlef a \trig
  \FE(Q^\true))
  &\textrm{by definition of $\FE$} \\
&\equiv a \leftor \inv^\true(\FE(Q^\true))
  &\textrm{by definition of $\inv^\true$} \\
&\equiv a \leftor Q^\true
  &\textrm{by induction hypothesis} \\
&\equiv P.
\end{align*}

Similarly, we see that for all $\false$-terms $P$, $\inv^\false(\FE(P)) \equiv
P$, by induction on $P$. In the base case $P \equiv \false$ and we have
$\inv^\false(\FE(P)) \equiv \inv^\false(\false) \equiv \false \equiv P$. For
the inductive case we have $P \equiv a \leftand Q^\false$ and
\begin{align*}
\inv^\false(\FE(P)) &\equiv \inv^\false(\FE(Q^\false) \tlef a \trig
  \FE(Q^\false))
  &\textrm{by definition of $\FE$} \\
&\equiv a \leftand \inv^\false(\FE(Q^\false))
  &\textrm{by definition of $\inv^\false$} \\
&\equiv a \leftand Q^\false
  &\textrm{by induction hypothesis} \\
&\equiv P.
\end{align*}

Now we check that for all $\ell$-terms $P$, $\inv^\ell(\FE(P)) \equiv P$.
We observe that either $P \equiv a \leftand Q^\true$ or $P \equiv \neg a
\leftand Q^\true$. In the first case we have
\begin{align*}
\inv^\ell(\FE(P)) &\equiv \inv^\ell(\FE(Q^\true) \tlef a
  \trig \FE(Q^\true)\sub{\true}{\false})
  &\textrm{by definition of $\FE$} \\
&\equiv a \leftand \inv^\true(\FE(Q^\true))
  &\textrm{by definition of $\inv^\ell$} \\
&\equiv a \leftand Q^\true
  &\textrm{as shown above} \\
&\equiv P.
\end{align*}
In the second case we have that
\begin{align*}
\inv^\ell(\FE(P)) &\equiv \inv^\ell(\FE(Q^\true)\sub{\true}{\false} \tlef a
  \trig \FE(Q^\true))
  &\textrm{by definition of $\FE$} \\
&\equiv \neg a \leftand \inv^\true(\FE(Q^\true))
  &\textrm{by definition of $\inv_\ell$} \\
&\equiv \neg a \leftand Q^\true
  &\textrm{as shown above} \\
&\equiv P.
\end{align*}

We now prove that for all $*$-terms $P$, $\inv^*(\FE(P)) \equiv P$, by
induction on $P$ modulo the complexity of $\ell$-terms. In the base case we are
dealing with $\ell$-terms. Because an $\ell$-term has neither a cd nor a dd we
have $\inv^*(\FE(P)) \equiv \inv^\ell(\FE(P)) \equiv P$, where the first
equality is by definition of $\inv^*$ and the second was shown above. For the
induction we have either $P \equiv Q \leftand R$ or $P \equiv Q \leftor R$. In
the first case note that by Theorem \ref{thm:cddd}, $\FE(P)$ has a cd and no
dd. So we have 
\begin{align*}
\inv^*(\FE(P)) &\equiv \inv^*(\cd_1(\FE(P))\ssub{\Box_1}{\true}{\Box_2}{\false})
  \leftand \inv^*(\cd_2(\FE(P)))
  &\textrm{by definition of $\inv^*$} \\
&\equiv \inv^*(\FE(Q)) \leftand \inv^*(\FE(R))
  &\textrm{by Theorem \ref{thm:cddd}} \\
&\equiv Q \leftand R
  &\textrm{by induction hypothesis} \\
&\equiv P.
\end{align*}
In the second case, again by Theorem \ref{thm:cddd}, $\FE(P)$ has a dd and no
cd. So we have that
\begin{align*}
\inv^*(\FE(P)) &\equiv \inv^*(\dd_1(\FE(P))\ssub{\Box_1}{\true}{\Box_2}{\false})
  \leftor \inv^*(\dd_2(\FE(P)))
  &\textrm{by definition of $\inv^*$} \\
&\equiv \inv^*(\FE(Q)) \leftor \inv^*(\FE(R))
  &\textrm{by Theorem \ref{thm:cddd}} \\
&\equiv Q \leftor R
  &\textrm{by induction hypothesis} \\
&\equiv P.
\end{align*}

Finally, we prove the theorem's statement by making a case distinction on the
grammatical category of $P$. If $P$ is a $\true$-term, then $\FE(P)$ has only
$\true$-leaves and hence $\inv(\FE(P)) \equiv \inv^\true(\FE(P)) \equiv P$,
where the first equality is by definition of $\inv$ and the second was shown
above. If $P$ is an $\false$-term, then $\FE(P)$ has only $\false$-leaves and
hence $\inv(\FE(P)) \equiv \inv^\false(\FE(P)) \equiv P$, where the first
equality is by definition of $\inv$ and the second was shown above. If $P$ is a
$\true$-$*$-term, then it has both $\true$ and $\false$-leaves and hence,
letting $P \equiv Q \leftand R$,
\begin{align*}
\inv(\FE(P)) &\equiv \inv^\true(\tsd_1(\FE(P))\sub{\Box}{\true}) \leftand
  \inv^*(\tsd_2(\FE(P)))
  &\textrm{by definition of $\inv$} \\
&\equiv \inv^\true(\FE(Q)) \leftand \inv^*(\FE(R))
  &\textrm{by Theorem \ref{thm:tsd}} \\
&\equiv Q \leftand R
  &\textrm{as shown above} \\
&\equiv P,
\end{align*}
which completes the proof.
\end{proof}

\chapter{Proofs for FSCL}
\label{chap:sclproofs}

%
%
%
\section{Correctness of \nfstitle}
\label{sec:sclnf}
In order to prove that $\nfs: \ST \to \SNF$ is indeed a normalization function
we need to prove that for all $\SCL$-terms $P$, $\nfs(P)$ terminates, $\nfs(P)
\in \SNF$ and $\EqFSCL \vdash \nfs(P) = P$. To arrive at this result, we prove
several intermediate results about the functions $\nfs^n$ and $\nfs^c$ in the
order in which their definitions were presented in Section \ref{sec:snf}. For
the sake of brevity we will not explicitly prove that these functions
terminate. To see that each function terminates consider that a termination
proof would closely mimic the proof structure of the lemmas dealing with the
grammatical categories of the images of these functions.

\begin{lemma}
\label{lem:ptpfs}
For any $P^\false$ and $P^\true$, $\EqFSCL \vdash P^\false = P^\false \sleftand
x$ and $\EqFSCL \vdash P^\true = P^\true \sleftor x$.
\end{lemma}
\begin{proof}
We prove both claims simultaneously by induction. In the base case we have
$\false = \false \sleftand x$ by \eqref{ax:scl6}. The base case for the
second claim follows from that for the first claim by duality.

For the induction we have $(a \sleftor P^\false) \sleftand Q^\false = (a
\sleftor P^\false) \sleftand (Q^\false \sleftand x)$ by the induction
hypothesis and the result follows from \eqref{ax:scl7}. For the second claim we
again appeal to duality.
\end{proof}

\begin{lemma}
\label{lem:seqs2}
The following equations can all be derived by equational logic and $\EqFSCL$.
\begin{enumerate}[itemsep=5pt]
\item $(x \sleftor \true) \sleftand \neg y = \neg((x \sleftor \true) \sleftand
  y)$
  \label{eq:b4}
\item $(x \sleftand (y \sleftand (z \sleftor \true))) \sleftor
  (w \sleftand (z \sleftor \true)) = ((x \sleftand y) \sleftor w) \sleftand
  (z \sleftor \true)$
  \label{eq:b5}
\item $(x \sleftor ((y \sleftor \true) \sleftand (z \sleftand \false)))
  \sleftand ((w \sleftor \true) \sleftand (z \sleftand \false)) =
  ((x \sleftand (w \sleftor \true)) \sleftor (y \sleftor \true)) \sleftand (z
  \sleftand \false)$
  \label{eq:b6}
\item $(x \sleftor ((y \sleftor \true) \sleftand (z \sleftand \false)))
  \sleftand (w \sleftand \false) = ((\neg x \sleftand (y \sleftor \true))
  \sleftor (w \sleftand \false)) \sleftand (z \sleftand \false)$
  \label{eq:b7}
\end{enumerate}
\end{lemma}
\begin{proof}
We derive the equations in order.
\begin{align*}
(&x \sleftor \true) \sleftand \neg y \\
&= \neg((\neg x \sleftand \false) \sleftor y)
  &\textrm{by \eqref{ax:scl1}, \eqref{ax:scl2} and \eqref{ax:scl3}} \\
&= \neg((x \sleftand \false) \sleftor y)
  &\textrm{by \eqref{ax:scl8s}} \\
&= \neg((x \sleftor \true) \sleftand y)
  &\textrm{by \eqref{ax:scl9s}} \displaybreak[0]\\[5pt]
&(x \sleftand (y \sleftand (z \sleftor \true))) \sleftor
  (w \sleftand (z \sleftor \true)) \\
&= ((x \sleftand y) \sleftand (z \sleftor \true)) \sleftor
  (w \sleftand (z \sleftor \true))
  &\textrm{by \eqref{ax:scl7}} \\
&= ((x \sleftand y) \sleftor w) \sleftand (z \sleftor \true)
  &\textrm{by the dual of \eqref{ax:scl10s}} \\[5pt]
(&x \sleftor ((y \sleftor \true) \sleftand (z \sleftand \false)))
  \sleftand ((w \sleftor \true) \sleftand (z \sleftand \false)) \\
&= (x \sleftor ((y \sleftand \false) \sleftor (z \sleftand \false)))
  \sleftand ((w \sleftor \true) \sleftand (z \sleftand \false)) 
  &\textrm{by \eqref{ax:scl9s}} \\
&= ((x \sleftor (y \sleftand \false)) \sleftor (z \sleftand \false))
  \sleftand ((w \sleftor \true) \sleftand (z \sleftand \false)) 
  &\textrm{by the dual of \eqref{ax:scl7}} \\
&= (\neg(x \sleftor (y \sleftand \false)) \sleftor (w \sleftor \true))
  \sleftand (z \sleftand \false)
  &\textrm{by Lemma \ref{lem:seqs} \eqref{eq:b1}} \\
&= ((\neg x \sleftand (\neg y \sleftor \true)) \sleftor (w \sleftor \true))
  \sleftand (z \sleftand \false)
  &\textrm{by \eqref{ax:scl1}, \eqref{ax:scl2} and \eqref{ax:scl3}} \\
&= ((\neg x \sleftand (y \sleftor \true)) \sleftor (w \sleftor \true))
  \sleftand (z \sleftand \false)
  &\textrm{by the dual of \eqref{ax:scl8s}} \\
&= ((\neg x \sleftand (y \sleftor \true)) \sleftor (w \sleftor (\true
  \sleftor (y \sleftor \true)))) \sleftand (z \sleftand \false)
  &\textrm{by the dual of \eqref{ax:scl6}} \\
&= ((\neg x \sleftand (y \sleftor \true)) \sleftor ((w \sleftor \true)
  \sleftor (y \sleftor \true))) \sleftand (z \sleftand \false)
  &\textrm{by the dual of \eqref{ax:scl7}} \\
&= ((x \sleftand (w \sleftor \true)) \sleftor (y \sleftor \true))
  \sleftand (z \sleftand \false)
  &\textrm{by the dual of Lemma \ref{lem:seqs} \eqref{eq:b1}} \\[5pt]
(&x \sleftor ((y \sleftor \true) \sleftand (z \sleftand
  \false))) \sleftand (w \sleftand \false) \\
&= (\neg x \sleftor (w \sleftand \false)) \sleftand (((y \sleftor \true)
  \sleftand (z \sleftand \false)) \sleftand (w \sleftand \false))
  &\textrm{by Lemma \ref{lem:seqs} \eqref{eq:b1}} \\
&= (\neg x \sleftor (w \sleftand \false)) \sleftand ((y \sleftor \true)
  \sleftand (z \sleftand \false))
  &\textrm{by \eqref{ax:scl6} and \eqref{ax:scl7}} \\
&= ((\neg x \sleftor (w \sleftand \false)) \sleftand (y \sleftor \true))
  \sleftand (z \sleftand \false)
  &\textrm{by \eqref{ax:scl7}} \\
&= ((\neg x \sleftand (y \sleftor \true)) \sleftor ((w \sleftand \false)
  \sleftand (y \sleftor \true))) \sleftand (z \sleftand \false)
  &\textrm{by the dual of \eqref{ax:scl10s}} \\
&= ((\neg x \sleftand (y \sleftor \true)) \sleftor (w \sleftand \false))
  \sleftand (z \sleftand \false)
  &\textrm{by \eqref{ax:scl6} and \eqref{ax:scl7}}
  &\qedhere
\end{align*}
\end{proof}

\begin{lemma}
\label{lem:nfsn}
For all $P \in \SNF$, if $P$ is a $\true$-term then $\nfs^n(P)$ is an
$\false$-term, if it is an $\false$-term then $\nfs^n(P)$ is a $\true$-term, if
it is a $\true$-$*$-term then so is $\nfs^n(P)$, and
\begin{equation*}
\EqFSCL \vdash \nfs^n(P) = \neg P.
\end{equation*}
\end{lemma}
\begin{proof}
We first prove the claims for $\true$-terms, by induction on $P^\true$.  In the
base case $\nfs^n(\true) = \false$. It is immediate that $\nfs^n(\true)$ is an
$\false$-term. The claim that $\EqFSCL \vdash \nfs^n(\true) = \neg \true$ is
immediate by \eqref{ax:scl1}. For the inductive case we have that $\nfs^n((a
\sleftand P^\true) \sleftor Q^\true) = (a \sleftor \nfs^n(Q^\true)) \sleftand
\nfs^n(P^\true)$, where we assume that $\nfs^n(P^\true)$ and $\nfs^n(Q^\true)$
are $\false$-terms and that $\EqFSCL \vdash \nfs^n(P^\true) = \neg P^\true$
and $\EqFSCL \vdash \nfs^n(Q^\true) = \neg Q^\true$. It follows from the
induction hypothesis that $\nfs^n((a \sleftand P^\true) \sleftor Q^\true)$ is
an $\false$-term. Furthermore, noting that by the induction hypothesis we may
assume that $\nfs^n(P^\true)$ and $\nfs^n(Q^\true)$ are $\false$-terms, we
have:
\begin{align*}
\nfs^n((a \sleftand P^\true) \sleftor Q^\true)
&= (a \sleftor \nfs^n(Q^\true)) \sleftand \nfs^n(P^\true)
  &\textrm{by definition} \\
&= (a \sleftor (\nfs^n(Q^\true) \sleftand \false)) \sleftand (\nfs^n(P^\true)
  \sleftand \false)
  &\textrm{by Lemma \ref{lem:ptpfs}} \\
&= (\neg a \sleftor (\nfs^n(P^\true) \sleftand \false)) \sleftand
  (\nfs^n(Q^\true) \sleftand \false )
  &\textrm{by Lemma \ref{lem:seqs} \eqref{eq:b2}} \\
&= (\neg a \sleftor \nfs^n(P^\true)) \sleftand \nfs^n(Q^\true)
  &\textrm{by Lemma \ref{lem:ptpfs}} \\
&= (\neg a \sleftor \neg P^\true) \sleftand \neg Q^\true
  &\textrm{by induction hypothesis} \\
&= \neg((a \sleftand P^\true) \sleftor Q^\true).
  &\textrm{by \eqref{ax:scl2} and its dual}
\end{align*}

For $\false$-terms we prove our claims by induction on $P^\false$. In the base
case $\nfs^n(\false) = \true$. It is immediate that $\nfs^n(\false)$ is a
$\true$-term. The claim that $\EqFSCL \vdash \nfs^n(\false) = \neg \false$ is
immediate by the dual of \eqref{ax:scl1}. For the inductive case we have that
$\nfs^n((a \sleftor P^\false) \sleftand Q^\false) = (a \sleftand
\nfs^n(Q^\false)) \sleftor \nfs^n(P^\false)$, where we assume that
$\nfs^n(P^\false)$ and $\nfs^n(Q^\false)$ are $\true$-terms and that $\EqFSCL
\vdash \nfs^n(P^\false) = \neg P^\false$ and $\EqFSCL \vdash \nfs^n(Q^\false) =
\neg Q^\false$. It follows from the induction hypothesis that $\nfs^n((a
\sleftor P^\false) \sleftand Q^\false)$ is a $\true$-term. Furthermore, noting
that by the induction hypothesis we may assume that $\nfs^n(P^\false)$ and
$\nfs^n(Q^\false)$ are $\true$-terms, the proof of derivably equality is dual
to that for $\nfs^n((a \sleftand P^\true) \sleftor Q^\true)$.

To prove the lemma for $\true$-$*$-terms we first verify that the auxiliary
function $\nfs^n_1$ returns a $*$-term and that for any $*$-term $P$, $\EqFSCL
\vdash \nfs^n_1(P) = \neg P$. We show this by induction on the number of
$\ell$-terms in $P$. For the base cases it is immediate by the above cases
for $\true$-terms and $\false$-terms that $\nfs^n_1(P)$ is a $*$-term.
Furthermore, if $P$ is an $\ell$-term with a positive determinative atom we
have:
\begin{align*}
\nfs^n_1((a \sleftand P^\true) \sleftor Q^\false)
&= (\neg a \sleftand \nfs^n(Q^\false)) \sleftor \nfs^n(P^\true)
  &\textrm{by definition} \\
&= (\neg a \sleftand (\nfs^n(Q^\false) \sleftor \true)) \sleftor
  (\nfs^n(P^\true) \sleftand \false)
  &\textrm{by Lemma \ref{lem:ptpfs}} \\
&= (\neg a \sleftor (\nfs^n(P^\true) \sleftand \false)) \sleftand
  (\nfs^n(Q^\false) \sleftor \true)
  &\textrm{by Lemma \ref{lem:seqs} \eqref{eq:b3}} \\
&= (\neg a \sleftor \nfs^n(P^\true)) \sleftand \nfs^n(Q^\false)
  &\textrm{by Lemma \ref{lem:ptpfs}} \\
&= (\neg a \sleftor \neg P^\true) \sleftand \neg Q^\false
  &\textrm{by induction hypothesis} \\
&= \neg((a \sleftand P^\true) \sleftor Q^\false).
  &\textrm{by \eqref{ax:scl2} and its dual}
\end{align*}
If $P$ is an $\ell$-term with a negative determinative atom the proof proceeds
the same, substituting $\neg a$ for $a$ and applying \eqref{ax:scl3} where
needed. For the inductive step we assume that the result holds for all
$*$-terms with fewer $\ell$-terms than $P^* \sleftand Q^d$ and $P^* \sleftor
Q^c$. We note that each application of $\nfs^n_1$ changes the main connective
(not occurring inside an $\ell$-term) and hence the result is a $*$-term.
Derivable equality is, given the induction hypothesis, an instance of (the dual
of) \eqref{ax:scl2}.

With this result we can now see that $\nfs^n(P^\true \sleftand Q^*)$ is indeed
a $\true$-$*$-term. We note that, by the above, Lemma \ref{lem:ptpfs}
implies that $\neg P^\true = \neg P^\true \sleftand \false$. Now we find that:
\begin{align*}
\nfs^n(P^\true \sleftand Q^*)
&= P^\true \sleftand \nfs^n_1(Q^*)
  &\textrm{by definition} \\
&= P^\true \sleftand \neg Q^*
  &\textrm{as shown above} \\
&= (P^\true \sleftor \true) \sleftand \neg Q^*
  &\textrm{by Lemma \ref{lem:ptpfs}} \\
&= \neg((P^\true \sleftor \true) \sleftand Q^*)
  &\textrm{by Lemma \ref{lem:seqs2} \eqref{eq:b4}} \\
&= \neg(P^\true \sleftand Q^*).
  &\textrm{by Lemma \ref{lem:ptpfs}}
\end{align*}
Hence for all $P \in \SNF$, $\EqFSCL \vdash \nfs^n(P) = \neg P$.
\end{proof}

\begin{lemma}
\label{lem:nfsc1}
For any $\true$-term $P$ and $Q \in \SNF$, $\nfs^c(P, Q)$ has the same
grammatical category as $Q$ and
\begin{equation*}
\EqFSCL \vdash \nfs^c(P, Q) = P \sleftand Q.
\end{equation*}
\end{lemma}
\begin{proof}
By induction on the complexity of the $\true$-term. In the base case we see
that $\nfs^c(\true, P) = P$, which is clearly of the same grammatical category
as $P$. Derivable equality is an instance of \eqref{ax:scl4}.

For the inductive step we assume that the result holds for all $\true$-terms of
lesser complexity than $a \sleftand P^\true$. The claim about the grammatical
category follows immediately from the induction hypothesis. For the claim about
derivable equality we make a case distinction on the grammatical category of
the second argument. If the second argument is a $\true$-term, we prove
derivable equality as follows:
\begin{align*}
\nfs^c((a \sleftand P^\true) \sleftor Q^\true, &R^\true) \\
&= (a \sleftand \nfs^c(P^\true, R^\true)) \sleftor \nfs^c(Q^\true, R^\true)
  &\textrm{by definition} \\
&= (a \sleftand (P^\true \sleftand R^\true)) \sleftor (Q^\true \sleftand
  R^\true)
  &\textrm{by induction hypothesis} \\
&= (a \sleftand (P^\true \sleftand (R^\true \sleftor \true))) \sleftor (Q^\true
  \sleftand (R^\true \sleftor \true))
  &\textrm{by Lemma \ref{lem:ptpfs}} \\
&= ((a \sleftand P^\true) \sleftor Q^\true) \sleftand (R^\true \sleftor \true)
  &\textrm{by Lemma \ref{lem:seqs2} \eqref{eq:b5}} \\
&= ((a \sleftand P^\true) \sleftor Q^\true) \sleftand R^\true.
  &\textrm{by Lemma \ref{lem:ptpfs}}
\end{align*}
If the second argument is an $\false$-term, we prove derivable equality as
follows:
\begin{align*}
\nfs^c((a &\sleftand P^\true) \sleftor Q^\true, R^\false) \\
&= (a \sleftor \nfs^c(Q^\true, R^\false)) \sleftand \nfs^c(P^\true, R^\false)
  &\textrm{by definition} \\
&= (a \sleftor (Q^\true \sleftand R^\false)) \sleftand (P^\true \sleftand
  R^\false)
  &\textrm{by induction hypothesis} \\
&= (a \sleftor ((Q^\true \sleftor \true) \sleftand (R^\false \sleftand \false)))
  \sleftand ((P^\true \sleftor \true) \sleftand (R^\false \sleftand \false))
  &\textrm{by Lemma \ref{lem:ptpf}} \\
&= ((a \sleftand (P^\true \sleftor \true)) \sleftor (Q^\true \sleftor \true))
  \sleftand (R^\false \sleftand \false))
  &\textrm{by Lemma \ref{lem:seqs2} \eqref{eq:b6}} \\
&= ((a \sleftand P^\true) \sleftor Q^\true) \sleftand R^\false.
  &\textrm{by Lemma \ref{lem:ptpf}}
\end{align*}

If the second argument is $\true$-$*$-term, the result follows from the case
where the second argument is a $\true$-term and \eqref{ax:scl7}.
\end{proof}

\begin{lemma}
\label{lem:nfsc2}
For any $\false$-term $P$ and $Q \in \SNF$, $\nfs^c(P, Q)$ is a
$\false$-term and
\begin{equation*}
\EqFSCL \vdash \nfs^c(P, Q) = P \sleftand Q.
\end{equation*}
\end{lemma}
\begin{proof}
The grammatical result is immediate and the claim about derivable equality
follows from Lemma \ref{lem:ptpfs}, \eqref{ax:scl7} and \eqref{ax:scl6}.
\end{proof}

\begin{lemma}
\label{lem:nfsc3}
For any $\true$-$*$-term $P$ and $\true$-term $Q$, $\nfs^c(P, Q)$ has the same
grammatical category as $P$ and
\begin{equation*}
\EqFSCL \vdash \nfs^c(P, Q) = P \sleftand Q.
\end{equation*}
\end{lemma}
\begin{proof}
By \eqref{ax:scl7} it suffices to prove the claims for $\nfs^c_1$, i.e., that
$\nfs^c_1(P^*, Q^\true)$ is a $*$-term and that $\EqFSCL \vdash \nfs^c_1(P^*,
Q^\true) = P^* \sleftand Q^\true$. We prove this by induction on the number of
$\ell$-terms in $P^*$. In the base case we deal with $\ell$-terms and the
grammatical claim follows from Lemma \ref{lem:nfsc1}. We prove derivable
equality as follows, letting $\hat{a} \in \{a, \neg a\}$:
\begin{align*}
\nfs^c_1((\hat{a} \sleftand P^\true) \sleftor Q^\false, R^\true)
&= (\hat{a} \sleftand \nfs^c(P^\true, R^\true)) \sleftor Q^\false
  &\textrm{by definition} \\
&= (\hat{a} \sleftand (P^\true \sleftand R^\true)) \sleftor Q^\false
  &\textrm{by Lemma \ref{lem:nfsc1}} \\
&= ((\hat{a} \sleftand P^\true) \sleftand R^\true) \sleftor Q^\false
  &\textrm{by \eqref{ax:scl7}} \\
&= ((\hat{a} \sleftand P^\true) \sleftand (R^\true \sleftor \true))
  \sleftor (Q^\false \sleftand \false)
  &\textrm{by Lemma \ref{lem:ptpfs}} \\
&= ((\hat{a} \sleftand P^\true) \sleftor (Q^\false \sleftand \false))
  \sleftand (R^\true \sleftor \true)
  &\textrm{by Lemma \ref{lem:seqs} \eqref{eq:b3}} \\
&= ((\hat{a} \sleftand P^\true) \sleftor Q^\false) \sleftand R^\true.
  &\textrm{by Lemma \ref{lem:ptpfs}}
\end{align*}

For the induction step we assume that the result holds for all $*$-terms with
fewer $\ell$-terms than $P^* \sleftand Q^d$ and $P^* \sleftor Q^c$.  In the
case of conjunctions the results follow from the induction hypothesis and
\eqref{ax:scl7}. In the case of disjunctions the results follow immediately
from the induction hypothesis, Lemma \ref{lem:ptpfs} and the dual of
\eqref{ax:scl10s}.
\end{proof}

\begin{lemma}
\label{lem:nfsc4}
For any $\true$-$*$-term $P$ and $\false$-term $Q$, $\nfs^c(P, Q)$ is an
$\false$-term and
\begin{equation*}
\EqFSCL \vdash \nfs^c(P, Q) = P \sleftand Q.
\end{equation*}
\end{lemma}
\begin{proof}
By Lemma \ref{lem:nfsc1} and \eqref{ax:scl7} it suffices to prove that
$\nfs^c_2(P^*, Q^\false)$ is an $\false$-term and that $\EqFSCL \vdash
\nfs^c_2(P^*, Q^\false) = P^* \sleftand Q^\false$. We prove this by induction
on the number of $\ell$-terms in $P^*$. In the base case we deal with
$\ell$-terms and the grammatical claim follows from Lemma \ref{lem:nfsc1}. We
derive the remaining claim for $\ell$-terms with positive determinative atoms
as:
\begin{align*}
\nfs^c_2((a \sleftand P^\true) \sleftor Q^\false, R^\false)
&= (a \sleftor Q^\false) \sleftand \nfs^c(P^\true, R^\false)
  &\textrm{by definition} \\
&= (a \sleftor Q^\false) \sleftand (P^\true \sleftand R^\false)
  &\textrm{by Lemma \ref{lem:nfsc1}} \\
&= ((a \sleftor Q^\false) \sleftand P^\true) \sleftand R^\false
  &\textrm{by \eqref{ax:scl7}} \\
&= ((a \sleftor (Q^\false \sleftand \false)) \sleftand (P^\true \sleftor
  \true)) \sleftand R^\false
  &\textrm{by Lemma \ref{lem:ptpfs}} \\
&= ((a \sleftand (P^\true \sleftor \true)) \sleftor (Q^\false \sleftand
  \false)) \sleftand R^\false
  &\textrm{by Lemma \ref{lem:seqs} \eqref{eq:b3}} \\
&= ((a \sleftand P^\true) \sleftor Q^\false) \sleftand R^\false.
  &\textrm{by Lemma \ref{lem:ptpfs}}
\end{align*}
For $\ell$-terms with negative determinative atoms we derive:
\begin{align*}
\nfs^c_2((\neg a \sleftand P^\true) &\sleftor Q^\false, R^\false) \\
&= (a \sleftor \nfs^c(P^\true, R^\false)) \sleftand Q^\false
  &\textrm{by definition} \\
&= (a \sleftor (P^\true \sleftand R^\false)) \sleftand Q^\false
  &\textrm{by induction hypothesis} \\
&= (a \sleftor ((P^\true \sleftor \true) \sleftand (R^\false \sleftand
  \false))) \sleftand (Q^\false \sleftand \false)
  &\textrm{by Lemma \ref{lem:ptpfs}} \\
&= ((\neg a \sleftand (P^\true \sleftor \true)) \sleftor (Q^\false \sleftand
  \false)) \sleftand (R^\false \sleftand \false)
  &\textrm{by Lemma \ref{lem:seqs2} \eqref{eq:b7}} \\
&= ((\neg a \sleftand P^\true) \sleftor Q^\false) \sleftand R^\false.
  &\textrm{by Lemma \ref{lem:ptpfs}}
\end{align*}

For the induction step we assume that the result holds for all $*$-terms with
fewer $\ell$-terms than $P^* \sleftand Q^d$ and $P^* \sleftor Q^c$.  In the
case of conjunctions the results follow from the induction hypothesis and
\eqref{ax:scl7}. In the case of disjunctions note that by Lemma \ref{lem:nfsn}
and the proof of Lemma \ref{lem:nfsc3}, we have that $\nfs^n(\nfs^c_1(P^*,
\nfs^n(R^\false)))$ is a $*$-terms with same number of $\ell$-terms as $P^*$.
The grammatical result follows from this fact and the induction hypothesis.
Furthermore, noting that by the same argument $\nfs^n(\nfs^c_1(P^*,
\nfs^n(R^\false))) = \neg(P^* \sleftand \neg R^\false)$, we derive:
\begin{align*}
\nfs^c_2(P^* \sleftor Q^c, R^\false)
&= \nfs^c_2(\nfs^n(\nfs^c_1(P^*, \nfs^n(R^\false))), \nfs^c_2(Q^c, R^\false))
  &\textrm{by definition} \\
&= \nfs^n(\nfs^c_1(P^*, \nfs^n(R^\false))) \sleftand (Q^c \sleftand R^\false)
  &\textrm{by induction hypothesis} \\
&= \neg(P^* \sleftand \neg R^\false) \sleftand (Q^c \sleftand R^\false)
  &\textrm{as shown above} \\
&= (\neg P^* \sleftor R^\false) \sleftand (Q^c \sleftand R^\false)
  &\textrm{by \eqref{ax:scl3} and \eqref{ax:scl2}} \\
&= (\neg P^* \sleftor (R^\false \sleftand \false)) \sleftand (Q^c \sleftand
  (R^\false \sleftand \false))
  &\textrm{by Lemma \ref{lem:ptpfs}} \\
&= (P^* \sleftor Q^c) \sleftand (R^\false \sleftand \false)
  &\textrm{by Lemma \ref{lem:seqs} \eqref{eq:b1}} \\
&= (P^* \sleftor Q^c) \sleftand R^\false.
  &\textrm{by Lemma \ref{lem:ptpfs}}
\end{align*}
This completes the proof.
\end{proof}

\begin{lemma}
\label{lem:nfsc5}
For any $P, Q \in \SNF$, $\nfs^c(P, Q)$ is in $\SNF$ and
\begin{equation*}
\EqFSCL \vdash \nfs^c(P, Q) = P \sleftand Q.
\end{equation*}
\end{lemma}
\begin{proof}
By the four preceding lemmas it suffices to show that $\nfs^c(P^\true \sleftand
Q^*, R^\true \sleftand S^*)$ is in $\SNF$ and that $\EqFSCL \vdash
\nfs^c(P^\true \sleftand Q^*, R^\true \sleftand S^*) = (P^\true \sleftand Q^*)
\sleftand (R^\true \sleftand S^*)$. By \eqref{ax:scl7}, in turn, it suffices to
prove that $\nfs^c_3(P^*, Q^\true \sleftand R^*)$ is a $*$-term and that
$\EqFSCL \vdash \nfs^c_3(P^*, Q^\true \sleftand R^*) = P^* \sleftand (Q^\true
\sleftand R^*)$. We prove this by induction on the number of $\ell$-terms in
$R^*$. In the base case we have that $\nfs^c_3(P^*, Q^\true \sleftand R^\ell)
= \nfs^c_1(P^*, Q^\true) \sleftand R^\ell$. The results follow from Lemma
\ref{lem:nfsc3} and \eqref{ax:scl7}.

For conjunctions the result follows from the induction hypothesis and
\eqref{ax:scl7} and for disjunctions it follows from Lemma \ref{lem:nfsc3}
and \eqref{ax:scl7}.
\end{proof}

\thmnfscorrect*
\begin{proof}
By induction on the complexity of $P$. If $P$ is an atom, the result is by
\eqref{ax:scl4}, \eqref{ax:scl5} and its dual. If $P$ is $\true$ or $\false$
the result is by identity. For the induction we get the result by Lemma
\ref{lem:nfsn}, Lemma \ref{lem:nfsc5} and \eqref{ax:scl2}.
\end{proof}

\section{Correctness of \invstitle}
\label{sec:sclinv}

\thminvscorrect*
\begin{proof}
We first prove that for all $\true$-terms $P$, $\invs^\true(\SE(P)) \equiv P$,
by induction on $P$. In the base case $P \equiv \true$ and we have
$\invs^\true(\SE(P)) \equiv \invs^\true(\true) \equiv \true \equiv P$. For the
inductive case we have $P \equiv (a \sleftand Q^\true) \sleftor R^\true$ and
\begin{align*}
\invs^\true(\SE(P)) &\equiv \invs^\true(\SE(Q^\true) \tlef a \trig
  \SE(R^\true))
  &\textrm{by definition of $\SE$} \\
&\equiv (a \sleftand \invs^\true(\SE(Q^\true))) \sleftor
  \invs^\true(\SE(R^\true))
  &\textrm{by definition of $\invs^\true$} \\
&\equiv (a \sleftand Q^\true) \sleftor R^\true
  &\textrm{by induction hypothesis} \\
&\equiv P.
\end{align*}

Similarly we see that for all $\false$-terms $P$, $\invs^\false(\SE(P)) \equiv
P$, by induction on $P$. In the base case $P \equiv \false$ and we have
$\invs^\false(\SE(P)) \equiv \invs^\false(\false) \equiv \false \equiv P$. For
the inductive case we have $P \equiv (a \sleftor Q^\false) \sleftand R^\false$
and
\begin{align*}
\invs^\false(\SE(P)) &\equiv \invs^\false(\SE(R^\false) \tlef a \trig
  \SE(Q^\false))
  &\textrm{by definition of $\SE$} \\
&\equiv (a \sleftor \invs^\false(\SE(Q^\false))) \sleftand
  \invs^\false(\SE(R^\false))
  &\textrm{by definition of $\invs^\false$} \\
&\equiv (a \sleftor Q^\false) \sleftand R^\false
  &\textrm{by induction hypothesis} \\
&\equiv P.
\end{align*}

Now we check that for all $\ell$-terms $P$, $\invs^\ell(\SE(P)) \equiv P$.  We
observe that either $P \equiv (a \sleftand Q^\true) \sleftor R^\false$ or $P
\equiv (\neg a \sleftand Q^\true) \sleftor R^\false$. In the first case we have
\begin{align*}
\invs^\ell(\SE(P)) &\equiv \invs^\ell(\SE(Q^\true) \tlef a \trig \SE(R^\false))
  &\textrm{by definition of $\SE$} \\
&\equiv (a \sleftand \invs^\true(\SE(Q^\true))) \sleftor
  \invs^\false(\SE(R^\false))
  &\textrm{by definition of $\invs^\ell$} \\
&\equiv (a \sleftand Q^\true) \sleftor R^\false
  &\textrm{as shown above} \\
&\equiv P.
\end{align*}
In the second case we have that
\begin{align*}
\invs^\ell(\SE(P)) &= \invs^\ell(\SE(R^\false) \tlef a \trig \SE(Q^\true))
  &\textrm{by definition of $\SE$} \\
&\equiv (\neg a \sleftand \invs^\true(\SE(Q^\true))) \sleftor
  \invs^\false(\SE(R^\false))
  &\textrm{by definition of $\invs^\ell$} \\
&\equiv (\neg a \sleftand Q^\true) \sleftor R^\false
  &\textrm{as shown above} \\
&\equiv P.
\end{align*}

We now prove that for all $*$-terms $P$, $\invs^*(\SE(P)) \equiv P$, by
induction on $P$ modulo the complexity of $\ell$-terms. In the base case we are
dealing with $\ell$-terms. Because an $\ell$-term has neither a cd nor a dd we
have $\invs^*(\SE(P)) \equiv \invs^\ell(\SE(P)) \equiv P$, where the first
equality is by definition of $\invs^*$ and the second was shown above. For the
induction we have either $P \equiv Q \sleftand R$ or $P \equiv Q \sleftor R$.
In the first case note that by Theorem \ref{thm:scddd}, $\SE(P)$ has a cd and
no dd.  So we have 
\begin{align*}
\invs^*(\SE(P)) &\equiv \invs^*(\cd_1(\SE(P))\sub{\Box}{\true}) \sleftand
  \invs_*(\cd_2(\SE(P)))
  &\textrm{by definition of $\invs^*$} \\
&\equiv \invs^*(\SE(Q)) \sleftand \invs^*(\SE(R))
  &\textrm{by Theorem \ref{thm:scddd}} \\
&\equiv Q \sleftand R
  &\textrm{by induction hypothesis} \\
&\equiv P.
\end{align*}
In the second case, again by Theorem \ref{thm:scddd}, $P$ has a dd and no
cd. So we have that
\begin{align*}
\invs^*(\SE(P)) &\equiv \invs^*(\dd_1(\SE(P))\sub{\Box}{\false}) \sleftor
  \invs_*(\dd_2(\SE(P)))
  &\textrm{by definition of $\invs^*$} \\
&\equiv \invs^*(\SE(Q)) \sleftor \invs^*(\SE(R))
  &\textrm{by Theorem \ref{thm:scddd}} \\
&\equiv Q \sleftor R
  &\textrm{by induction hypothesis} \\
&\equiv P.
\end{align*}

Finally, we prove the theorem's statement by making a case distinction on the
grammatical category of $P$. If $P$ is a $\true$-term, then $\SE(P)$ has only
$\true$-leaves and hence $\invs(\SE(P)) \equiv \invs^\true(\SE(P)) \equiv P$,
where the first equality is by definition of $\invs$ and the second was shown
above. If $P$ is a $\false$-term, then $\SE(P)$ has only $\false$-leaves and
hence $\invs(\SE(P)) \equiv \invs^\false(\SE(P)) \equiv P$, where the first
equality is by definition of $\invs$ and the second was shown above. If $P$ is
a $\true$-$*$-term, then it has both $\true$ and $\false$-leaves and hence,
letting $P \equiv Q \sleftand R$,
\begin{align*}
\invs(\SE(P)) &\equiv \invs^\true(\tsd_1(\SE(P))\sub{\Box}{\true}) \sleftand
  \invs^*(\tsd_2(\SE(P)))
  &\textrm{by definition of $\invs$} \\
&\equiv \invs^\true(\SE(Q)) \sleftand \invs^*(\SE(R))
  &\textrm{by Theorem \ref{thm:stsd}} \\
&\equiv Q \sleftand R
  &\textrm{as shown above} \\
&\equiv P,
\end{align*}
which completes the proof.
\end{proof}

%
%
\cleardoublepage
\phantomsection
\addcontentsline{toc}{chapter}{Bibliography}
\bibliography{references}
\end{document}